\documentclass[10pt,a4paper, fleqn]{article}
\usepackage{etex}
\usepackage[utf8]{inputenc}
\usepackage[TS1,T1]{fontenc}
\usepackage{xspace}

\usepackage[a4paper,tmargin=3truecm,bmargin=3truecm,hmargin=1.7truecm]{geometry}

\usepackage{tikz-cd}
\usetikzlibrary{external}

\usepackage[english]{babel}

\usepackage[bitstream-charter, cal=cmcal]{mathdesign}

\usepackage{enumitem}
\newlist{enum-hypothesis}{enumerate}{1}
\setlist[enum-hypothesis]{label=(\arabic*),itemsep=0pt, parsep=0pt}

\setlist[enumerate,1]{label=\arabic*., ref=\arabic*, topsep=1pt, itemsep=2pt, parsep=0pt, leftmargin=1.5em, itemindent=0em, labelsep=0.2em, labelwidth=1.3em}
\setlist[enumerate,2]{label=\alph*., ref=\theenumi.\alph*, topsep=1pt, itemsep=2pt, parsep=0pt, leftmargin=0.5em, itemindent=0em, labelsep=0.2em, labelwidth=1.5em}
\setlist[enumerate,3]{label=\roman*., ref=\theenumii.\roman*, topsep=1pt, itemsep=2pt, parsep=0pt, leftmargin=0.5em, itemindent=0em, labelsep=0.2em, labelwidth=1.2em}

\usepackage{array}   
\newcolumntype{R}{>{\raggedleft\arraybackslash$}p{1.5em}<{$}} 
\usepackage{booktabs, dcolumn}
\usepackage{graphicx, subfig}

\usepackage{pbox}

\usepackage{amsmath,mathtools,slashed,mathdots}
\usepackage[thmmarks,amsmath]{ntheorem}
\usepackage{bm}

\newtheorem{theorem}{Theorem}[section]
\newtheorem{proposition}[theorem]{Proposition}
\newtheorem{lemma}[theorem]{Lemma}
\newtheorem{corollary}[theorem]{Corollary}

\newtheorem{definition}[theorem]{Definition}

\theoremsymbol{$\square$}
\theoremstyle{plain}
\theorembodyfont{\upshape}
\newtheorem{remark}[theorem]{Remark}
\theoremsymbol{$\Diamond$}

\theoremsymbol{}
\theoremstyle{break}
\theorembodyfont{\slshape}
\newtheorem{hypothesis}[theorem]{Hypothesis}

\theoremheaderfont{\scshape}
\theorembodyfont{\upshape} 
\theoremstyle{nonumberplain}
\theoremseparator{} 
\theoremsymbol{\rule{1.3ex}{1.3ex}} 
\newtheorem{proof}{Proof}

\DeclareSymbolFont{largesymbols}{OMX}{cmex}{m}{n}

\usepackage[square,numbers,compress,merge]{natbib}
\usepackage{hyperref}

\hypersetup{
plainpages=false,
colorlinks=true,
linkcolor=black, 
anchorcolor=black, 
citecolor=black, 
urlcolor=black, 
menucolor=black, 
filecolor=black, 
bookmarksopen=true,
bookmarksnumbered=true}

\hypersetup{
pdfauthor={Thierry Masson, Gaston Nieuviarts},
pdftitle={Lifting Bratteli Diagrams between Krajewski Diagrams: Spectral Triples, Spectral Actions, and AF algebras},
pdfsubject={},
pdfcreator={pdflatex},
pdfproducer={pdflatex},
pdfkeywords={}
}

\newcommand{\bbC}{\mathbb{C}}

\newcommand{\bbR}{\mathbb{R}}

\newcommand\bbZ{\mathbb{Z}}
\newcommand{\bbbone}{{\text{\usefont{U}{bbold}{m}{n}\char49}}}
\newcommand{\bbbzero}{{\text{\usefont{U}{bbold}{m}{n}\char48}}} \newcommand{\hbbbone}{\widehat{\bbbone}}

\newcommand{\calA}{\mathcal{A}}
\newcommand{\calB}{\mathcal{B}}

\newcommand{\calH}{\mathcal{H}}
\newcommand{\calK}{\mathcal{K}}
\newcommand{\calL}{\mathcal{L}}
\newcommand{\calM}{\mathcal{M}}
\newcommand{\calN}{\mathcal{N}}
\newcommand{\calO}{\mathcal{O}}

\newcommand{\calT}{\mathcal{T}}
\newcommand{\calU}{\mathcal{U}}

\newcommand{\kT}{\mathfrak{T}}

\newcommand{\bi}{\mathbf{i}}
\newcommand{\bn}{\mathbf{n}}
\newcommand{\bem}{{\mathbf{m}}}
\newcommand{\bOmega}{\mathbf{\Omega}}
\newcommand{\bmu}{\bm{\mu}}
\newcommand{\bomega}{\bm{\omega}}

\newcommand\halgA{\widehat{\algA}}
\newcommand\halgB{\widehat{\algB}}

\newcommand\hphi{\hat{\phi}}

\newcommand\hD{\widehat{D}}

\newcommand\hw{\hat{w}}
\newcommand{\hA}{\hat{A}}
\newcommand{\hB}{\hat{B}}
\newcommand{\hGammaB}{\widehat{\Gamma}_\algB}

\newcommand{\Bpsi}{\overline{\psi}}

\newcommand{\BK}{\overline{K}}
\newcommand{\BL}{\overline{L}}
\newcommand{\BA}{\overline{A}}
\newcommand{\Bxi}{\bar{\xi}}
\newcommand{\Beta}{\bar{\eta}}
\newcommand{\Be}{\bar{e}}

\newcommand{\Bp}{\bar{p}}
\newcommand{\Bq}{\bar{q}}
\newcommand{\BU}{\bar{U}}

\newcommand{\Tpsi}{\widetilde{\psi}}

\newcommand{\TU}{\widetilde{U}}
\newcommand{\TGamma}{\widetilde{\Gamma}}
\newcommand{\Tu}{\tilde{u}}
\newcommand{\Tv}{\tilde{v}}

\newcommand{\BTU}{\overline{\widetilde{U}}}
\newcommand{\BTu}{\bar{\tilde{u}}}

\newcommand{\smallpmatrix}[1]{\left( \begin{smallmatrix} #1 \end{smallmatrix} \right)}
\newcommand{\defeq}{\vcentcolon=} 

\DeclareMathOperator{\tsum}{\textstyle\sum}
\DeclareMathOperator{\toplus}{\textstyle\oplus}

\DeclareMathOperator{\ad}{ad}
\DeclareMathOperator{\Aut}{Aut}

\DeclareMathOperator{\Id}{Id}

\DeclareMathOperator{\Ker}{Ker}	 

\DeclareMathOperator{\Ran}{Ran}	 
\DeclareMathOperator{\Span}{Span}

\DeclareMathOperator{\tr}{tr}	   
\DeclareMathOperator{\Tr}{Tr}	   

\DeclarePairedDelimiter\norm{\lVert}{\rVert}

\newcommand{\algA}{\calA}
\newcommand{\algB}{\calB}
\newcommand{\algAp}{\calA'}
\newcommand{\algBp}{\calB'}

\newcommand{\modM}{\calM}
\newcommand{\modN}{\calN}

\newcommand{\grast}{\bullet} 

\newcommand{\vm}{\sigma} 
\newcommand{\Bvm}{\overline{\sigma}} 
\newcommand{\wm}{\tau} 
\newcommand{\Jim}{\kappa} 
\newcommand{\JimA}[1][]{\Jim_{\algA #1}} %
\newcommand{\JimB}[1][]{\Jim_{\algB #1}} %
\newcommand{\hJim}{\widehat{\Jim}} %
\newcommand{\hJimA}[1][]{\widehat{\Jim}_{\algA #1}} %
\newcommand{\spm}{\mathbf{T}} 
\newcommand{\spe}{\mathbf{t}} 
\newcommand{\hspm}{\mathbf{T}} 
\newcommand{\Mu}{u} 

\newcommand{\dd}{\text{\textup{d}}}
\newcommand{\ddU}{\text{\textup{d}}_U}

\newcommand{\bddU}{\text{\textbf{d}}_U}

\newcommand{\phiMod}[1][]{\phi_{\text{\textup{Mod}}#1}}

\newcommand{\hs}{\calH} 
\newcommand{\hsK}{\calK} 
\newcommand{\piD}{\pi_{D}}
\newcommand{\hhs}{\widehat{\hs}} 
\newcommand{\hsA}[1][]{\hs_{\algA #1}} 
\newcommand{\hsB}[1][]{\hs_{\algB #1}} 
\newcommand{\hsAp}[1][]{\hs_{\algAp #1}} 
\newcommand{\hsBp}[1][]{\hs_{\algBp #1}} 
\newcommand{\hsiA}[1][]{\hs_{\algA #1}} 
\newcommand{\hsiB}[1][]{\hs_{\algB #1}} 
\newcommand{\phiH}[1][]{\phi_{\hs #1}} 

\newcommand{\hphiH}[1][]{\hphi_{\hs #1}}

\newcommand{\DhA}[1][]{D_{\halgA #1}}
\newcommand{\DhB}[1][]{D_{\halgB #1}}
\newcommand{\piDhA}{\pi_{\DhA}}
\newcommand{\piDhB}{\pi_{\DhB}}
\newcommand{\hshA}[1][]{\hs_{\halgA #1}} 
\newcommand{\hshB}[1][]{\hs_{\halgB #1}} 

\newcommand{\ThshA}[1][]{\widetilde{\hs}_{\halgA #1}}
\newcommand{\ThshB}[1][]{\widetilde{\hs}_{\halgB #1}}

\newcommand{\phiA}{\phi_\algA}
\newcommand{\phiB}{\phi_\algB}

\newcommand{\piA}[1][]{\pi_{\algA #1}}
\newcommand{\piAp}[1][]{\pi_{\algAp #1}}
\newcommand{\piDA}{\pi_{\DA}}
\newcommand{\piB}[1][]{\pi_{\algB #1}}
\newcommand{\piBp}[1][]{\pi_{\algBp #1}}
\newcommand{\piDB}{\pi_{\DB}}

\newcommand{\bbboneA}{\bbbone_\algA}
\newcommand{\bbboneB}{\bbbone_\algB}
\newcommand{\uA}[1][]{{u_{\algA #1}}}
\newcommand{\uB}[1][]{{u_{\algB#1}}}
\newcommand{\UA}{{U_\algA}}
\newcommand{\UB}{{U_\algB}}

\newcommand{\proj}{\pi}
\newcommand{\projhs}{\proj^{\hs}}
\newcommand{\projA}{\proj^{\algA}}
\newcommand{\projB}{\proj^{\algB}}
\newcommand{\projHA}{\proj^{\hs_\algA}}
\newcommand{\projHB}{\proj^{\hs_\algB}}

\newcommand{\inj}{\iota}
\newcommand{\injhs}{\inj_{\hs}}
\newcommand{\injA}{\inj_{\algA}}
\newcommand{\injB}{\inj_{\algB}}
\newcommand{\injHA}{\inj_{\hs_\algA}}

\newcommand{\DA}[1][]{D_{\algA #1}}
\newcommand{\DB}[1][]{D_{\algB #1}}
\newcommand{\DAp}[1][]{D_{\algAp #1}}
\newcommand{\DBp}[1][]{D_{\algBp #1}}
\newcommand{\JA}[1][]{J_{\algA #1}}
\newcommand{\JB}[1][]{J_{\algB #1}}
\newcommand{\JAp}[1][]{J_{\algAp #1}}
\newcommand{\JBp}[1][]{J_{\algBp #1}}
\newcommand{\JhA}[1][]{J_{\halgA #1}}
\newcommand{\JhB}[1][]{J_{\halgB #1}}
\newcommand{\gammaA}[1][]{\gamma_{\algA #1}}
\newcommand{\gammaB}[1][]{\gamma_{\algB #1}}
\newcommand{\gammahA}[1][]{\gamma_{\halgA #1}}
\newcommand{\gammahB}[1][]{\gamma_{\halgB #1}}
\newcommand{\gammaAp}[1][]{\gamma_{\algAp #1}}
\newcommand{\gammaBp}[1][]{\gamma_{\algBp #1}}
\newcommand{\GammaA}[1][]{\Gamma_{\algA #1}}
\newcommand{\TGammaA}[1][]{\TGamma_{\algA #1}}
\newcommand{\GammaB}[1][]{\Gamma_{\algB #1}}
\newcommand{\epsilonA}{\epsilon_\algA}
\newcommand{\epsilonB}{\epsilon_\algB}
\newcommand{\epsilonhA}{\epsilon_{\halgA}}
\newcommand{\epsilonhB}{\epsilon_{\halgB}}

\newcommand{\DM}{D_M}
\newcommand{\JM}{J_M}
\newcommand{\gammaM}{\gamma_M}

\newcommand{\TNIC}{{\small TNIC}}

\newcounter{mnotecount}[section]
\renewcommand{\themnotecount}{\thesection.\arabic{mnotecount}}
\newcommand{\mnote}[1]%
{\protect{\stepcounter{mnotecount}}${}^{\text{\footnotesize$\bullet$\themnotecount}}$%
\reversemarginpar%
\marginpar{\raggedleft\footnotesize$\bullet$\themnotecount: #1}}

\newlength{\mnotewidth}
\numberwithin{equation}{section}
\allowdisplaybreaks

\begin{document}
\renewcommand\figurename{Fig.}


{
\makeatletter\def\@fnsymbol{\@arabic}\makeatother 
\title{Lifting Bratteli Diagrams between Krajewski Diagrams: \\Spectral Triples, Spectral Actions, and $AF$ algebras}
\author{T. Masson, G. Nieuviarts\\
{\small Centre de Physique Théorique}%
\\
\small{Aix Marseille Univ, Université de Toulon, CNRS, CPT, Marseille, France}\\[2ex]
}
\date{}
\maketitle
}

\begin{abstract}
In this paper, we present a framework to construct sequences of spectral triples on top of an inductive sequence defining an $AF$-algebra. One aim of this paper is to lift arrows of a Bratteli diagram to arrows between Krajewski diagrams. The spectral actions defining Non-commutative Gauge Field Theories associated to two spectral triples related by these arrows are compared (tensored by a commutative spectral triple to put us in the context of Almost Commutative manifolds). This paper is a follow up of a previous one in which this program was defined and physically illustrated in the framework of the derivation-based differential calculus, but the present paper focuses more on the mathematical structure without trying to study the physical implications.
\end{abstract}


\tableofcontents

\section{Introduction}
\label{sec introduction}

In this paper, we continue the investigation started in \cite{MassNieu21q} on a new approach to propose a framework in Non-Commutative Gauge Field Theory (NCGFT) to construct “unifying theories”. This framework relies on Approximately Finite dimensional ($AF$) $C^*$-algebras (see \cite{Blac06a, Davi96a, RordLarsLaus00a} for instance). As explained in \cite{MassNieu21q} (to which we refer for more details), the idea is to take advantage of two features of $AF$-algebras. On the one hand, it is a direct limit of finite-dimensional $C^*$-algebra, which are, up to isomorphisms, finite sum of matrix algebras over $\bbC$. So one has a way to “approximate” an infinite dimensional algebra by finite dimensional structures. On the other hand, NCGFTs have been investigated on algebras of the type $C^\infty(M) \otimes \algA$ (“Almost Commutative Manifolds”) where $\algA$ is a finite dimension algebra and $(M,g)$ is a Riemannian spin manifold  equipped with its canonical spectral triple. These NCGFTs are naturally of Yang-Mills-Higgs type, and the proposition of a reconstruction of the Standard Model of Particles Physics (in \cite{ChamConnMarc07a} for instance, see also \cite{DungSuij12w} and \cite{Suij15a} for reviews and references) shows the relevance and interest of this approach to Gauge Field Theories (GFT). 

In the past, tentatives have been proposed to extend the framework of Almost Commutative algebras in order to go beyond the Standard Model of Particles Physics, see for instance \cite{Step06e, Step07z, MarcSuij14c}. One can consider that the present work is part of this line of inquiry, taking a different route. Namely, let $\algA = \overline{\cup_{n\geq 0} \algA_n}$ be an $AF$ algebra, where $\algA_n$ are finite dimensional algebras. It is convenient to describe $\algA$ as the direct limit $\algA = \varinjlim \algA_n$ of the inductive sequence of the finite dimensional algebras $\{ (\algA_n, \phi_{n,m}) \, / \,  0 \leq n < m \}$ where $\phi_{n,m} : \algA_n \to \algA_{m}$ are one-to-one unital homomorphisms that satisfy the composition property $\phi_{m,p} \circ \phi_{n,m} = \phi_{n,p}$ for any $0 \leq n < m < p$. From these relations, one needs only to describe the homomorphisms $\phi_{n,n+1} : \algA_n \to \algA_{n+1}$. For any $n \geq 0$, let us introduce an odd (resp. even) real spectral triple $(\algA_n, \hs_n, D_n, J_n)$ (resp. $(\algA_n, \hs_n, D_n, J_n, \gamma_n)$). The purpose of the present paper is to define a good notion of compatibility inherited from the maps  $\phi_{n,n+1}$ between these spectral triples such that one can consider the sequence $\{ (\algA_n, \hs_n, D_n, J_n) \}_{n \geq 0}$ (resp. $\{ (\algA_n, \hs_n, D_n, J_n, \gamma_n) \}_{n \geq 0}$) as a finite dimensional approximation of a limiting spectral triple $(\algA, \hs, D, J)$ (resp. $(\algA, \hs, D, J, \gamma)$) on $\algA$. Thanks to the compatibility condition that is required between two successive spectral triples in this sequence, their spectral actions can be compared so that a “Limiting Non-Commutative Gauge Field Theory” on $\algA$ can be considered (at least in a formal sense) which is approximated by finite dimensional NCGFTs.

From a physical point of view, one can consider our contribution as a proposal for a general framework to elaborate NCGFTs in a GUT-like way. The usual GUT are based on a large gauge group from which, after applying successive Spontaneous Symmetry Breaking Mechanisms (SSBM), one gets a smaller gauge group corresponding to the desired phenomenology. In our approach, we consider two finite dimensional algebras $\algA$ and $\algB$, corresponding, in the usual NCGFT approach, to two gauge groups (modulo the tensor product with the canonical spectral triple of a compact Riemannian spin manifold). Let us denote by NCGFT$_\algA$ and NCGFT$_\algB$ the corresponding NCGFTs. If $\phi : \algA \to \algB$ is a one-to-one homomorphism, then NCGFT$_\algB$ has a larger gauge group than NCGFT$_\algA$, which provides a GUT-like situation, and the former may contain more degrees of freedom than the latter. In order to be able to compare these two NCGFTs, we introduce a constraint at the level of the two spectral triples in the form of a notion of “$\phi$-compatibility”. This notion of $\phi$-compatibility is proposed for generic algebras in Sect.~\ref{sec general situations}, but it reveals its true richness for $AF$-algebras, see Sect.~\ref{sec AF algebras}.

As a matter of fact, two notions of $\phi$-compatibility are proposed: a so-called $\phi$-compatibility (Def.~\ref{def phi compatibility operators}) and a so-called strong $\phi$-compatibility (Def.~\ref{def strong phi compatibility operators}), which is stronger, as its name suggests. The strong $\phi$-compatibility is more natural from a mathematical point of view, and it has indeed been used in the literature (see for instance \cite{ChriIvan06a}, \cite{Lai13a}, \cite{FlorGhor19p}). For instance, in Prop.~\ref{prop strong and not strong phi compatibility} we show how it is compatible with composition of operators and with adjointness, in Prop.~\ref{prop KO dim strong phi compatibility} we show that it constrains the $KO$-dimensions to be the same, and in Prop.~\ref{prop unitary equi triple and (st) phi comp} we show how it is compatible with unitary equivalence of real spectral triples. But, strong $\phi$-compatibility is too restrictive from a physical point of view, since, for instance, it imposes that the Dirac operator $\DB$ cannot couple inherited and new degrees of freedom at the level of $\algB$. From a physical point of view, $\phi$-compatibility looks more natural since it is based on constraints on inherited degrees of freedom only, so that, for instance, it allows the Dirac operator $\DB$ to couple inherited and new degrees of freedom.

\smallskip
As mentioned in \cite{MassNieu21q}, we are not aware of any empirical fact suggesting that such a radical new approach could be suitable for Particles Physics. Nevertheless the study of this mathematical framework reveals some relevant and compelling structures, and we feel that the forecasted phenomenological investigations will make appear nice ways to explore different kinds of unifications.

In \cite{MassNieu21q}, we investigated this framework using derivation-based noncommutative geometry, and we exhibited interesting results from the point of view of the SSBM. In the present paper, we focus on the spectral triple approach. One main result of the paper is the description of what can be called a “lifting” of arrows in a Bratteli diagram (which characterizes the given $AF$-algebra) to arrows between Krajewski diagrams which describe finite dimensional real spectral triples. Another result is the possible comparison between successive spectral actions defined by the spectral triples introduced in a compatible way in the sequence $\{ (\algA_n, \hs_n, D_n, J_n) \}_{n \geq 0}$ or $\{ (\algA_n, \hs_n, D_n, J_n, \gamma_n) \}_{n \geq 0}$. This comparison permits to get an idea of the physical content of the “unifying” NCGFT that could be formally considered in the limit.

\medskip
The paper is organized as follows. In Sect.~\ref{sec spectral triples and GFT}, we recall some main facts about NCGFTs and spectral triples. Since this is a well-known subject, we focus on the structures that will be used later in the paper, in particular the universal differential calculus. In Sect.~\ref{sec normal form finite real spectral triples}, we recall the classification of finite (real) spectral triples using Krajewski diagrams. We outline the steps of this classification in detail since some intermediate results that lead to these diagrams will be used later. In Sect.~\ref{sec one step in the sequence}, we describe how to lift arrows in a Bratteli diagram to arrows between Krajewski diagrams. This results leads to the construction of a sequence of NCGFTs on top of an $AF$-algebra. We focus mainly on the “one step structure”.  Finally, in Sect.~\ref{sec spectral actions AF AC manifold}, we show how spectral actions are related in this sequence, taking one step in this sequence as an illustration. We show, in a formal way, that the spectral action at one step in the sequence is part of the spectral action for the next step. It is out of the scope of this paper to construct realistic models. The physical implications of the corresponding NCGFT limit will not be discussed in details in this paper: only some general results related to other works will be presented.

\section{Spectral Triples and Gauge Fields Theories}
\label{sec spectral triples and GFT}

In this section, we recall some main facts about the construction of Gauge Fields Theories from Spectral Triples. We also take the opportunity to introduce notations for further developments. We refer to \cite{ConnMarc08b, Suij15a, Mass12a} for further details.

Let $(\algA, \hs, D)$ be a spectral triple and denote by $\pi : \algA \to \calB(\hs)$ the representation on the Hilbert space $\hs$. This makes $\hs$ a left $\algA$-module. We will always suppose that $\algA$ is unital, with unit $\bbbone$.  In the following, we will not need to consider the analytic axioms since we consider only finite dimensional algebras and representations.

An even spectral triple $(\algA, \hs, D, \gamma)$ is equipped with a $\bbZ_2$-grading linear map $\gamma$ on $\hs$ such that $\gamma^\dagger = \gamma$, $\gamma^2 = 1$, $\gamma D + D \gamma = 0$ ($D$ is odd), $\gamma \pi(a) = \pi(a)\gamma$ for any $a \in \algA$ ($\algA$ is even). The grading $\gamma$ induces a decomposition $\hs = \hs^{+} \oplus \hs^{-}$ according to the eigenvalues $\pm 1$ of $\gamma$. Spectral triples without such a grading are referred to as odd spectral triples.

A real spectral triple $(\algA, \hs, D, J)$ is equipped with a map $J :\hs \to \hs$ which is an anti-unitary operator: $\langle J \psi_1, J \psi_2 \rangle = \langle \psi_2, \psi_1 \rangle$ for any $\psi_1, \psi_2 \in \hs$ such that $[a, J b^* J^{-1}] = 0$ (commutant property) and $[[D, a], J b^\ast J^{-1}] = 0$ (first-order condition) for any $a, b \in \algA$. The map $\hs \times \algA \to \hs$ defined by $(\xi, a) \mapsto J a^\ast J^{-1} \xi$ defines a right module structure on $\hs$ so that $\hs$ is a $\algA$-bimodule. We denote by $a^\circ$ the element in the opposite algebra $\algA^\circ$ which corresponds to $a \in \algA$.\footnote{$\algA \simeq \algA^\circ$ as vector spaces by the formal map $\algA \ni a \mapsto a^\circ \in \algA^\circ$ and the new product in $\algA^\circ$ is $a^\circ b^\circ \defeq (b a)^\circ$.} Then, using $a^\circ \mapsto J a^\ast J^{-1}$ as a left representation of $\algA^\circ$, $\hs$ becomes a left $\algA \otimes \algA^\circ$-module. We will frequently write $a^\circ \psi =  J a^\ast J^{-1} \psi = \psi a$ for any $a \in \algA$ and $\psi \in \hs$. We define $\algA^{e} \defeq \algA \otimes \algA^\circ$. An even real spectral triple is an uplet $(\algA, \hs, D, J, \gamma)$ with $\gamma$ as before. Notice then that $\gamma a^\circ = a^\circ \gamma$ for any $a \in \algA$, and so $\gamma$ commutes with the left representation of $\algA^{e}$ on $\hs$. 

In the odd and even cases, the $KO$-dimensions $n \mod 8$ are given in Table~\ref{table KO dimensions}, where the numbers $\epsilon, \epsilon', \epsilon'' = \pm 1$ are defined by the requirements $J^2 = \epsilon$, $J D = \epsilon ' D J$, and $J \gamma = \epsilon'' \gamma J$ (in the even case). When $J^2 = -1$ and $\hs$ is finite dimensional, then its dimension is even (see \cite[Lemma~3.8]{Suij15a} for instance).

\begin{table}
\centering
\begin{tabular}{RRRRRRRRR}
\toprule
n & 0 & 1 & 2 & 3 & 4 & 5 & 6 & 7 \\
\midrule
\epsilon & 1 & 1 & -1 & -1 & -1 & -1 & 1 & 1 \\
\epsilon' & 1 & -1 & 1 & 1 & 1 & -1 & 1 & 1 \\
\epsilon'' & 1 & & -1 & & 1 & & -1 & \\
\bottomrule
\end{tabular}
\caption{$KO$-dimensions of real spectral triples.}
\label{table KO dimensions}
\end{table}

\medskip
Two spectral triples $(\algA, \hs, D)$ and $(\algA', \hs', D')$ are unitary equivalent when there exists a unitary operator $U : \hs \rightarrow \hs'$ and an algebra isomorphism $\phi : \algA \rightarrow \algA'$ such that $\pi' \circ \phi = U \pi U^{-1}$, $D' = U D U^{-1}$, $J' = U J U^{-1}$, and $\gamma' = U \gamma U^{-1}$, whenever the operators $J$, $J'$, $\gamma$ and $\gamma'$ exist. 

A symmetry of a spectral triple is a unitary equivalence between two spectral triples such that $\hs' = \hs$, $\algA' = \algA$, and $\pi' = \pi$, so that $U : \hs \to \hs$ and $\phi \in \Aut(\algA)$, \textit{i.e.} a symmetry acts only on $D$, $J$ and $\gamma$. In the following, we will only consider automorphisms $\phi$ which are $\algA$-inner, that is, there is a unitary $u \in \calU(\algA)$ such that $\phi_u(a) = u a u^\ast$. This unitary in $\algA$ defines the unitary $U = \pi(u) J \pi(u) J^{-1} : \hs \rightarrow \hs$, which can be interpreted as the conjugation with $\pi(u)$ for the bimodule structure. A straightforward computation shows that $U$ leaves $J$ and $\gamma$ invariant, and the Dirac operator $D$ is modified as $D^u = D + \pi(u) [D, \pi(u)^\ast] + \epsilon' J\left( \pi(u) [D, \pi(u)^\ast] \right) J^{-1}$. The usual way to look at this relation is to interpret the commutator with $D$ as a differential: this expression tells us that $D$ is modified by the addition of two inhomogeneous terms of the form ``$u \dd u^{-1}$''. Notice that, depending on the sign of $\epsilon'$, these two inhomogeneous terms produce a commutator or an anticommutator (from the point of view of the bimodule structure on $\hs$).

By definition, gauge transformations are inner symmetries of a spectral triple. In order to compensate for the inhomogeneous terms, we can use the same trick as in ordinary gauge field theory: add to the first order differential operator $D$ a gauge potential. To do that, we need a convenient notion of noncommutative connections. Let us consider the universal differential calculus $(\Omega^\grast_U(\algA), \ddU)$. A noncommutative connection is defined as a $1$-form $\omega = \sum_i a^{0}_i \ddU a^{1}_i \in \Omega^1_U(\algA)$ (finite sum). Elements in the vector spaces $\Omega^n_U(\algA)$ can be represented as bounded operators on $\hs$ by
\begin{equation*}
\piD \big(\tsum_i a^{0}_i \ddU a^{1}_i \cdots \ddU a^{n}_i \big) \defeq \tsum_i \pi(a^{0}_i) [ D, \pi(a^{1}_i)] \cdots [ D, \pi(a^{n}_i)].
\end{equation*}
Notice that the map $\piD$ is not a representation of the graded differential algebra $\Omega^\grast_U(\algA)$. In particular, $\ddU$ is not represented by the commutator $[D, -]$ as a differential. The representation $\piD$ can also be used to represent $n$-forms on the right module structure of $\hs$ by $\tsum_i a^{0}_i \ddU a^{1}_i \cdots \ddU a^{n}_i \mapsto J \piD\left(\tsum_i a^{0}_i \ddU a^{1}_i \cdots \ddU a^{n}_i \right) J^{-1}$. The map $\piD$ may have a non trivial kernel, this is why we will prefer to use $\omega \in \Omega^1_U(\algA)$ instead of $\piD(\omega)$ in some forthcoming constructions.

Given $D$ and $\omega \in \Omega^1_U(\algA)$, one defines the operator $D_\omega \defeq D + \piD(\omega) + \epsilon' J \piD(\omega) J^{-1}$. By a gauge transformation $u \in \calU(\algA)$, $D_\omega$ is transformed into
\begin{align*}
(D_\omega)^u 
&= D + \pi(u) \piD(\omega) \pi(u)^\ast + \pi(u) [D, \pi(u)^\ast] 
\\
&\hspace{1cm} + \epsilon' J \pi(u) \piD(\omega) \pi(u)^\ast J^{-1} + \epsilon' J \pi(u) [D, \pi(u)^\ast] J^{-1}.
\end{align*}
This relation can be written as $D_{\omega^u}$, where $\omega^u \in \Omega^1_U(\algA)$ is a gauge transformation of $\omega$ defined as $\omega^u \defeq u \omega u^\ast + u \ddU u^\ast$.

\medskip
In the following, we will need a convenient presentation of the differential graded algebra $(\Omega^\grast_U(\algA), \ddU)$. We follow the presentation in \cite{Mass95a}.\footnote{We owe this presentation to Michel~Dubois-Violette.} For any $n\geq 0$, let $\calT^n \algA \defeq \algA^{\otimes^{n+1}}$ and let $\calT^\grast \algA = \toplus_{n\geq 0} \calT^n \algA$. This is a graded algebra for the product $\calT^n \algA \otimes \calT^{n'} \algA \to \calT^{n+n'} \algA$ defined by $(a^0 \otimes \cdots \otimes a^n) (a'^0 \otimes \cdots \otimes a'^{n'}) \defeq a^0 \otimes \cdots \otimes a^n a'^0 \otimes \cdots \otimes a'^{n'}$. In particular, $\calT^\grast \algA$ is a bimodule over $\algA = \calT^0 \algA$. Define $\ddU : \calT^n \algA \to \calT^{n+1} \algA$ as 
\begin{align*}
\ddU (a^0 \otimes \cdots \otimes a^n) =
& \bbbone \otimes a^0 \otimes \cdots \otimes a^n
\\
& + \tsum_{p=1}^{n} (-1)^p a^0 \otimes \cdots \otimes a^{p-1} \otimes \bbbone \otimes a^{p} \otimes \cdots \otimes a^n
\\
& + (-1)^{n+1} a^0 \otimes \cdots \otimes a^n \otimes \bbbone.
\end{align*}
Then $\ddU$ is a derivation of degree $1$ on the graded algebra $\calT^\grast \algA$ such that $\ddU^2 = 0$. Notice that $\ddU(a) = \bbbone \otimes a - a \otimes \bbbone$ on $\calT^0 \algA$. It is convenient to introduce the maps $i^p_{\bbbone}(a^0 \otimes \cdots \otimes a^n) \defeq a^0 \otimes \cdots \otimes a^{p-1} \otimes \bbbone \otimes a^p \otimes \cdots \otimes a^n$ for any $p = 0, \dots, n+1$, with the convention that for $p=0$, the tensor factor $\bbbone$ is added before $a^0$ (for $p=n+1$, it is added after $a^n$). Then $\ddU = \tsum_{p=0}^{n+1} (-1)^p \,  i^p_{\bbbone} : \calT^n \algA \to \calT^{n+1} \algA$.

Let $\mu : \calT^1 \algA \to \calT^0 \algA$ be the multiplication map $a^0 \otimes a^1 \mapsto a^0 a^1$, and define $\Omega^1_U(\algA) \defeq \ker \mu \subset \calT^1 \algA$. Then $\ddU$ maps $\calT^0 \algA$ into $\Omega^1_U(\algA)$ and $\Omega^1_U(\algA)$ is generated, as a bimodule on $\algA$, by the $\ddU a$'s for $a \in \algA$.\footnote{If $\tsum_{i} a_i^0 \otimes a_i^1 \in \calT^1 \algA$ is such that $\mu(\tsum_{i} a_i^0 \otimes a_i^1) = \tsum_{i} a_i^0 a_i^1 = 0$, then $\tsum_{i} a_i^0 \otimes a_i^1 = \tsum_{i} a_i^0 (\bbbone \otimes a_i^1 -  a_i^1 \otimes \bbbone) = \tsum_{i} a_i^0 \ddU a_i^1$.} Let $\Omega^0_U(\algA) \defeq \algA$ and $\Omega^n_U(\algA) \defeq \Omega^1_U(\algA) \otimes_{\algA} \cdots \otimes_{\algA} \Omega^1_U(\algA)$ ($n$ times tensor product over $\algA$) for any $n \geq 2$ and $\Omega^\grast_U(\algA) \defeq \toplus_{n \geq 0} \Omega^n_U(\algA)$. Equivalently, $\Omega^\grast_U(\algA)$ is the graded sub-algebra of $\calT^\grast \algA$ generated by $\Omega^0_U(\algA)$ and $\Omega^1_U(\algA)$. One can then check that $\Omega^n_U(\algA) \subset \calT^n \algA$ is generated by the $a^0 \ddU a^1 \cdots \ddU a^n$ for $a^0, \dots, a^n \in \algA$, so that $\ddU$ restricts to maps $\Omega^n_U(\algA) \to \Omega^{n+1}_U(\algA)$, and then $(\Omega^\grast_U(\algA), \ddU)$ is a graded differential sub-algebra of $(\calT^\grast \algA, \ddU)$.

Let us consider the case $\algA = \toplus_{i=1}^{r} \algA_i$, where $\algA_i$ are unital algebras with units $\bbbone_{\algA_i}$. It will be useful in later discussions to use explicit presentations of $(\calT^\grast \algA, \ddU)$ and $(\Omega^\grast_U(\algA), \ddU)$ constructed as follows. Let 
\begin{align*}
\kT^0\algA \defeq \left\{
\smallpmatrix{ a_1 & 0 & \cdots & 0 \\
0 & a_2 & \cdots & 0 \\
\vdots & & \ddots &  \\
0 & 0 & \cdots & a_r}
\mid 
a = \toplus_{i=1}^{r} a_i \in \algA \right\}.
\end{align*}
For any $n \geq 1$ and any $1 \leq i_0, \dots, i_n \leq r$, let us introduce the notation $\algA^\otimes_{i_0, \dots, i_{n}} \defeq \algA_{i_0} \otimes \cdots \otimes \algA_{i_n}$. Now, let $\kT^n_{i_1, \dots, i_{n-1}}\algA$ be the set of matrices with entries in $\algA^\otimes_{i, i_1, \dots, i_{n-1}, j}$ at row $i$ and column $j$. This can be schematically visualized as
\begin{align*}
\begin{pmatrix} 
	\algA^\otimes_{1, i_1, \dots, i_{n-1}, 1} 
	& \algA^\otimes_{1, i_1, \dots, i_{n-1}, 2} 
	& \cdots 
	& \algA^\otimes_{1, i_1, \dots, i_{n-1}, r}
\\
	\algA^\otimes_{2, i_1, \dots, i_{n-1}, 1} 
	& \algA^\otimes_{2, i_1, \dots, i_{n-1}, 2} 
	& \cdots 
	& \algA^\otimes_{2, i_1, \dots, i_{n-1}, r}
\\
	\vdots & \vdots &  & \vdots 
\\
	\algA^\otimes_{r, i_1, \dots, i_{n-1}, 1} 
	& \algA^\otimes_{r, i_1, \dots, i_{n-1}, 2} 
	& \cdots 
	& \algA^\otimes_{r, i_1, \dots, i_{n-1}, r}
\end{pmatrix}
\end{align*}
where the first and last algebras in the tensor products will play a crucial role in the following. Combining the products 
\begin{align*}
\algA^\otimes_{i, i_1, \dots, i_{n-1}, k}
\otimes 
\algA^\otimes_{k, j_1, \dots, j_{n'-1}, j}  
\to 
\algA^\otimes_{i, i_1, \dots, i_{n-1}, k, j_1, \dots, j_{n'-1}, j},
\end{align*} 
defined by the product in $\algA_{k}$, and the usual rules for matrix multiplications, one gets products
\begin{align*}
\kT^n_{i_1, \dots, i_{n-1}}\algA \otimes \kT^{n'}_{j_1, \dots, j_{n'-1}}\algA
\to \toplus_{k=1}^{r} \kT^{n+n'}_{i_1, \dots, i_{n-1}, k, j_1, \dots, j_{n'-1}}\algA .
\end{align*} 
Let us introduce
\begin{align*}
\kT^n \algA \defeq \toplus_{i_1, \dots, i_{n-1} = 1}^{r} \kT^n_{i_1, \dots, i_{n-1}}\algA
\quad \text{and} \quad 
\kT^\grast \algA \defeq \toplus_{n \geq 0} \kT^n \algA
\end{align*}
then $\kT^\grast \algA$ is a graded algebra for the global product induced by the products defined above. Explicitly, for $\toplus_{i_1, \dots, i_{n-1} =1}^{r} \big( a^0_{i_0} \otimes a^1_{i_1} \otimes \cdots \otimes a^{n-1}_{i_{n-1}} \otimes a^n_{i_n} \big)_{i_0, i_n = 1}^{r} \in \kT^n \algA$ and $\toplus_{j_1, \dots, j_{n'-1} =1}^{r} \big( b^0_{j_0} \otimes b^1_{j_1} \otimes \cdots \otimes b^{n'-1}_{j_{n'-1}} \otimes b^{n'}_{j_{n'}} \big)_{j_0, j_{n'} = 1}^{r} \in \kT^{n'} \algA$, their product in $\kT^{n+n'} \algA$ is
\begin{align}
\label{eq product kT n np}
\toplus_{\substack{i_1, \dots, i_{n-1}, i_n,\\ j_1, \dots, j_{n'-1} =1}}^{r}
\big(
a^0_{i} \otimes a^1_{i_1} \otimes \cdots \otimes a^{n-1}_{i_{n-1}} \otimes a^n_{i_n} 
b^0_{i_n} \otimes b^1_{j_1} \otimes \cdots \otimes b^{n'-1}_{j_{n'-1}} \otimes b^{n'}_{j}
\big)_{i, j=1}^{r}
\end{align}

Let $\bmu$ be the component-wise product on $\kT^1 \algA$. Since multiplications by elements in $\algA_i$ and $\algA_j$ are zero for $i \neq j$, the resulting matrix is diagonal, and so one gets a natural map $\bmu : \kT^1 \algA \to \kT^0 \algA$. Let $\bOmega^1_U(\algA) \defeq \ker \bmu \subset \kT^1 \algA$ and $\bOmega^\grast_U(\algA) \subset \kT^\grast \algA$ be the graded sub-algebra generated by $\bOmega^0_U(\algA) \defeq \kT^0\algA$ and $\bOmega^1_U(\algA)$. For any $p = 0, \dots, n+1$, define $\bi^p_{\bbbone} : \kT^n_{i_1, \dots, i_{n-1}}\algA \to \toplus_{k=0}^{r} \kT^{n+1}_{i_1, \dots, i_{p-1}, k, i_{p}, \dots, i_{n-1}}\algA$ by inserting $\bbbone = \toplus_{k=1}^{r} \bbbone_{\algA_k}$ component-wise, \textit{i.e.} $\bi^p_{\bbbone} = \toplus_{k=1}^{r} \bi^p_{\bbbone_{\algA_k}}$ with obvious notations. Then one can define $\bddU \defeq \tsum_{p=0}^{n+1} (-1)^p \,  \bi^p_{\bbbone} : \kT^n \algA  \to \kT^{n+1} \algA$.

\begin{proposition}
The map $\bddU$ is a differential on $\kT^\grast \algA$ and there is an isomorphism $t : \calT^\grast \algA \to \kT^\grast \algA$ of graded differential algebras which induces an isomorphism of the graded differential (sub)algebras $\Omega^\grast_U(\algA)$ and $\bOmega^\grast_U(\algA)$.
\end{proposition}

\begin{proof}
For $n=0$, one defines $t(\toplus_{i=1}^{r} a_i) = \smallpmatrix{ a_1 & \cdots & 0 \\
 & \ddots &  \\
0 & \cdots & a_r} \in \kT^0\algA$ for any $\toplus_{i=1}^{r} a_i \in \algA$. For $n \geq 1$, consider any $a^{0} \otimes \cdots \otimes a^{n} \in \calT^n \algA$ with $a^{p} = \toplus_{i=1}^{r} a^{p}_i$ where $a^{p}_i \in \algA_i$. Expanding the tensor products along these direct sums, one gets a sum of terms of the form $a^{0}_{i_0} \otimes \cdots \otimes a^{n}_{i_n} \in \algA^\otimes_{i_0, \dots, i_{n}}$ that we assemble as elements in $\kT^n_{i_1, \dots, i_{n-1}}\algA$. This defines the map $t : \calT^n \algA \to \kT^n \algA$, which, for any $n \geq 0$, is by construction an isomorphism of vector spaces. A straightforward computation shows that the product on $\kT^\grast \algA$ is such that $t$ is a homomorphism of graded algebras.

By construction of $\bi^p_{\bbbone}$, one has $t \circ i^p_\bbbone = \bi^p_{\bbbone} \circ t$, so that $\bddU$ is a differential on $\kT^\grast \algA$ and $t$ is an isomorphism of differential algebras.

Finally, the map $\bmu$ has been defined such that $t \circ \mu = \bmu \circ t$ so that $t$ identifies $\Omega^1_U(\algA)$ with $\bOmega^1_U(\algA)$, and so $\Omega^\grast_U(\algA)$ with $\bOmega^\grast_U(\algA)$.
\end{proof}

Notice that, with $\hbbbone \defeq t(\bbbone \otimes \bbbone) \in \kT^1 \algA$, one has $\bddU t(a) = [ \hbbbone, t(a)]$ for any $a \in \algA$.

\medskip
We now suppose that there is an orthogonal decomposition of the Hilbert space $\hs = \toplus_{i=1}^{r} \hs_i$ such that the representation decomposes along $\pi = \toplus_{i=1}^{r} \pi_i$ where $\pi_i$ is a representation of $\algA_i$ on $\hs_i$: for any $\psi = \toplus_{i=1}^{r} \psi_i \in \hs$ and $a = \toplus_{i=1}^{r} a_i \in\algA$, $\pi(a) \psi = \toplus_{i=1}^{r} \pi_i(a_i) \psi_i$. Then the Dirac operator $D$ decomposes as a $r \times r$ matrix of operators $D_{j}^{i} : \hs_i \to \hs_j$. 

We propose to write the representation $\piD$ as follows. Consider any $\bomega \in \bOmega^n_U(\algA) \subset \kT^n \algA$ which decomposes along a sum of typical terms $\toplus_{i_1, \dots, i_{n-1} =1}^{r} \big( a^0_{i} \otimes a^1_{i_1} \otimes \cdots \otimes a^{n-1}_{i_{n-1}} \otimes a^n_{j} \big)_{i, j = 1}^{r} \in \kT^n \algA$. Then $\piD(\bomega)$ is the $r \times r$ matrix of operators 
\begin{align}
\label{eq piD sum algebras}
\piD(\bomega)_i^j = 
\!\!\!\!\!\sum_{ \substack{ \text{all terms at the}\\ \text{$(i,j)$ entry in $\bomega$} } } \!\!\!\!\!
\tsum_{i_1, \dots, i_{n-1} = 1}^{r} 
a^{0}_{i} D_{i}^{i_1} a^{1}_{i_1}  D_{i_1}^{i_2} \cdots D_{i_{n-2}}^{i_{n-1}} a^{n-1}_{i_{n-1}} D_{i_{n-1}}^{j} a^{n}_{j} : 
\hs_j \to \hs_i
\end{align}
Notice that, since $\bomega \in \bOmega^n_U(\algA)$, these sums define bounded operators because only commutators $[D, a]$ could appear in $\piD(\bomega)$ (this is not necessarily the case for a generic element in $\kT^n \algA$).

\section{Normal Forms of Finite Real Spectral Triples}
\label{sec normal form finite real spectral triples}

In this section we recall all the important facts about finite real spectral triples that will be needed later. In particular their classification by Krajewski diagrams \cite{Kraj98e} (see also \cite{Suij15a}, in which a sketch of this classification is given). All the missing proofs of the results presented below are given in \cite{Nieu22a} with the same notations.

\smallskip
Many results rely on the following  well-known technical result, which results from the existence of cyclic vectors in $\bbC^n$ for the matrix multiplication:
\begin{lemma}
\label{lemma technical module endo reduction}
For any $n \geq 1$ and any vector space $V$, a linear map $\Psi : \bbC^n \otimes V \to \bbC^n \otimes V$ such that $\Psi(a \xi \otimes v) = a \Psi(\xi \otimes v)$ for any $a \in M_n(\bbC)$, $\xi \in \bbC^n$ and $v \in V$, reduces to a linear map $\varphi : V \to V$ such that $\Psi(\xi \otimes v) = \xi \otimes \varphi(v)$.
\end{lemma}

\subsection{Finite Spectral Triples}
\label{subsec finite spectral triples}

A spectral triple $(\algA, \hs, D)$ is said to be finite if $\algA$ is a finite dimensional involutive $\bbC$-algebra and $\hs$ is a finite dimensional Hilbert space on which $\algA$ is represented. 
The faithful representation $\pi$ of $\algA$ on $\hs$ will be omitted when no confusion is possible. By the Wedderburn Theorem, the algebra is of the form $\algA = \toplus_{i=1}^{r} M_{n_i}(\bbC)$. In the following, we will write $\algA_i = M_{n_i} = M_{n_i}(\bbC)$ since no other matrix algebras will be considered. Let $\inj^i : \algA_i \to \algA$  be the canonical inclusion and $\proj_i : \algA \to \algA_i$ be the canonical projection.

Consider the set $\Lambda \defeq \{ \bn_1, \dots, \bn_r \}$ of irreducible representations (irreps) of $\algA$, where $\bn_i$ is a short notation that designates at the same time the integer $n_i$ defining the irrep (on the space $\bbC^{n_i}$) and the integer $i$ (the same that appears in $\algA = \toplus_{i=1}^{r} M_{n_i}$). $\Lambda$ is completely defined by $\algA$ and, reciprocally, $\algA = \toplus_{i=1}^{r} M_{n_i}$ can be recovered from $\Lambda$.  Denote by $\hs_{\bn_i} \defeq \bbC^{n_i}$ the irreducible representations (irreps) of the $\algA_i$'s, and so of $\algA$. 

The Hilbert space $\hs$ can be decomposed into orthogonal components $\hhs_{\bn_i} \defeq \inj^i(\algA_i) \hs$, so that $\hs = \toplus_{i=1}^{r} \hhs_{\bn_i}$. Define $\injhs^i : \hhs_{\bn_i} \to \hs$ and $\projhs_i : \hs \to \hhs_{\bn_i}$ the natural inclusions and (orthogonal) projections. Then there are integers $\mu_i$, the multiplicities of the irreps, such that $\hhs_{\bn_i} \simeq \hs_{\bn_i} \otimes \bbC^{\mu_i} = \bbC^{n_i} \otimes \bbC^{\mu_i}$. So, up to unitary equivalence, the Hilbert space $\hs$ can be decomposed as $\hs \simeq \toplus_{i=1}^{r} \bbC^{n_i} \otimes \bbC^{\mu_i}$ and we now suppose that a unitary map has been chosen such that $\hhs_{\bn_i} = \bbC^{n_i} \otimes \bbC^{\mu_i}$.\footnote{For sake of completeness, let us mention that the scalar product of this decomposition is the usual one:  $\langle \psi, \psi' \rangle_{\hs} = \tsum_{i=1}^{r} \langle \xi_i, \xi'_i \rangle_{\bbC^{n_i}} \langle \vm_i, \vm'_i \rangle_{\bbC^{\mu_i}}$ for any $\psi = \toplus_{i=1}^{r} \xi_i \otimes \vm_i$ (and the same for $\psi'$) where $\xi_i \in \bbC^{n_i}$ and $\vm_i \in \bbC^{\mu_i}$.} If one requires a faithful representation of $\algA$, then $\mu_i \geq 1$ for all $i$.

In the even case, one has:
\begin{lemma}
\label{lemma gamma non real}
$\gamma$ decomposes along a family of linear maps $\ell_{i} : \bbC^{\mu_{i}} \to  \bbC^{\mu_{i}}$ such that $\gamma(\xi_i \otimes \vm_i) = \xi_i \otimes \ell_i(\vm_i)$ for any $\xi_i \otimes \vm_i \in \bbC^{n_{i}} \otimes \bbC^{\mu_{i}}$. This family satisfies $\ell_{i}^\dagger = \ell_{i}$ and $\ell_{i}^2 = 1$.
\end{lemma}

Let us consider any orthonormal basis $\{ \vm_{i}^{p} \}_{1 \leq p \leq \mu_{i}}$ of $\bbC^{\mu_{i}}$. Then, for any $1 \leq i \leq r$, let $\Gamma^{(0)}_{\bn_i} \defeq \{ (i, p) \mid 1 \leq p \leq \mu_{i} \}$, and for any $v = (i,p) \in \Gamma^{(0)}_{\bn_i}$, define $\lambda : \Gamma^{(0)}_{\bn_i} \to \Lambda$ as $\lambda(v) \defeq \bn_i$. Notice that $\mu_{i} = \# \Gamma^{(0)}_{\bn_i}$. For any $v \in \Gamma^{(0)}_{\bn_i}$, we then define
\begin{align*}
\hs_{v} \defeq \Span \{ \xi_i \otimes \vm_{i}^{p} \mid \xi_i \in \bbC^{n_i} \} \simeq \hs_{\bn_i}
\end{align*}
In the even case, we require the basis $\{ \vm_{i}^{p} \}_{1 \leq p \leq \mu_{i}}$ to be eigenvectors of $\ell_{i}$ with eigenvalues $s_{i}^{p} = \pm 1$. Then $\gamma$ restricts to the multiplication by $s_{i}^{p}$ on $\hs_{v}$ with $v = (i,p)$. We define $s (v) = s_{i}^{p}$ for any $v$.

The map $\lambda$ is extended in an obvious way on the set $\Gamma^{(0)} \defeq \cup_{i=1}^{r} \Gamma^{(0)}_{\bn_i}$ and there is an orthogonal decomposition of $\hs$ into irreps $\hs = \toplus_{v \in \Gamma^{(0)}} \hs_{v}$. Let $e = (v_1, v_2) \in \Gamma^{(0)} \times \Gamma^{(0)}$, then the Dirac operator decomposes along maps $D_{e} : \hs_{v_1} \to \hs_{v_2}$. With $\Be \defeq (v_2, v_1)$, $D^\dagger = D$ is equivalent to $D_{\Be} = D_{e}^\dagger$. In the even case, $\gamma D = - D \gamma$ implies that $s(v_2) D_e = - s(v_1) D_e$, so that $D_{e}$ is non-zero only when $s(v_2) = - s(v_1)$.

\bigskip
The previous decomposition of the spectral triples $(\algA, \hs, D)$ or $(\algA, \hs, D, \gamma)$ can be summarized using a decorated graph $\Gamma$, a so-called Krajewski Diagram, together with $\Lambda$:
\begin{enumerate}
\item The set of vertex $\Gamma^{(0)}$ of the graph is equipped with a map $\lambda : \Gamma^{(0)} \to \Lambda$. By a slight abuse of notation, the map $\lambda$ will sometimes be used in the compact notation $\bbC^{\lambda(v)} = \bbC^{n_i}$. We will also use the map $i(v) \defeq i$ for $\lambda(v) = \bn_i$.

\item For any vertex $v \in \Gamma^{(0)}$, define $\hs_v \defeq \hs_{\lambda(v)} = \bbC^{\lambda(v)}$. The element $\lambda(v) \in \Lambda$ is a decoration of the vertex $v$.

\item For any $\bn_i \in \Lambda$, define $\Gamma^{(0)}_{\bn_i} \defeq \{ v \in \Gamma^{(0)} \mid \lambda(v) = \bn_i \} = \lambda^{-1}(\bn_i)$ and $\mu_{i} \defeq \# \Gamma^{(0)}_{\bn_i}$.

\item In the even case, a second decoration is the assignment of a grading map $s(v) = \pm 1$.

\item For every $e = (v_1, v_2) \in \Gamma^{(0)} \times \Gamma^{(0)}$, let $\Be \defeq (v_2, v_1)$.

\item The space $\Gamma^{(1)} \subset \Gamma^{(0)} \times \Gamma^{(0)}$ of edges of the graph are couples $e = (v_1, v_2)$ such that:
\begin{enumerate}
\item  there is a non-zero linear map $D_e : \hs_{v_1} \to \hs_{v_2}$ such that $D_{\Be} = D_e^\dagger : \hs_{v_2} \to \hs_{v_1}$. 
\item $s(v_2) = - s(v_1)$ in the even case;
\end{enumerate}
Then $D_e$ defines a decoration of $e$.
\end{enumerate}

\medskip
Given such a Krajewski Diagram, one can construct a spectral triple up to unitary equivalence.

\subsection{Finite Real Spectral Triples}

Let us now consider (odd) finite (resp. even) real spectral triples $(\algA, \hs, D, J)$ (resp. $(\algA, \hs, D, J, \gamma)$). The Hilbert space $\hs$ is then a bimodule over $\algA = \toplus_{i=1}^{r} M_{n_i}$, or equivalently a left $\algA^{e}$-module, with $\algA^{e} = \toplus_{{i}, {j} = 1}^{r} M_{n_{i}} \otimes M_{n_{j}}^\circ$. This implies that $\hs$ decomposes into orthogonal components $\hhs_{\bn_i \bn_j} \defeq \inj^{i}(\algA_{i}) \inj^{j}(\algA_{j})^\circ \hs$, so that $\hs = \toplus_{{i}, {j}=1}^{r} \hhs_{\bn_i \bn_j}$.
 
 Denote by ${\bbC^m}^\top$ (${}^\top$ for transpose) the $m$-dimensional $\bbC$-vector space of row vectors, which is a natural right $M_m$-module, and denote by $\bbC^{m \circ}$ its corresponding left $M_m^\circ$-module.\footnote{$\bbC^m \simeq \bbC^{m \circ}$ as column vectors by the formal map $\bbC^m \ni \xi \mapsto \xi^\circ \in \bbC^{m \circ}$ and, for any $a \in M_m$ and $\xi \in \bbC^m$, $a^\circ \xi^\circ \defeq (\xi^\top a)^\top$. \label{footnote circ action}} Let us recall the following result:
\begin{lemma}
For any integers $n, m \geq 1$, the irreducible left $M_n \otimes M_m^\circ$-representations are isomorphic to $\bbC^n \otimes \bbC^{m \circ}$.
\end{lemma}

Let $\mu_{i j}$ be the multiplicity of the irrep $\hs_{\bn_i \bn_j} \defeq \bbC^{n_{i}} \otimes \bbC^{n_{j} \circ}$ of $M_{n_{i}} \otimes M_{n_{j}}^\circ$ and so of $\algA^{e}$, in $\hs$. Then one has $\hhs_{\bn_i \bn_j} \simeq \hs_{\bn_i \bn_j} \otimes \bbC^{\mu_{i j}}  \simeq \bbC^{n_{i}} \otimes \bbC^{\mu_{i j}} \otimes \bbC^{n_{j} \circ}$, so that $\hs \simeq \toplus_{{i}, {j} = 1}^{r} \bbC^{n_{i}} \otimes \bbC^{\mu_{i j}} \otimes \bbC^{n_{j} \circ}$. In the following, we suppose that a unitary map has been chosen such that $\hhs_{\bn_i \bn_j} = \bbC^{n_{i}} \otimes \bbC^{\mu_{i j}} \otimes \bbC^{n_{j} \circ}$. \footnote{The factor $\bbC^{\mu_{i j}}$ has been positioned in the middle to put forward the bimodule structure. In the proof of Prop.~\ref{prop diagonalization spm for KO dims} it will be convenient to change this convention.}

Denote by $J_0$ the anti-unitary operator on $\bbC^n \otimes \bbC^\mu \otimes {\bbC^m}^\circ$ defined by $\xi \otimes \vm \otimes \eta^\circ \mapsto \Bxi \otimes \Bvm \otimes \Beta^\circ$ where $\Bxi$ is the entrywise complex conjugated vector (the same for $\Bvm$ and $\Beta^\circ$). Then $J_0$ extends naturally to $\hs$ as an anti-unitary operator which preserves each summand $\hhs_{\bn_i \bn_j}$ and one has $J_0^{-1} = J_0$.\footnote{Notice that $J_0$ depends on the canonical basis for the vector spaces $\bbC^n$, $\bbC^\mu$ and $\bbC^m$. But any fixed orthonormal basis could have been used.} We will use the natural notation $J_0(\psi) = \Bpsi$ for any $\psi \in \hs$. Define $K \defeq J J_0$, \textit{i.e.} $J = K J_0$. For any $a = \toplus_{i=1}^{r} a_i \in \algA$, define $a^\top = \toplus_{i=1}^{r} a_i^\top$ where $a_i^\top = J_0 a_i^\ast J_0$ is the transpose of $a_i \in M_{n_i}$. For any operator $A$ on $\hs$, define $\BA \defeq J_0 A J_0$ (if $A$ is written as a matrix, $\BA$ is the entrywise complex conjugate matrix, whence the notation). The following result gives an explicit description of $J$ that will be used in Sect.~\ref{sec one step in the sequence}:

\begin{proposition}
\label{prop J K L}
$K$ is a unitary operator on $\hs$ such that $K \BK = \BK K = \epsilon$. For any $1 \leq {i}, {j} \leq r$, $K(\hhs_{\bn_i \bn_j}) = \hhs_{\bn_j \bn_i}$, so that $\hhs_{\bn_i \bn_j}$ and $\hhs_{\bn_j \bn_i}$ have the same dimension, \textit{i.e.} they correspond to the same multiplicity $\mu_{i j} = \mu_{{j} {i}}$.

There is a linear map $L_{i j} : \bbC^{\mu_{i j}} \to \bbC^{\mu_{{j} {i}}}$ satisfying $L_{i j}^\dagger = L_{i j}^{-1}$ and $L_{{j} {i}} \BL_{i j} = \BL_{{j} {i}} L_{i j} = \epsilon$, such that, for any $\xi_{i} \otimes \vm_{i j} \otimes \eta_{j}^\circ \in \hhs_{\bn_i \bn_j}$, $K(\xi_{i} \otimes \vm_{i j} \otimes \eta_{j}^\circ) = \eta_{j} \otimes L_{i j}(\vm_{i j}) \otimes \xi_{i}^\circ$.

For any $\xi_{i} \otimes \vm_{i j} \otimes \eta_{j}^\circ \in \hs_{\bn_i \bn_j}$, one has $J(\xi_{i} \otimes \vm_{i j} \otimes \eta_{j}^\circ) = \Beta_{j} \otimes L_{i j}(\Bvm_{i j}) \otimes \Bxi_{i}^\circ$.
\end{proposition}

In the even case, the following result gives an explicit description of $\gamma$:
\begin{proposition}
\label{prop gamma to ell}
In the even case, there is a family of linear maps $\ell_{i j} : \bbC^{\mu_{i j}} \to \bbC^{\mu_{i j}}$ such that $\gamma(\xi_{i} \otimes \vm_{i j} \otimes \eta_{j}^\circ) = \xi_{i} \otimes \ell_{i j}(\vm_{i j}) \otimes \eta_{j}^\circ$ for any $\xi_{i} \otimes \vm_{i j} \otimes \eta_{j}^\circ \in \hhs_{\bn_i \bn_j}$. This family satisfies $\ell_{i j}^\dagger = \ell_{i j}$ and $\ell_{i j}^2 = 1$.
\end{proposition}

Let us now describe, in the following two propositions, the key constructions which lead to the classification of finite real spectral triples. The content of these two propositions will be useful in Sect.~\ref{sec AF algebras}.

\begin{proposition}
\label{prop basis odd case}
Consider the odd case situation. 

For $1 \leq i \neq j  \leq r$, there is an orthonormal basis $\{ \vm_{i j}^{p} \}_{1 \leq p \leq \mu_{i j}}$ of $\bbC^{\mu_{i j}}$ such that $\vm_{j i}^{p} = L_{i j}(\Bvm_{i j}^{p})$ and $\vm_{i j}^{p} = \epsilon\, L_{j i}(\Bvm_{j i}^{p})$ for any $i<j$ and any $1 \leq p \leq \mu_{j i} = \mu_{i j}$. 

For $i = j$ and $\epsilon = 1$ ($KO$-dimensions $1$ and $7$), there is an orthonormal basis $\{ \vm_{ii}^{p} \}_{1 \leq p \leq \mu_{ii}}$ of $\bbC^{\mu_{ii}}$ such that $\vm_{ii}^{p} = L_{ii}(\Bvm_{ii}^{p})$.

For $i = j$ and $\epsilon = -1$ ($KO$-dimensions $3$ and $5$), $\mu_{ii}$ is even and there is an orthonormal basis $\{ \vm_{ii}^{p} \}_{1 \leq p \leq \mu_{ii}}$ of $\bbC^{\mu_{ii}}$ such that $\vm_{ii}^{2a} = L_{ii}(\Bvm_{ii}^{2a-1})$ and $\vm_{ii}^{2a-1} = \epsilon\, L_{ii}(\Bvm_{ii}^{2a})$ for any $a = 1, \dots, \mu_{ii}/2$.
\end{proposition}

\begin{proposition}
\label{prop basis even case}
Consider the even case situation. 

For $1 \leq i \neq j \leq r$, there is an orthonormal basis $\{ \vm_{i j}^{p} \}_{1 \leq p \leq \mu_{i j}}$ of $\bbC^{\mu_{i j}}$ of eigenvectors of $\ell_{i j}$ with eigenvalues $s_{i j}^{p} = \pm 1$ such that $\vm_{j i}^{p} = L_{i j}(\Bvm_{i j}^{p})$ and $\vm_{i j}^{p} = \epsilon\, L_{j i}(\Bvm_{j i}^{p})$ for any $i<j$, and $s_{j i}^{p} = \epsilon'' s_{i j}^{p}$.

For $i = j$, $\epsilon = 1$, and $\epsilon'' = 1$ ($KO$-dimension $0$), there is an orthonormal basis $\{ \vm_{ii}^{p} \}_{1 \leq p \leq \mu_{ii}}$ of $\bbC^{\mu_{ii}}$ of eigenvectors of $\ell_{ii}$ with eigenvalues $s_{i}^{p} = \pm 1$ such that $\vm_{ii}^{p} = L_{ii}(\Bvm_{ii}^{p})$.

For $i = j$ and $\epsilon = -1$ ($KO$-dimensions $2$ and $4$), or $\epsilon = 1$ and $\epsilon'' = -1$ ($KO$-dimension $6$), $\mu_{ii}$ is even and there is an orthonormal basis $\{ \vm_{ii}^{p} \}_{1 \leq p \leq \mu_{ii}}$ of $\bbC^{\mu_{ii}}$ of eigenvectors of $\ell_{ii}$ with eigenvalues $s_{i}^{p} = \pm 1$ such that $\vm_{ii}^{2a} = L_{ii}(\Bvm_{ii}^{2a-1})$, $\vm_{ii}^{2a-1} = \epsilon\, L_{ii}(\Bvm_{ii}^{2a})$, and $s_{i}^{2a} = \epsilon'' s_{i}^{2a-1}$ for any $a = 1, \dots, \mu_{ii}/2$.
In $KO$-dimensions $2$ and $6$, one can choose the basis such that $s_{i}^{2a} = +1$ and $s_{i}^{2a-1} = -1$.\end{proposition}

The proofs of Prop.~\ref{prop basis odd case} and \ref{prop basis even case} can be found with the present notations in \cite{Nieu22a}. They are adapted from \cite{Suij15a}. Let us just mention some points that will be used later (see proof of Prop.~\ref{prop diagonalization spm for KO dims} in Sect.~\ref{sec AF algebras}). For $1 \leq i < j  \leq r$, the orthonormal basis $\{ \vm_{i j}^{p} \}_{1 \leq p \leq \mu_{i j}}$ of $\bbC^{\mu_{i j}}$ can be chosen with few constraints, and we construct from it the basis $\{ \vm_{j i}^{p} \defeq L_{i j}(\Bvm_{i j}^{p}) \}_{1 \leq p \leq \mu_{j i}}$ of $\bbC^{\mu_{j i}}$. In that construction, some free choices can be made for later purpose. For $i = j$, the proof relies on an iterative construction of the basis $\{ \vm_{i i}^{p} \}_{1 \leq p \leq \mu_{i i}}$ of $\bbC^{\mu_{i i}}$ using properties of the map $L_{ii}$ (and $\ell_{i i}$ in the even case). Here again, some free choices are allowed at all steps.

\smallskip
We are now in position to use these results to decompose in a suitable way the Hilbert space $\hs$ into irreps. We already know that $\hs = \toplus_{{i}, {j}=1}^{r} \hhs_{\bn_i \bn_j}$ and that $\hhs_{\bn_i \bn_j} = \bbC^{n_{i}} \otimes \bbC^{\mu_{i j}} \otimes \bbC^{n_{j} \circ}$. Using the orthonormal basis $\{ \vm_{i j}^{p} \}_{1 \leq p \leq \mu_{i j}}$ of $\bbC^{\mu_{i j}}$ given in Prop.~\ref{prop basis odd case} or Prop.~\ref{prop basis even case}, let us define 
\begin{align}
\label{hsv from vmijp basis}
\hs_v \defeq \Span \{ \xi_{i} \otimes \vm_{i j}^{p} \otimes \eta_{j}^\circ \mid \xi_{i} \in \bbC^{n_{i}} \text{ and }  \eta_{j}^\circ \in \bbC^{n_{j} \circ} \} \simeq \hs_{\bn_i \bn_j}
\end{align}
 We have then the orthogonal decomposition of $\hs$ along irreps: $\hs = \toplus_{v \in \Gamma^{(0)}} \hs_v$.

Consider first the odd case. For any $1 \leq i, j \leq r$, define the set $\Gamma^{(0)}_{\bn_i \bn_j} \defeq \{ (i, p, j) \mid 1 \leq p \leq \mu_{i j} \}$, and for any $v = (i, p, j) \in \Gamma^{(0)}_{\bn_i \bn_j}$, define $\lambda, \rho : \Gamma^{(0)}_{\bn_i \bn_j} \to \Lambda$ as $\lambda(v) \defeq \bn_i$ and $\rho(v) \defeq \bn_j$. Notice that $\mu_{i j} = \# \Gamma^{(0)}_{\bn_i \bn_j}$. Define $\Jim : \Gamma^{(0)}_{\bn_i \bn_j} \to \Gamma^{(0)}_{\bn_j \bn_i}$ as $\Jim(v) \defeq (j, p, i)$ for any $v = (i, p, j)$. These maps induce an involution on $\Gamma^{(0)} \defeq \cup_{i,j=1}^{r} \Gamma^{(0)}_{\bn_i \bn_j}$ with the property $\lambda \circ \Jim = \rho$ (and so $\rho \circ \Jim = \lambda$), where $\lambda, \rho : \Gamma^{(0)} \to \Lambda$ are defined in an obvious way. This involution encodes some properties of the family of maps $L_{i j}$, and so of the map $J : \hs_v \to \hs_{\Jim(v)}$ for any $v \in \Gamma^{(0)}$. 

In the even case, the basis in Prop.~\ref{prop basis even case} are composed of eigenvectors of $\gamma$, and by construction, $\gamma$ is the multiplication by $\pm 1$ on every $\hs_v$. We define a grading decoration of $v$ as $s(v)  = \pm 1$, which is the eigenvalue of the associated eigenvector. Notice then that $s \circ \Jim = \epsilon'' s$ as can be checked in Prop.~\ref{prop basis even case}. The grading decoration $s$ fully determines $\gamma$.

The Dirac operator decomposes along the orthogonal subspaces $\hs_v$ as $D_e : \hs_{v_1} \to \hs_{v_2}$ where we define $e \defeq (v_1, v_2) \in  \Gamma^{(0)} \times  \Gamma^{(0)}$. With $\Be \defeq (v_2, v_1)$, $D^\dagger = D$ is equivalent to $D_{\Be} = D_e^\dagger$. Moreover, the first-order condition imposes some restrictions on the $e = (v_1, v_2)$ such that $D_e \neq 0$ (see below). Let $\Jim(e) \defeq (\Jim(v_1), \Jim(v_2))$. Then the relation $J D = \epsilon' D J$ implies that $D_{e}$  and $D_{\Jim(e)}$ are related by $J$ and $\epsilon'$ (an explicit expression is given below). In particular, they are both zero or non-zero at the same time. In the even case, the relation $\gamma D = - D \gamma$ implies that $D_{e}$ is non-zero only when $s(v_2) = - s(v_1)$.

\bigskip
Let us abstract the construction using a decorated graph $\Gamma$, together with $\Lambda$ and the $KO$-dimension $d$. 
\begin{enumerate}
\item The set of vertex $\Gamma^{(0)}$ of the graph is equipped with two maps $\lambda, \rho : \Gamma^{(0)} \to \Lambda \times \Lambda$, that we write as a single map $\pi_{\lambda\rho} \defeq \lambda \times \rho$, and define $i(v) \defeq i$ and $j(v) \defeq j$ for $\pi_{\lambda\rho}(v) = (\bn_i, \bn_j)$.

\item There is an involution $\Jim : \Gamma^{(0)} \to \Gamma^{(0)}$ such that $\lambda \circ \Jim = \rho$ and such that $\Jim(v) = v$ when $\lambda(v) = \rho(v)$ in $KO$-dimensions $0$, $1$, and $7$. 

\item For any vertex $v \in \Gamma^{(0)}$ with $\pi_{\lambda\rho}(v) = (\bn_i, \bn_j)$, define $\hs_v \defeq \hs_{\lambda(v) \rho(v)} = \bbC^{\lambda(v)} \otimes \bbC^{\rho(v) \circ} = \bbC^{n_{i}} \otimes \bbC^{n_{j} \circ}$. The element $(\bn_i, \bn_j) \in \Lambda \times \Lambda$ is a decoration of the vertex $v$. 

\item Define $\Gamma^{(0)}_{\bn_i \bn_j} \defeq \{ v \in \Gamma^{(0)} \mid \pi_{\lambda\rho}(v) = (\bn_i, \bn_j) \} = \pi_{\lambda\rho}^{-1}(\bn_i, \bn_j)$ and $\mu_{ij} \defeq \# \Gamma^{(0)}_{\bn_i \bn_j}$.

\item  Define $\hJim_{v} : \hs_v \to \hs_{\Jim(v)}$ as $\hJim_{v}(\xi^{(v)} \otimes \eta^{(v)\circ}) = \eta^{(v)} \otimes \xi^{(v)\circ}$ for any $\xi^{(v)} \in \bbC^{\lambda(v)}$ and $\eta^{(v)\circ} \in \bbC^{\rho(v)\circ}$. Notice that $\hJim_{\Jim(v)} \circ \hJim_{v} = \Id_{\hs_v}$.

\item If the $KO$-dimension is even, a second decoration of each vertex is the assignment of a grading map $s(v) = \pm 1$ such that $s \circ \Jim = \epsilon'' s$.

\item If the $KO$-dimension is $2$, $3$, $4$, $5$, or $6$, then $\mu_{ii}$ is even and another decoration of each vertex $v \in\Gamma^{(0)}_{\bn_i \bn_i}$  is the parity $\chi(v) = 0, 1$ such that $\chi(\Jim(v)) = 1 - \chi(v)$, so that half of the vertices in $\Gamma^{(0)}_{\bn_i \bn_i}$ are decorated by the value $0$ (resp. $1$).

\item For any $v \in \Gamma^{(0)}$, define
\begin{align}
\label{eq epsilon(v, d)}
\epsilon(v, d) \defeq 
\begin{cases}
1 & \text{for $i(v) < j(v)$},\\
\epsilon & \text{for $i(v) > j(v)$},\\
1 & \text{for $i(v) = j(v)$ and $d = 0, 1, 7$},\\
\epsilon^{\chi(v)} & \text{for $i(v) = j(v)$ and $d = 2, 3, 4, 5, 6$.}
\end{cases}
\end{align}
One can check that $\epsilon(v, d) \epsilon(\Jim(v), d) = \epsilon$ for any $v \in \Gamma^{(0)}$.

\item For every $e = (v_1, v_2) \in \Gamma^{(0)} \times \Gamma^{(0)}$, let $\Be \defeq (v_2, v_1)$ and $\Jim(e) \defeq (\Jim(v_1), \Jim(v_2))$.

\item The space $\Gamma^{(1)} \subset \Gamma^{(0)} \times \Gamma^{(0)}$ of edges of the graph are couples $e = (v_1, v_2)$ such that:
\begin{enumerate}
\item $\lambda(v_1) = \lambda(v_2)$ or $\rho(v_1) = \rho(v_2)$ (or both);
\item $s(v_2) = - s(v_1)$ in the even case;
\item  there is a non-zero linear map $D_e : \hs_{v_1} \to \hs_{v_2}$ such that:
\begin{enumerate}
\item $D_{\Be} = D_e^\dagger : \hs_{v_2} \to \hs_{v_1}$;

\item $D_{\Jim(e)} = \epsilon' \epsilon(v_1, d) \epsilon(v_2, d)\, \hJim_{v_2} J_0 D_e J_0 \hJim_{\Jim(v_1)} : \hs_{\Jim(v_1)} \to \hs_{\Jim(v_2)}$;
\label{item DJim(e)}

\item For $\lambda(v_1) = \lambda(v_2)$ and $\rho(v_1) \neq \rho(v_2)$, $D_e = \bbbone_{n_{i_1}} \otimes D_{R,e}$ with $D_{R,e} : \bbC^{n_{j_1}\circ} \to \bbC^{n_{j_2}\circ}$; \label{item De = 1 DRe}

\item For $\lambda(v_1) \neq \lambda(v_2)$ and $\rho(v_1) = \rho(v_2)$,  $D_e = D_{L,e} \otimes \bbbone_{n_{j_1}}$ with  $D_{L,e} : \bbC^{n_{i_1}} \to \bbC^{n_{i_2}}$. \label{item De = DLe 1}
\end{enumerate}
Then $D_e$ defines a decoration of $e$.

\end{enumerate}
\end{enumerate}

\medskip
For any $\xi_{i} \otimes \eta_{j}^\circ \in \hs_{v_1}$, it is convenient to write $D_e (\xi_{i_1} \otimes \eta_{j_1}^\circ) = D_{L,e}^{(1)} \xi_{i_1} \otimes D_{R,e}^{(2)} \eta_{j_1}^\circ$ as a sumless Sweedler-like notation\footnote{This notation is usual for computations on coalgebras.} where there is an implicit summation over finite families of operators $D_{L,e}^{(1)} : \bbC^{n_{i_1}} \to \bbC^{n_{i_2}}$ and $D_{R,e}^{(2)} : \bbC^{n_{j_1}\circ} \to \bbC^{n_{j_2}\circ}$. In the previous points~\ref{item De = 1 DRe} and \ref{item De = DLe 1}, this decomposition is explicitly given (summation reduced to a single term).

One can see $\Gamma^{(0)}$ as a set of points on top of the points $\Lambda \times \Lambda$, where the (down) projection is $\pi_{\lambda\rho}$. Each point in $\Gamma^{(0)}_{\bn_i \bn_j} = \pi_{\lambda\rho}^{-1}(\bn_i, \bn_j)$ is a copy of the irrep $\hs_{\bn_i \bn_j}$: we can look at $v$ as an element of the “fiber” $\Gamma^{(0)}_{\bn_i \bn_j}$ on top of $(\bn_i, \bn_j)$. The edges in $\Gamma^{(1)}$, once projected in $\Lambda \times \Lambda$, connect points horizontally, vertically  or self-connect a (projected) point. A convenient representation of $\Gamma$ is then a $3$-dimensional set of points decorated by some values (as seen above) and linked by decorated lines, see Fig.~\ref{fig krajewskiDiagram}.

\begin{figure}
{\centering \includegraphics[]{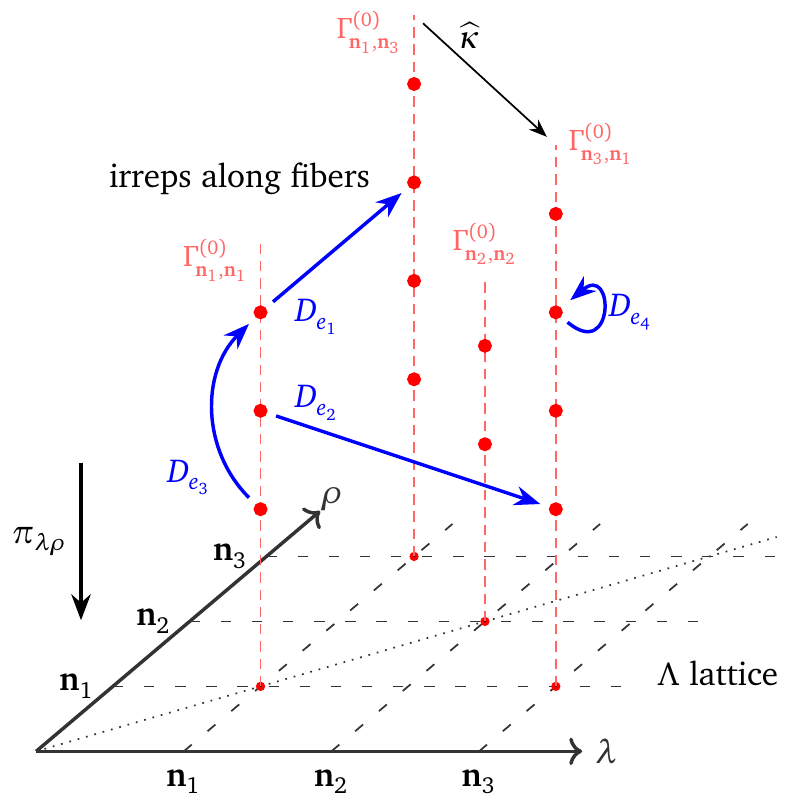}\par}
\caption{A Krajewski diagram for the algebra $M_{n_1} \oplus M_{n_2} \oplus M_{n_3}$ in the $3$-dimensional representation explained in the text. The component $D_{e_1}$ (resp. $D_{e_2}$) of the Dirac operator joins two irreps with same $\lambda$ (resp. same $\rho$); $D_{e_3}$ and $D_{e_4}$ join irreps with the same $\lambda$ and $\rho$, but $D_{e_4}$ can be non zero only in the odd case. The maps $\hJim$ realize the axial symmetry defined by the dotted line in the $\Lambda$ lattice.}
\label{fig krajewskiDiagram}
\end{figure}

These data completely determine a real (odd or even) spectral triple. A vertex $v \in \Gamma^{(0)}$ defines the irrep $\hs_v$ with multiplicity $\mu(v) \defeq \# \Gamma^{(0)}_{\pi_{\lambda\rho}(v)}$, so the Hilbert space is $\hs \defeq \toplus_{v \in \Gamma^{(0)}} \hs_v = \toplus_{i,j=1}^{r} \hhs_{\bn_i \bn_j}$ with $\hhs_{\bn_i \bn_j} = \bbC^{n_i} \otimes \bbC^{\mu_{i j}} \otimes \bbC^{n_j \circ}$.  Any operator $A$ on $\hs$ decomposes into linear maps $A_{v_2}^{v_1} : \hs_{v_1} \to \hs_{v_2}$.

It will be useful to describe the representation $\pi$ along these two decompositions. For any $a = \toplus_{i=1}^{r} a_i \in \algA$, any $v=(i,p,j)$, and any $\psi = \toplus_{v \in \Gamma^{(0)}} \psi_v = \toplus_{i,j=1}^{r} \xi_{i} \otimes \vm_{i j}^{p} \otimes \eta_{j}^\circ$ with $\psi_v \in \hs_v = \bbC^{\lambda(v)} \otimes \bbC^{\rho(v) \circ}$ and $\xi_{i} \otimes \vm_{i j}^{p} \otimes \eta_{j}^\circ \in \bbC^{n_i} \otimes \bbC^{\mu_{i j}} \otimes \bbC^{n_j \circ}$, one has $\pi(a) \psi = \toplus_{v \in \Gamma^{(0)}} a_{i(v)} \psi_v = \toplus_{i, j=1}^{r} (a_i \xi_{i}) \otimes \vm_{i j}^{p} \otimes \eta_{j}^\circ$ where $a_{i(v)} \psi_v$ is the multiplication of the matrix $a_{i(v)}$ on the left factor of $\bbC^{\lambda(v)} \otimes \bbC^{\rho(v) \circ}$ and $a_i \xi_{i}$ is the usual matrix multiplication on $\bbC^{n_i}$. In other words, the decomposition of the operator $\pi(a)$ along the $\hs_v$'s is
\begin{align}
\label{eq pi(a) decomposition along Hv}
\pi(a)_{v_2}^{v_1} = a_{i(v_1)} \delta_{v_2}^{v_1} : \hs_{v_1} \to \hs_{v_2}
\end{align} 
(where $\delta_{v_2}^{v_1}$ is the Kronecker symbol). In the real case, for any $a = \toplus_{i=1}^{r} a_i \in \algA$, any $b = \toplus_{j=1}^{r} b_j \in \algA$, and any $\psi_v \in \hs_v$, one has $a b^\circ \psi_v = a_{i(v)} b_{j(v)}^\circ \psi_v \in \hs_v$ ($\pi$ omitted) where $b_{j(v)}^\circ$ acts on the right factor of $\bbC^{\lambda(v)} \otimes \bbC^{\rho(v) \circ}$ (see footnote~\ref{footnote circ action}). A similar relation holds on $\bbC^{n_i} \otimes \bbC^{\mu_{i j}} \otimes \bbC^{n_j \circ}$.

In the even case, $\gamma$ is determined as the multiplication by the decoration $s(v) = \pm 1$ on $\hs_{v}$. The real operator $J$ is reconstructed by the family of maps
\begin{align*}
J_v \defeq \epsilon(v, d)\, J_0 \hJim_{v} = \epsilon(v, d)\, \hJim_{v} J_0 : \hs_v \to \hs_{\Jim(v)}
\end{align*}
or, equivalently, with $i = i(v)$ and $j = j(v)$, and any $\xi_{i} \otimes \eta_{j}^\circ \in \hs_v$, $J(\xi_{i} \otimes \eta_{j}^\circ) = \epsilon(v, d)\, \Beta_{j} \otimes \Bxi_{i}^\circ \in \hs_{\Jim(v)}$. 
In other words, $J_{v_2}^{v_1} = \epsilon(v_1, d)\delta_{\Jim(v_1)}^{v_1} \, J_0 \hJim_{v_1}$.

The Dirac operator is reconstructed by the decorations $D_{e}$ of the edges $e \in \Gamma^{(1)}$. Introduce an orthonormal basis for each $\bbC^{\mu_{ij}}$ and label all these basis vectors in the union of all the $\bbC^{\mu_{ij}}$'s as $\{ \vm_v \}_{v \in \Gamma^{(0)}}$: for any $v \in \Gamma^{(0)}_{\bn_i \bn_j}$, $\vm_v$ is an element of an orthonormal basis of $\bbC^{\mu_{ij}}$. We use the identification $\hs_v = \Span \{ \xi \otimes \vm_v \otimes \eta^\circ \mid \xi \in \bbC^{\lambda(v)}, \eta^\circ \in \bbC^{\mu(v) \circ} \}$. For any $e = (v_1, v_2) \in \Gamma^{(1)}$ with $v_1 \in \Gamma^{(0)}_{\bn_{i_1} \bn_{j_1}}$ and $v_2 \in \Gamma^{(0)}_{\bn_{i_2} \bn_{j_2}}$,  define $\hD_{e} : \hhs_{\bn_{i_1} \bn_{j_1}} \to \hhs_{\bn_{i_2} \bn_{j_2}}$, for any $\xi \in \bbC^{n_{i_1}}$ and $\eta^\circ \in \bbC^{n_{j_1}\circ}$, as
\begin{align*}
\hD_{e} (\xi \otimes \vm_{v} \otimes \eta^\circ)
&=
\begin{cases}
0 & \text{if $v \neq v_1$} \\
( D_{L,e}^{(1)} \xi) \otimes \vm_{v_2} \otimes (D_{R,e}^{(2)} \eta^\circ) & \text{if $v = v_1$}.
\end{cases}
\end{align*}
Then $D$ is completely given as a matrix with entries $\hD_{e}$ in the decomposition $\hs = \toplus_{i,j=1}^{r} \hhs_{\bn_i \bn_j}$.

\medskip
One can write a specific version of \eqref{eq piD sum algebras} for the decomposition $\hs = \toplus_{v \in \Gamma^{(0)}} \hs_v$ in terms of the operators $D_e$ for $e \in \Gamma^{(1)}$. For any $\bomega \in \bOmega^n_U(\algA) \subset \kT^n \algA$ which decomposes along a sum of typical terms $\toplus_{i_1, \dots, i_{n-1} =1}^{r} \big( a^0_{i} \otimes a^1_{i_1} \otimes \cdots \otimes a^{n-1}_{i_{n-1}} \otimes a^n_{j} \big)_{i, j = 1}^{r} \in \kT^n \algA$ and any $v_0, v_n \in \Gamma^{(0)}$, one has
\begin{align}
\label{eq piD AF algebras}
\piD(\bomega)_{v_0}^{v_n} = 
\!\!\!\!\!\!\!\sum_{ \substack{ \text{all terms at the}\\ \text{$(i(v_0), i(v_n))$ entry in $\bomega$} } } \!\!\!\!\!\!\!
\tsum_{v_1, \dots, v_{n-1} \in \Gamma^{(0)}}
a^{0}_{i(v_0)} D_{(v_1, v_0)} a^{1}_{i(v_1)}  D_{(v_2, v_1)} \cdots
 a^{n-1}_{i(v_{n-1})} D_{(v_n, v_{n-1})} a^{n}_{i(v_n)} : 
\hs_{v_n} \to \hs_{v_0}
\end{align}
In this formula, one supposes $D_{(v_{i+1}, v_i)} = 0$ when $(v_{i+1}, v_i) \not\in \Gamma^{(1)}$.

\section{Lifting one step of the defining inductive sequence}
\label{sec one step in the sequence}

In this section, we study the “lifting” to spectral triples of a one-to-one homomorphism $\phi : \algA \to \algB$. As explained in Sect.~\ref{sec introduction}, the main idea, which is central in our paper, is to define a notion of $\phi$-compatibility for the structures defining spectral triples $(\algA, \hsA, \DA, \JA, \gammaA)$ and $(\algB, \hsB, \DB, \JB, \gammaB)$ on top of $\algA$ and $\algB$. This construction, applied in Sect.~\ref{sec AF algebras} to $AF$-algebras, can be interpreted as a lift of arrows in a Bratteli diagram to arrows between Krajewski diagrams.

\subsection{General situations}
\label{sec general situations}

The first structure to consider are  the Hilbert spaces $\hsA$ and $\hsB$, that we can consider as left modules on $\algA$ and $\algB$ via their corresponding representations that are not explicitly written in the following. Similarly to \cite[Definition~15]{MassNieu21q}, we introduce the following definition:

\begin{definition}
\label{def phi phiH}
A morphism (bounded linear map) of Hilbert spaces $\phiH : \hsA \to \hsB$ is $\phi$-compatible if $\phiH(a \psi) = \phi(a) \phiH(\psi)$ for any $a \in \algA$ and $\psi \in \hsA$ (the representations $\piA$ and $\piB$ are omitted in this relation).
\end{definition}

In this definition, we suppose that $\phiH$ is only a bounded linear map: we look at the category of Hilbert spaces as a dagger category. But in order to get some useful results, we will assume later that $\phiH$ is an isometry (see Sect.~\ref{subsec normalized phiH map}).

Given the morphism $\phiH : \hsA \to \hsB$, one can decompose $\hsB$ as $\hsB = \phiH(\hsA) \oplus \phiH(\hsA)^{\perp}$ in a unique $\phiH$-dependent way, where $\phiH(\hsA) = \Ran(\phiH)$ is the range of $\phiH$. This implies that any operator $B$ on $\hsB$ can be decomposed as $B = \smallpmatrix{ B_{\phi}^{\phi} & B_{\phi}^{\perp} \\ B_{\perp}^{\phi} & B_{\perp}^{\perp} }$ with obvious notations, for instance $B_{\phi}^{\perp} : \phiH(\hsA)^{\perp} \to \phiH(\hsA)$. In this orthogonal decomposition, one has $B^\dagger = \smallpmatrix{ B_{\phi}^{\phi \dagger} & B_{\perp}^{\phi \dagger} \\ B_{\phi}^{\perp \dagger} & B_{\perp}^{\perp \dagger} }$. 

\begin{definition}[$\phi$-compatibility of operators]
\label{def phi compatibility operators}
Given two operators $A$ on $\hsA$ and $B$ on $\hsB$, we say that they are $\phi$-compatible if $\phiH( A \psi) = B_{\phi}^{\phi} \phiH( \psi)$ for any $\psi \in \hsA$ (equality in $\phiH(\hsA)$).
\end{definition}

This definition makes sense since both sides belong to $\phiH(\hsA)$. Notice that, by an abuse of notation, we use the terminology “$\phi$-compatibility” but this notion depends on the couple of maps $(\phi, \phiH)$.

One can define a stronger $\phi$-compatibility between $A$ and $B$:
\begin{definition}[Strong $\phi$-compatibility of operators]
\label{def strong phi compatibility operators}
Given two operators $A$ on $\hsA$ and $B$ on $\hsB$, we say that they are strong $\phi$-compatible if $\phiH( A \psi) = B \phiH( \psi)$ for any $\psi \in \hsA$ (equality in $\hsB$).
\end{definition}

\begin{remark}
 Notice that these two $\phi$-compatibility conditions imply that $\Ker \phiH \subset \Ker \phiH \circ A$, since, if $\psi \in \Ker \phiH$, then $0 = B_{\phi}^{\phi} \phiH( \psi) = \phiH( A \psi)$ in the first case, and similarly in the second case. A sufficient condition for this to hold for every $A$ is to require $\phiH$ to be one-to-one.
 \end{remark}
 
 \begin{remark}
 \label{rmk piA(a) st phi comp piB(phi(a))}
Definition~\ref{def phi phiH} implies that $\piA(a)$ and $\piB(\phi(a))$ are strong $\phi$-compatible for any $a \in \algA$.
 \end{remark}
 
The following Proposition gives other consequences of the two definitions, where diagonality refers to the previously defined $2\times 2$ matrix decomposition.

\begin{proposition}\phantom{A}
\label{prop strong and not strong phi compatibility}
\begin{enumerate}
\item $\phi$-compatibility and strong $\phi$-compatibility are stable under sums of operators.\label{item phi comp and st phi comp sums}

\item Compositions of strong $\phi$-compatible operators are strong $\phi$-compatible (this is not necessarily true for $\phi$-compatible operators). \label{item st phi comp composition}

\item If $A$ on $\hsA$ and $B$ on $\hsB$ are strong $\phi$-compatible then $B_{\perp}^{\phi} = 0$. \label{item st phi comp Bperpphi = 0}

\item Strong $\phi$-compatibility implies $\phi$-compatibility. \label{item st phi comp implies phi comp}

\item If $B_{\perp}^{\phi} = 0$, the $\phi$-compatibility implies the strong $\phi$-compatibility. \label{item Bperpphi = 0 and phi comp implies st phi comp}

\item When $B$ is self-adjoint, strong $\phi$-compatibility implies that $B$ is diagonal.\label{item B self-adjoint str phi comp}

\item If $A$ on $\hsA$ and $B$ on $\hsB$ are strong $\phi$-compatible and $A$ and $B$ are unitaries, then $A^\dagger$ and $B^\dagger$  are strong $\phi$-compatible and $B$ is diagonal.\label{item A B unitaries str phi comp}

\item For any $a \in \algA$, the operator $\piB \circ \phi(a)$ on $\hsB$ reduces to a diagonal matrix $\piB \circ \phi(a) = \smallpmatrix{ \piB \circ \phi(a)_{\phi}^{\phi} & 0 \\ 0 & \piB \circ \phi(a)_{\perp}^{\perp} }$.\label{item piBphi(a) diagonal}
\end{enumerate}
\end{proposition}

\begin{proof}
Point~\ref{item phi comp and st phi comp sums} is obvious by linearity of the compatibility conditions and the matrix decompositions. For point~\ref{item st phi comp composition}, let $A_1, A_2$ be two operators on $\hsA$ and $B_1, B_2$ two operators on $\hsB$ which are strong $\phi$-compatible with $A_1$ and $A_2$ respectively. Then for any $\psi \in \hsA$, one has $\phiH(A_1 A_2 \psi) = B_1 \phiH(A_2 \psi) = B_1 B_2 \phiH(\psi)$ so that $A_1 A_2$ is strong $\phi$-compatible with $B_1 B_2$. For $\phi$-compatibility, this line of reasoning is not possible in general.

One can identify $\phiH( \psi)$ with $\smallpmatrix{\phiH( \psi) \\ 0} \in \phiH(\hsA) \oplus \phiH(\hsA)^{\perp} = \hsB$ (resp. $\phiH( A \psi)$ with $\smallpmatrix{\phiH( A \psi) \\ 0}$), so that $B_{\phi}^{\phi} \phiH( \psi)$ identifies with $\smallpmatrix{ B_{\phi}^{\phi} \phiH( \psi) \\ 0 }$ while $B \phiH( \psi)$ identifies with $\smallpmatrix{ B_{\phi}^{\phi} \phiH( \psi) \\ B_{\perp}^{\phi} \phiH( \psi) }$. The $\phi$-compatibility condition implies that the map $B_{\phi}^{\phi} : \phiH(\hsA) \to \phiH(\hsA)$ is completely determined by $A$ and $\phiH$, while the strong $\phi$-compatibility condition implies firstly that $\phiH( A \psi) = B_{\phi}^{\phi} \phiH( \psi)$, and secondly that $B_{\perp}^{\phi} : \phiH(\hsA) \to \phiH(\hsA)^{\perp}$ is the zero map, which is point~\ref{item st phi comp Bperpphi = 0}. So, using these results, one gets that the strong $\phi$-compatibility implies the $\phi$-compatibility condition (which only constrains the  $B_{\phi}^{\phi}$ component of $B$), which is point~\ref{item st phi comp implies phi comp}. For point~\ref{item Bperpphi = 0 and phi comp implies st phi comp}, from $B_{\perp}^{\phi} = 0$ and $\phiH( A \psi) = B_{\phi}^{\phi} \phiH( \psi)$, one gets $B \phiH( \psi) = \smallpmatrix{ B_{\phi}^{\phi} \phiH( \psi) \\ B_{\perp}^{\phi} \phiH( \psi) } = \smallpmatrix{ B_{\phi}^{\phi} \phiH( \psi) \\ 0 } = \smallpmatrix{\phiH( A \psi) \\ 0} = \phiH( A \psi)$, which is the strong $\phi$-compatibility condition.

Point~\ref{item B self-adjoint str phi comp}: if $B$ is self-adjoint, the condition $B = B^\dagger$ implies $B_{\phi}^{\perp \dagger} = B_{\perp}^{\phi} = 0$, so that $B$ is diagonal.

Point~\ref{item A B unitaries str phi comp}: if $A$ and $B$ are unitaries, then $\phiH( \psi) = \phiH( A^\dagger A \psi)$ on the one hand and $\phiH( \psi) = B^\dagger B \phiH( \psi)$ on the other hand, so that $\phiH( A^\dagger A \psi) = B^\dagger B \phiH( \psi) = B^\dagger \phiH( A \psi)$. Since $A$ is invertible, any $\psi' \in \hsA$ can be written as $\psi' = A \psi$, so that $\phiH( A^\dagger \psi) = B^\dagger \phiH( \psi)$ for any $\psi$, which proves that $A^\dagger$ and $B^\dagger$  are strong $\phi$-compatible. The strong $\phi$-compatibilities implies $B_{\perp}^{\phi} = 0$ and $(B^\dagger)_{\perp}^{\phi} = B_{\phi}^{\perp \dagger} = 0$, and so $B$ is diagonal.

Point~\ref{item piBphi(a) diagonal}: let us use the notation $\piB \circ \phi(a) = \smallpmatrix{ \piB \circ \phi(a)_{\phi}^{\phi} & \piB \circ \phi(a)_{\phi}^{\perp} \\ \piB \circ \phi(a)_{\perp}^{\phi} & \piB \circ \phi(a)_{\perp}^{\perp} }$ for any $a \in \algA$. From Definition~\ref{def phi phiH}, $\piB \circ \phi(a)$ is strong $\phi$-compatible with $\piA(a)$, so that $\piB \circ \phi(a)_{\perp}^{\phi} = 0$. Since $\piB \circ \phi(a^\ast) = \piB \circ \phi(a)^\dagger$, this implies that $\piB \circ \phi(a^\ast)_{\phi}^{\perp} = 0$ for any $a$, so that $\piB \circ \phi(a)$ reduces to a diagonal matrix.
\end{proof}

One can associate to $B = \smallpmatrix{ B_{\phi}^{\phi} & B_{\phi}^{\perp} \\ B_{\perp}^{\phi} & B_{\perp}^{\perp} }$ the operator $\widehat{B}_{\phi}^{\phi} = \smallpmatrix{ B_{\phi}^{\phi} & 0 \\ 0 & 0 }$. Then the $\phi$-compatibility between $A$ and $B$ is equivalent to the strong $\phi$-compatibility between $A$ and $\widehat{B}_{\phi}^{\phi}$.

\smallskip
\begin{definition}[$\phi$-compatibility of spectral triples]
\label{def phi compatibility of spectral triples}
Assume given a $\phi$-compatible map $\phiH : \hsA \to \hsB$.

Two odd spectral triples $(\algA, \hsA, \DA)$ and $(\algB, \hsB, \DB)$ are said to be $\phi$-compatible if $\DA$ is $\phi$-compatible with $\DB$.

Two real spectral triples $(\algA, \hsA, \DA, \JA)$ and $(\algB, \hsB, \DB, \JB)$ are said to be $\phi$-compatible if $\DA$ (resp. $\JA$) is $\phi$-compatible with $\DB$ (resp. $\JB$).

In the even case for $\algA$, one requires that $\algB$ is also even and that the grading operators $\gammaA$ and $\gammaB$ are $\phi$-compatible.
\end{definition}

Strong $\phi$-compatibility of spectral triples can be defined in an obvious way.

\begin{remark}
Notice that strong $\phi$-compatibility of spectral triples is similar to the condition (3) given in \cite[Def~2.1]{FlorGhor19p} where their couple $(\phi, I)$ corresponds to our couple $(\phi, \phiH)$. We depart from this paper where  inductive sequences of spectral triples are studied in the following way: we will restrict our analysis to the algebraic part of spectral triples since only $AF$-algebras will be considered later, so that the analytic part is quite trivial in our situation, and we will focus on gauge fields theories defined on top of spectral triples. For instance, conditions like (ST1) (about the $*$-subalgebra $\algA^\infty$) and (ST2) (about the compactness of the resolvent of the Dirac operator) in  \cite{FlorGhor19p} will not be considered here. Other papers use also this notion of strong $\phi$-compatibility, see for instance  \cite{ChriIvan06a} and \cite{Lai13a}.
But, since we are interested in accumulating “new degrees of freedom” along the inductive limit, the $\phi$-compatibility condition will be more relevant than the strong $\phi$-compatibility condition in that respect.
\end{remark}

Since $\JA$ and $\JB$ define $\algA^{e} = \algA \otimes \algA^\circ$ and $\algB^{e} = \algB \otimes \algB^\circ$ modules structures on $\hsA$ and $\hsB$, it is convenient to express $\phi$-compatibility in terms of this structure. The homomorphism $\phi$ defines a canonical homomorphism of algebras $\phi^\circ : \algA^\circ \to \algB^\circ$ by the relation $\phi^\circ(a^\circ) \defeq \phi(a)^\circ$. We then define $\phi^{e} : \algA^{e} \to \algB^{e}$ as $\phi^{e} \defeq \phi \otimes \phi^\circ$, \textit{i.e.} $\phi^{e}(a_1 \otimes a_2^\circ) = \phi(a_1) \otimes \phi^\circ(a_2^\circ)$. Let $\modM$ (resp. $\modN$) be a $\algA$-bimodule (resp. $\algB$-bimodule), which is also a $\algA^{e}$-left module ( resp. $\algB^{e}$-left module) by $(a_1 \otimes a_2^\circ) e \defeq a_1 e a_2$ for any $e \in \modM$ and $a_1, a_2 \in \algA$ (and similar relations for $\algB$ and $\modN$). Then, we say that a linear map between the two bimodules $\phiMod : \modM \to \modN$ is $\phi$-compatible if it is $\phi^{e}$-compatible between the two left modules, that is $\phiMod( (a_1 \otimes a_2^\circ) e) = \phi^{e}(a_1 \otimes a_2^\circ) \phiMod(e)$, which is equivalent to $\phiMod( a_1 e  a_2) = \phi(a_1) \phiMod(e) \phi(a_2)$.

\begin{lemma}
Suppose that $\phiH : \hsA \to \hsB$ is $\phi$-compatible as a map of left modules and that $\JA$ and $\JB$ are strong $\phi$-compatible. Then $\phiH$ is $\phi^{e}$-compatible as a map between the bimodules defined by the real operators. 
\end{lemma}

\begin{proof}
For any $\psi \in \hsA$, $a_1, a_2 \in \algA$, by definition, one has $a_1 \psi a_2 = (a_1 \otimes a_2^\circ) \psi = a_1 \JA a_2^\ast \JA \psi$. On the one hand, since $\phiH$ is $\phi$-compatible, one has $\phiH(a_1 \psi) = \phi(a_1) \phiH(\psi)$. On the other hand, $\phiH(\psi a_2) = \phiH( \JA a_2^\ast \JA \psi) = \JB \phi(a_2)^\ast \JB \phiH(\psi) = \phiH(\psi) \phi(a_2)$.
\end{proof}

\begin{lemma}
\label{lemma JB JA st-phi-comp}
Suppose that $\JB$ is strong $\phi$-compatible with $\JA$:
\begin{enumerate}
\item $\epsilonA = \epsilonB$.

\item $\JB^{-1}$ is strong $\phi$-compatible with $\JA^{-1}$

\item $\JB$ is diagonal in its matrix decomposition.

\item If two operators $A$ on $\hsA$ and $B$ on $\hsB$ are $\phi$-compatible, then the operators $\JA A \JA^{-1}$ and $\JB B \JB^{-1}$ are $\phi$-compatible.
\end{enumerate}
\end{lemma}

\begin{proof}
From $\JA^2 = \epsilonA$ and $\JB^2 = \epsilonB$, one gets $\epsilonA \phiH(\psi) = \phiH(\JA^2 \psi) = \JB^2 \phiH(\psi) = \epsilonB \phiH(\psi)$ for any $\psi \in \hsA$, so that $\epsilonB = \epsilonA$. From this we deduce that $\JB^{-1} = \epsilonB \JB$ is  strong $\phi$-compatible with $\JA^{-1} = \epsilonA \JA$.

Let $\JB = \smallpmatrix{ \JB[,\phi]^{\phi} & \JB[,\phi]^{\perp} \\ \JB[,\perp]^{\phi} & \JB[,\perp]^{\perp} }$. Since $\JB$ is strong $\phi$-compatible with $\JA$, we already know that $\JB[,\perp]^{\phi} = 0$. Let $\psi_\algB \in \phiH(\hsB)$ and $\psi'_\algB \in \phiH(\hsA)^{\perp}$. Then $\JB (\psi_\algB) = \smallpmatrix{ \JB[,\phi]^{\phi}(\psi_\algB) \\ 0}$ and $\JB(\psi'_\algB) = \smallpmatrix{ \JB[,\phi]^{\perp}(\psi'_\algB) \\ \JB[,\perp]^{\perp}(\psi'_\algB) }$, so that $0 = \langle \psi'_\algB, \psi_\algB \rangle_{\hsB} = \langle \JB(\psi_\algB), \JB(\psi'_\algB) \rangle_{\hsB} = \langle \JB[,\phi]^{\phi}(\psi_\algB), \JB[,\phi]^{\perp}(\psi'_\algB) \rangle_{\hsB}$. From $\JB^{-1} = \epsilonB \JB$ and $\JB[,\perp]^{\phi} = 0$, one gets that  $\JB[,\phi]^{\phi}$ is invertible with $(\JB[,\phi]^{\phi})^{-1} = (\JB^{-1})_{\phi}^{\phi} = \epsilonB \JB[,\phi]^{\phi}$, so that $\JB[,\phi]^{\phi}(\phiH(\hsA)) = \phiH(\hsA)$, which implies that $\JB[,\phi]^{\perp}(\psi'_\algB) \in \phiH(\hsA)^{\perp}$, that is, $\JB[,\phi]^{\perp}(\psi'_\algB) = 0$ for any $\psi'_\algB \in \phiH(\hsA)^{\perp}$, and so $\JB[,\phi]^{\perp} = 0$.

From $(\JB B \JB^{-1})_{\phi}^{\phi} = \JB[,\phi]^{\phi} B_{\phi}^{\phi} (\JB[,\phi]^{\phi})^{-1}$, we deduce that the operators $\JA A \JA^{-1}$ and $\JB B \JB^{-1}$ are $\phi$-compatible.
\end{proof}

\begin{lemma}
\label{lemma gammaB gammaA phi compatibiliy diagonal}
Let us consider the even case and suppose $\gammaB$ is $\phi$-compatible with $\gammaA$.
\begin{enumerate}
\item Then $\gammaB$ is diagonal in its matrix decomposition, so that strong $\phi$-compatibility and $\phi$-compatibility between $\gammaB$ and $\gammaA$ are equivalent.\label{item gammaB diagonal}
 \item Then $\phiH$ is diagonal for the matrix decomposition induced by $\hsA = \hsA^{+} \oplus \hsA^{-}$ and $\hsB = \hsB^{+} \oplus \hsB^{-}$, so that $\phiH$ restricts to maps $\hsA^{\pm} \to \hsB^{\pm}$.\label{item phiH diagonal grading}
 \end{enumerate}
\end{lemma}

\begin{proof}
Point~\ref{item gammaB diagonal}: since $\gammaB^\dagger = \gammaB$, one has $\gammaB = \smallpmatrix{ \gammaB[,\phi]^{\phi} & \gammaB[,\phi]^{\perp} \\ \gammaB[,\phi]^{\perp \dagger} & \gammaB[,\perp]^{\perp} }$. The $\phi$-compatibility implies $(\gammaB[,\phi]^{\phi})^2 \phiH(\psi) = \phiH(\gammaA^2 \psi) = \phiH(\psi)$, so that $(\gammaB[,\phi]^{\phi})^2 = 1$. Since $\gammaB^2 = 1$, one has $(\gammaB[,\phi]^{\phi})^2 + \gammaB[,\phi]^{\perp} \gammaB[,\phi]^{\perp \dagger} = 1$, from which we get $\gammaB[,\phi]^{\perp} \gammaB[,\phi]^{\perp \dagger} =0$, which implies $\gammaB[,\phi]^{\perp} = 0$, so that $\gammaB$ is diagonal. By Prop.~\ref{prop strong and not strong phi compatibility}, this implies strong $\phi$-compatibility.

Point~\ref{item phiH diagonal grading}: for every $\psi \in \hsA^{\pm}$, one has $\pm \phiH( \psi) = \phiH( \gammaA \psi) = \gammaB \phiH( \psi)$, so that $\phiH( \psi) \in \hsB^{\pm}$.
\end{proof}

\begin{proposition}\phantom{A}
\label{prop KO dim strong phi compatibility}
\begin{enumerate}
\item If two (odd/even) real spectral triples are $\phi$-compatible and $\JB$ is strong $\phi$-compatible with $\JA$, then they have the same $KO$-dimension (mod 8).

\item If two (odd/even) real spectral triples are strong $\phi$-compatible, then they have the same $KO$-dimension (mod 8).
\end{enumerate}
\end{proposition}

\begin{proof}
Let $(\algA, \hsA, \DA, \JA)$ and $(\algB, \hsB, \DB, \JB)$ be two $\phi$-compatible real spectral triples such that $\JB$ is strong $\phi$-compatible with $\JA$. In the even case, consider the gradings $\gammaA$ and $\gammaB$. We already know from Lemma~\ref{lemma JB JA st-phi-comp} that $\epsilonA = \epsilonB$. Using the fact that $\JB$ and $\gammaB$ are diagonal (Lemmas~\ref{lemma JB JA st-phi-comp} and \ref{lemma gammaB gammaA phi compatibiliy diagonal}), one has $(\JB \DB)_{\phi}^{\phi} = \JB[, \phi]^{\phi} \DB[, \phi]^{\phi}$, $(\DB \JB)_{\phi}^{\phi} = \DB[, \phi]^{\phi} \JB[, \phi]^{\phi}$,  $(\JB \gammaB)_{\phi}^{\phi} = \JB[, \phi]^{\phi} \gammaB[, \phi]^{\phi}$, and $(\gammaB \JB)_{\phi}^{\phi} = \gammaB[, \phi]^{\phi} \JB[, \phi]^{\phi}$, which implies, by $\phi$-compatibility, that $\epsilonB' = \epsilonA'$ and $\epsilonB'' = \epsilonA''$.

The second assertion follows from the first one.
 \end{proof}

The requirement that $\JB$ be strong $\phi$-compatible with $\JA$ seems to be inevitable in the generic situation to get the same $KO$-dimension. In the case of $AF$-algebras, this requirement will be a consequence of another requirement on the $\phiH$ map, see Prop.~\ref{prop JB JA strong phi compatibiliy relation on ukappa}.

\bigskip
Let $(\algA, \hsA, \DA, \JA)$ and $(\algAp, \hsAp, \DAp, \JAp)$ be two unitary equivalent real spectral triples for $\UA : \hsA \to \hsA'$ and $\phiA : \algA \to \algAp$ and let $(\algB, \hsB, \DB, \JB)$ and $(\algBp, \hsBp, \DBp, \JBp)$ be two unitary equivalent real spectral triples for $\UB : \hsB \to \hsBp$ and $\phiB : \algB \to \algBp$.

\begin{proposition}
\label{prop unitary equi triple and (st) phi comp}
Suppose that $(\algA, \hsA, \DA, \JA)$ and $(\algB, \hsB, \DB, \JB)$ are strong $\phi$-compatible (resp. $\phi$-compatible), and that there is a homomorphism of algebras $\phi' : \algA' \to \algB'$ and a morphism $\phiH' : \hsAp \to \hsBp$ such that $\phi' \circ \phiA = \phiB \circ \phi$ and $\phiH' (\UA \psi) = \UB \phiH(\psi)$ for any $\psi \in \hsA$ (resp. and suppose that $\UB$ is diagonal). Then $(\algAp, \hsAp, \DAp, \JAp)$ and $(\algBp, \hsBp, \DBp, \JBp)$ are strong $\phi'$-compatible (resp. $\phi$-compatible). If the spectral triples are even, the result holds also.
\end{proposition}

This result shows that strong $\phi$-compatibility (resp. $\phi$-compatibility) is transported by unitary equivalence if one assumes some natural conditions on the maps $\phi'$ and $\phiH'$, which are the commutativity of the following diagrams:
\begin{equation*}\tikzexternaldisable
\begin{tikzcd}[column sep=20pt, row sep=20pt]
\algA
	\arrow[r, "\phi"]
	\arrow[d, "\phiA"']
& \algB
	\arrow[d, "\phiB"]
\\
\algAp
	\arrow[r, "\phi'"]
&
\algBp	
\end{tikzcd} 
\quad \text{ and } \quad
\tikzsetnextfilename{fig-cd-2}
 \begin{tikzcd}[column sep=20pt, row sep=20pt]
\hsA
	\arrow[r, "\phiH"]
	\arrow[d, "\UA"']
& \hsB
	\arrow[d, "\UB"]
\\
\hsAp
	\arrow[r, "\phiH'"]
&
\hsBp	
\end{tikzcd} 
\end{equation*}\tikzexternalenable
(resp. and one requires $\UB$ to be diagonal).

\begin{proof}
For any $\psi' \in \hsAp$, let $\psi \in \hsA$ be the unique vector such that $\psi' = \UA \psi$, and for any $a' \in \algAp$, let $a \in \algA$ the unique element such that $a' = \phiA(a)$. Then one has $\phiH'( \piAp(a') \psi') = \phiH'( (\piAp \circ \phiA(a)) \UA \psi) = \phiH'( \UA \piA(a) \psi) = \UB \phiH( \piA(a) \psi) = \UB (\piB \circ \phi(a)) \phiH( \UA^{-1} \psi') = \UB (\piB \circ \phi(a)) \UB^{-1} \phiH'( \psi') = (\piBp \circ \phiB \circ \phi(a)) \phiH'( \psi') = (\piBp \circ \phi'(a')) \phiH'( \psi')$, so that $\phiH'$ is $\phi'$-compatible. Let $A$ and $B$ be strong $\phi$-compatible or $\phi$-compatible operators on $\hsA$ and $\hsB$ and define $A' \defeq \UA A \UA^{-1}$ and $B' \defeq \UB B \UB^{-1}$ on $\hsAp$ and $\hsBp$. In the strong $\phi$-compatibility case, one has $\phiH'(A' \psi') = \phiH'(\UA A \psi) = \UB \phiH(A \psi) = \UB B \phiH(\psi) = B' \UB \phiH(\psi) = B' \phiH'(\UA \psi) = B' \phiH'(\psi')$, so that $A'$ and $B'$ are strong $\phi'$-compatible. Applying this result to $\DAp$ and $\DBp$ (resp. $\JAp$ and $\JBp$, resp. $\gammaAp$ and $\gammaBp$ in the even case) shows that $(\algAp, \hsAp, \DAp, \JAp)$ and $(\algBp, \hsBp, \DBp, \JBp)$ are strong $\phi'$-compatible and similarly in the even case. In the $\phi$-compatibility case, since $\UB$ is diagonal, one has $B'^{\phi}_{\phi} = \UB^{\phi}_{\phi} B^{\phi}_{\phi} (\UB^{\phi}_{\phi})^{-1}$, and the conclusion follows in the same way.
\end{proof}

In the proof, the commutativity of the first diagram is only used when the representation $\piBp$ is applied, and more specifically, when this representation acts on $\phiH'(\hsAp)$. In other words, the minimal condition in this proof is that $\piBp \circ \phi' \circ \phiA = \piBp \circ \phiB \circ \phi$ holds as operators acting on $\phiH'(\hsAp) \subset \hsBp$.

\medskip
The map $\phi$ induces a natural map of graded algebras $\phi : \calT^\grast \algA \to \calT^\grast \algB$ by the relation $\phi(a^0 \otimes \cdots \otimes a^n) = \phi(a^0) \otimes \cdots \otimes \phi(a^n)$. If $\omega \in \Omega^1_U(\algA)$, then one can check that $\phi(\omega) \in \Omega^1_U(\algB)$, so that $\phi$ restricts to a map of graded algebras $\Omega^\grast_U(\algA) \to \Omega^\grast_U(\algB)$. If $\phi(\bbboneA) = \bbboneB$, then $\phi(\ddU a) = \phi(\bbboneA \otimes a - a \otimes \bbboneA) = \bbboneB \otimes \phi(a) - \phi(a) \otimes \bbboneB = \ddU \phi(b)$. If $\phi(\bbboneA) \neq \bbboneB$, let $p_\phi \defeq \phi(\bbboneA) \in \algB$ be the induced projection. Then $\phi(\ddU a) = p_\phi \otimes \phi(a) - \phi(a) \otimes p_\phi \in \Omega^1_U(\algB)$ can be written as $\phi(\ddU a) = p_\phi \ddU \phi(a) + \phi(a) \ddU (\bbboneB - p_\phi) = p_\phi \ddU \phi(a) - \phi(a) \ddU p_\phi$. This shows that $\phi$ is a homomorphism of differential algebras only when it is unital. In the following, we will use the most general relation $\phi(a^0 \ddU a^1) = \phi(a^0) \ddU \phi(a^1) - \phi(a^0 a^1) \ddU p_\phi$ since $\phi(a) p_\phi = \phi(a)$. 

 \begin{proposition}
 \label{prop universal forms and Dirac phi compatibility}
 Suppose that $\DB$ is $\phi$-compatible with $\DA$. 
 \begin{enumerate}
 \item For any $\omega \in \Omega^1_U(\algA)$, $\piDB\circ \phi(\omega)$ is $\phi$-compatible with $\piDA(\omega)$.
 
 \item Suppose that $\JB$ is strong $\phi$-compatible with $\JA$. For any unitaries $\uA \in \algA$ and $\uB \in \algB$ such that $\piA(\uA)$ and $\piB(\uB)$ are $\phi$-compatible and $\piB(\uB)$ is diagonal in the matrix decomposition, $\DB^{\uB}$ is $\phi$-compatible with $\DA^{\uA}$. \label{enum unitaries in algebras}
 
 \item Using the hypothesis of the previous points, $\DB[, \phi(\omega)]^{\uB}$ is $\phi$-compatible with $\DA[,\omega]^{\uA}$.
 \end{enumerate}
  \end{proposition}
  
Condition~\ref{enum unitaries in algebras} in this Proposition implies in particular that $\piBp \circ \phi' \circ \phiA = \piBp \circ \phiB \circ \phi$ (see comment after Prop.~\ref{prop unitary equi triple and (st) phi comp}) with $\algA' = \algA$, $\algB' = \algB$ and $\phi' = \phi$.
  
 \begin{proof}
We can reduce the general case to $\omega = a^0 \ddU a^1 \in \Omega^1_U(\algA)$. Let us then consider $\piDB\circ \phi(a^0 \ddU a^1) = \phi(a^0) [\DB, \phi(a^1)] - \phi(a^0 a^1) [\DB, p_\phi]$ (with $\piB$ omitted in this relation and the following). For any $\psi \in \hsA$, one has $\phi(a^0 a^1) [\DB, p_\phi] \phiH(\psi) = \phi(a^0 a^1) \DB \phiH(\psi) - \phi(a^0 a^1) \phi(\bbboneA) \DB \phiH(\psi) = 0$, since $p_\phi \phiH(\psi) = \phiH(\psi)$, so that $\piDB\circ \phi(a^0 \ddU a^1) \phiH(\psi) = \phi(a^0) [\DB, \phi(a^1)] \phiH(\psi)$. Using the matrix decomposition $\DB = \smallpmatrix{ \DB[,\phi]^{\phi} & \DB[,\phi]^{\perp} \\ \DB[,\perp]^{\phi} & \DB[,\perp]^{\perp} }$ and Point~\ref{item piBphi(a) diagonal} in Prop.~\ref{prop strong and not strong phi compatibility}, one gets
 \begin{align*}
\phi(a^0) [\DB, \phi(a^1)] 
 \begin{pmatrix} \phiH(\psi) \\  0 \end{pmatrix}
 &= 
 \begin{pmatrix}
 	\phi(a^0)_{\phi}^{\phi} [\DB[,\phi]^{\phi}, \phi(a^1)_{\phi}^{\phi}] \phiH(\psi)
 \\ 
	\phi(a^0)_{\perp}^{\perp} (\DB[,\perp]^{\phi} \phi(a^1)_{\phi}^{\phi} - \phi(a^1)_{\perp}^{\perp} \DB[,\perp]^{\phi} ) \phiH(\psi)
 \end{pmatrix}
 \end{align*}
 From this relation we get $\piDB (\phi(a^0 \ddU a^1))_{\phi}^{\phi} = \phi(a^0)_{\phi}^{\phi} [\DB[,\phi]^{\phi}, \phi(a^1)_{\phi}^{\phi}]$ and then $\piDB (\phi(a^0 \ddU a^1))_{\phi}^{\phi} \phiH(\psi) = \phi(a^0)_{\phi}^{\phi} [\DB[,\phi]^{\phi}, \phi(a^1)_{\phi}^{\phi}] \phiH(\psi) = \phiH( a^0 [\DA, a^1] \psi) = \phiH( \piDA(a^0 \ddU a^1) \psi)$ since $\DB$ is $\phi$-compatible with $\DA$.
 
Using the hypothesis that $\piB(\uB)$ is diagonal, a straightforward computation gives $(\piB(\uB)^\dagger [\DB, \piB(\uB)])_{\phi}^{\phi} = (\piB(\uB)^\dagger)_{\phi}^{\phi} [\DB[,\phi]^{\phi}, \piB(\uB)_{\phi}^{\phi} ]$ from which we deduce that $\piB(\uB)^\dagger [\DB, \piB(\uB)]$ is $\phi$-compatible with $\piA(\uA)^\dagger [\DA, \piA(\uA)]$. From Lemma~\ref{lemma JB JA st-phi-comp}, we deduce that $\JB \piB(\uB)^\dagger [\DB, \piB(\uB)] \JB^{-1}$ is $\phi$-compatible with $\JA \piA(\uA)^\dagger [\DA, \piA(\uA)] \JA^{-1}$, and so that $\DB^{\uB} = \DB + \piB(\uB)^\dagger [\DB, \piB(\uB)] + \epsilonB' \JB \piB(\uB)^\dagger [\DB, \piB(\uB)] \JB^{-1}$ is $\phi$-compatible with $\DA^{\uA} = \DA + \piA(\uA)^\dagger [\DA, \piA(\uA)] + \epsilonA' \JA \piA(\uA)^\dagger [\DA, \piA(\uA)] \JA^{-1}$ since $\epsilonB' = \epsilonA'$ by Prop~\ref{prop KO dim strong phi compatibility}.

The last point combines the two previous results by replacing $\DA$ by $\DA[,\omega] = \DA + \piDA(\omega) + \epsilonA' \JA \piDA(\omega) \JA^{-1}$ and $\DB$ by $\DB[, \phi(\omega)] = \DB + \piDB(\phi(\omega)) + \epsilonB' \JB \piDB(\phi(\omega)) \JB^{-1}$ which are $\phi$-compatible by the first point, Lemma~\ref{lemma JB JA st-phi-comp}, and Prop~\ref{prop KO dim strong phi compatibility}.
 \end{proof}

Notice that one can associate to any unitary $\uA \in \algA$ the diagonal (unitary) operator $\smallpmatrix{ \piB\circ \phi(\uA) & 0 \\ 0 & \bbbone_{\perp}^{\perp} }$ where $\piB\circ \phi(\uA) \phiH(\psi) \defeq \phiH(\piA(\uA) \psi)$ for any $\psi \in \hsA$ and $\bbbone_{\perp}^{\perp}$ is the identity operator on $\phiH(\hsA)^\perp$. But this operator is not necessarily of the form $\piB(\uB)$ for a unitary $\uB \in \algB$. In the case of $AF$-algebras, it will be possible to construct a unitary $\uB  \in \algB$ from $\uA$ such that $\piA(\uA)$ and $\piB(\uB)$ are (strong) $\phi$-compatible and $\piB(\uB)$ is diagonal, see Prop.~\ref{prop uB from uA case AF}.

A strong version of the previous proposition can be proposed, for which a proof is not necessary since it combines previous results and the same line of reasoning when computations are needed:
\begin{proposition}
Suppose that $\DB$ is strong $\phi$-compatible with $\DA$. 
 \begin{enumerate}
 \item For any $\omega \in \Omega^1_U(\algA)$, $\piDB\circ \phi(\omega)$ is strong $\phi$-compatible with $\piDA(\omega)$.
 
 \item Suppose that $\JB$ is strong $\phi$-compatible with $\JA$. For any unitaries $\uA \in \algA$ and $\uB \in \algB$ such that $\piA(\uA)$ and $\piB(\uB)$ are strong $\phi$-compatible, $\DB^{\uB}$ is strong $\phi$-compatible with $\DA^{\uA}$.
 
 \item Using the hypothesis of the previous points, $\DB[, \phi(\omega)]^{\uB}$ is strong $\phi$-compatible with $\DA[,\omega]^{\uA}$.
 \end{enumerate}
\end{proposition}

\subsection{Direct sums of algebras}
\label{sec direct sums og algebras}

Let us consider the more specific situation $\algA = \toplus_{i=1}^{r} \algA_i$ and $\algB = \toplus_{k=1}^{s} \algB_k$. We also suppose that there are (orthogonal) decompositions $\hsA = \toplus_{i=1}^{r} \hsiA[,i]$ and $\hsiB = \toplus_{k=1}^{s} \hsB[,k]$ such that the $\hsiA[,i]$ (resp. $\hsiB[,k]$) are Hilbert spaces on which $\algA_i$ (resp. $\algB_k$) are represented. In other words, the (left) module structures are compatibles with the direct sums of algebras and Hilbert spaces: for any $a = \toplus_{i=1}^{r} a_i \in \algA$ and $\psi = \toplus_{i=1}^{r} \psi_i \in \hsA$, one has $a \psi = \toplus_{i=1}^{r} a_i \psi_i$ (and similarly for $\algB$). 

Let $\injA^i : \algA_i \to \algA$  be the canonical inclusion and $\projA_i : \algA \to \algA_i$ be the canonical projection. With obvious notations, similar maps are defined for $\algB$, $\hsA$, and $\hsB$.

An operator $A$ on $\hsA$ can be decomposed along the operators $A_{j}^{i} \defeq \projHA_{j} \circ A \circ \injHA^{i} : \hsiA[,i] \to \hsiA[,j]$. The same holds for operators on $\hsB$. For computational purposes, we recall that one has
\begin{align*}
A \psi &= \toplus_{j=1}^{r} \big( \tsum_{i=1}^r  A_{j}^{i}(\psi_{i}) \big)
= \tsum_{i, j=1}^{r} \injHA^{j} \circ A_{j}^{i}(\psi_{i}).
\end{align*}
In the same way, a one-to-one homomorphism of algebras $\phi : \algA \to \algB$ decomposes along the maps $\phi_k^i \defeq \projB_k \circ \phi \circ \injA^i : \algA_i \to \algB_k$ and a morphism of Hilbert spaces $\phiH : \hsA \to \hsB$ decomposes along the $\phiH[,k]^i \defeq \projHB_k \circ \phiH \circ \injHA^i : \hsiA[,i] \to \hsiB[,k]$. One has
\begin{align*}
\phi(a) &= \toplus_{k=1}^{s} \big( \tsum_{i=1}^r  \phi_k^i(a_i) \big),
\quad \text{and} \quad
\phiH(\psi) = \toplus_{k=1}^{s} \big( \tsum_{i=1}^r  \phiH[,k]^i (\psi_i) \big).
\end{align*}
Notice also that $\phi(a a') = \phi(a) \phi(a')$ implies
\begin{align}
\label{eq phik product as sum ij}
\tsum_{i=1}^{r} \phi_k^i(a_i a'_i) = \tsum_{i,j=1}^{r} \phi_k^i(a_i) \phi_k^j(a'_j) \quad \text{for any $k=1, \dots, s$}
\end{align}

\begin{lemma}
The $\phi$-compatibility of $\phiH$ is equivalent to $\phiH[,k]^i (a_i \psi_i) = \phi_k^i(a_i) \phiH[,k]^i (\psi_i)$ for any $1 \leq i \leq r$, $1 \leq k \leq s$, $a_i \in \algA_i$ and $\psi_i \in \hsiA[,i]$.
\end{lemma}

\begin{proof}
One has $\phiH(a \psi) = \toplus_{k=1}^{s} \big( \tsum_{i=1}^r  \phi_k^i(a_i \psi_i) \big)$ and $\phi(a) \phiH(\psi) = \toplus_{k=1}^{s} \big( \tsum_{i=1}^r  \phi_k^i(a_i) \phiH[,k]^i (\psi_i) \big)$ so that $\phiH(a \psi) = \phi(a) \phiH(\psi)$ is equivalent to $\tsum_{i=1}^r  \phi_k^i(a_i \psi_i) = \tsum_{i=1}^r  \phi_k^i(a_i) \phiH[,k]^i (\psi_i)$ for any $k$. Taking $a_i$ and $\psi_i$ non-zero only for one value of $i$, this implies that $\phi_k^i(a_i \psi_i) = \phi_k^i(a_i) \phiH[,k]^i (\psi_i)$ for any $i$. Reciprocally, if this last equally is satisfied for any $i$, it implies the previous one by linearity.
\end{proof}

\begin{lemma}
Two operators $A$ on $\hsA$ and $B$ on $\hsB$ are strong $\phi$-compatible if and only if $\sum_{j=1}^{r} \phiH[,k]^{j} \circ A_{j}^{i} (\psi_i) = \tsum_{\ell=1}^{s} B_{k}^{\ell} \circ \phiH[,\ell]^{i} (\psi_i)$ for any $1 \leq i \leq r$, $1 \leq k \leq s$, and $\psi_i \in \hsiA[,i]$.

Two operators $A$ on $\hsA$ and $B$ on $\hsB$ are $\phi$-compatible if and only if $\sum_{j=1}^{r} \phiH[,k]^{j} \circ A_{j}^{i} (\psi_i) = \tsum_{\ell=1}^{s} B_{\phi, k}^{\phi, \ell} \circ \phiH[,\ell]^{i} (\psi_i)$ for any $1 \leq i \leq r$, $1 \leq k \leq s$, and $\psi_i \in \hsiA[,i]$.
\end{lemma}

\begin{proof}
On the one hand, one has $\phiH(A \psi) = \toplus_{k=1}^{s} \big( \sum_{i,j=1}^{r} \phiH[,k]^{j} \circ A_{j}^{i} (\psi_i) \big)$ and on the other hand $B \phiH(\psi) = \toplus_{k=1}^{s} \big( \tsum_{\ell=1}^{s} \tsum_{i=1}^{r} B_{k}^{\ell} \circ \phiH[,\ell]^{i} (\psi_i) \big)$. So, the relation $\phiH(A \psi) = B \phiH(\psi)$ is equivalent to $\sum_{i,j=1}^{r} \phiH[,k]^{j} \circ A_{j}^{i} (\psi_i) = \tsum_{\ell=1}^{s} \tsum_{i=1}^{r} B_{k}^{\ell} \circ \phiH[,\ell]^{i} (\psi_i)$ for any $k$. Taking $\psi_i$ non-zero only for one value of $i$, this implies $\sum_{j=1}^{r} \phiH[,k]^{j} \circ A_{j}^{i} (\psi_i) = \tsum_{\ell=1}^{s} B_{k}^{\ell} \circ \phiH[,\ell]^{i} (\psi_i)$ for any $i$ and $k$. By linearity, this relation implies the previous one.

Concerning the $\phi$-compatibility, one can replace $B$ by $\widehat{B}_{\phi}^{\phi}$ in the previous result. Since $\widehat{B}_{\phi}^{\phi}$ acts only on $\phiH(\hsA)$, one can replace $\widehat{B}_{\phi, k}^{\phi, \ell}$ by the operators $B_{\phi, k}^{\phi, \ell} : \projHB_{\ell} \circ \phiH(\hsA) \to \projHB_{k} \circ \phiH(\hsA)$ in the final relation.
\end{proof}

\medskip
We can extend the maps $\phi_k^i$ as $\phi_{k_0, \dots, k_n}^{i_0, \dots, i_n} : \algA^\otimes_{i_0, \dots, i_{n}} \to \algB^\otimes_{k_0, \dots, k_{n}}$ by $\phi_{k_0, \dots, k_n}^{i_0, \dots, i_n}(a^0_{i_0} \otimes \cdots \otimes a^n_{i_n}) \defeq \phi_{k_0}^{i_0}(a^0_{i_0}) \otimes  \cdots \otimes \phi_{k_n}^{i_n}(a^n_{i_n})$ for any $i_0, \dots, i_n$ and $k_0, \dots, k_n$, and then we define maps $\hphi : \kT^n \algA \to \kT^n \algB$, for any $n \geq 1$, by $\toplus_{i_1, \dots, i_{n-1} = 1}^{r} \big( a^0_{i} \otimes a^1_{i_1} \otimes \cdots \otimes a^{n-1}_{i_{n-1}} \otimes a^n_{j} \big)_{i,j=1}^{r} \mapsto \toplus_{k_1, \dots, k_{n-1} = 1}^{s} \big( \tsum_{i_1, \dots, i_{n-1}=1}^{r} \phi_{k, k_1, \dots, k_{n-1}, \ell}^{i, i_1, \dots, i_{n-1}, j}( a^0_{i} \otimes a^1_{i_1} \otimes \cdots \otimes a^{n-1}_{i_{n-1}} \otimes a^n_{j}) \big)_{k,\ell=1}^{s}$, and, for $n=0$, the diagonal matrix with entries $a_i$ at $(i,i)$ is sent to the diagonal matrix with entries $\tsum_{i=1}^{r} \phi_{k}^{i}(a_i)$ at $(k,k)$. Using \eqref{eq product kT n np} and  \eqref{eq phik product as sum ij}, one can check that $\hphi : \kT^\grast \algA \to \kT^\grast \algB$ is a homomorphism of graded algebras and that $\hphi(\bOmega^1_U(\algA)) \subset \bOmega^1_U(\algB)$, so that $\hphi : \bOmega^\grast_U(\algA) \to \bOmega^\grast_U(\algB)$ is a homomorphism of graded algebras. Obviously, these properties are consequences of the general situation described in Sect.~\ref{sec general situations}.

\subsection{\texorpdfstring{$AF$-algebras}{AF-algbras}}
\label{sec AF algebras}

We consider now the special case of sums of matrix algebras, $\algA = \toplus_{i=1}^{r} M_{n_i}$ and $\algB = \toplus_{k=1}^{s} M_{m_k}$. We use similar notations to the ones in \cite{MassNieu21q}. Let us introduce the projection and injection maps $\projA_i$, $\projB_k$, $\injA^i$ and $\injB^k$. Let  $\phi : \algA = \toplus_{i=1}^{r} M_{n_i} \to \algB = \toplus_{k=1}^{s} M_{m_k}$ be a one-to-one homomorphism. It is taken in its simplest form, and we normalize it such that, for any $a = \toplus_{i=1}^{r} a_i$,
\begin{align}
\label{eq phi-k(a)}
\phi_k (a) \defeq \projB_k\circ \phi (a)
&= 
\begin{pmatrix}
a_1 \otimes \bbbone_{\alpha_{k1}} & 0 & \cdots & 0 & 0\\
0 & a_2 \otimes \bbbone_{\alpha_{k2}} & \cdots & 0 & 0\\
\vdots & \vdots & \ddots & \vdots & \vdots\\
0 & 0 & \cdots & a_r \otimes \bbbone_{\alpha_{kr}} & 0 \\
0 & 0 & \cdots & 0 & \bbbzero_{n_{0,k}}
\end{pmatrix}
\end{align}
where the integers $\alpha_{ki} \geq 0$ are the multiplicities of the inclusions of $M_{n_i}$ into $M_{m_k}$,  $\bbbzero_{n_{0,k}}$ is the $n_{0,k} \times n_{0,k}$ zero matrix such that $n_{0,k} \geq 0$ satisfies $m_k = n_{0,k} + \tsum_{i=1}^{r} \alpha_{ki} n_i$, and 
\begin{align*}
a_i \otimes \bbbone_{\alpha_{ki}}
&= \left.
\begin{pmatrix}
a_i & 0 & 0 & 0 \\
0 & a_i & 0 & 0 \\
\vdots & \vdots & \ddots & \vdots \\
0 & 0 & \cdots & a_i \\
\end{pmatrix}
\right\} \text{$\alpha_{ki}$ times.}
\end{align*}
We define the maps $\phi_{k}^{i} \defeq \phi_{k} \circ \injA^i : M_{n_i} \to M_{m_k}$, which take the explicit form
\begin{align}
\label{eq phi-k-i(ai)}
\phi_{k}^{i}(a_i)
&= \begin{pmatrix}
0  & \cdots & 0 & 0 & 0 & \cdots & 0 \\
\vdots  & \ddots & \vdots & \vdots  & \vdots & \cdots & \vdots \\
0  & \cdots & 0 & 0 & 0 & \cdots & 0 \\
0  & \cdots & 0 & a_i \otimes \bbbone_{\alpha_{ki}} & 0 & \cdots & 0 \\
0  & \cdots & 0 & 0 & 0 & \cdots & 0 \\
\vdots  & \vdots & \vdots & \vdots  & \vdots & \ddots & \vdots \\
0  & \cdots & 0 & 0 & 0 & \cdots & 0
\end{pmatrix}
\end{align}
The maps $\phi_{k}^{i}$ satisfy a stronger relation than \eqref{eq phik product as sum ij}: for any $i,j = 1, \dots, r$ and $k=1, \dots, s$, 
\begin{align*}
\phi_k^i(a_i) \phi_k^j(a'_j)
&= \begin{cases}
0 & \text{if $i \neq j$}\\
\phi_k^i(a_i a'_i) & \text{if $i = j$}.
\end{cases}
\end{align*}

If $\algA = \algA_{n}$ and $\algB = \algA_{n+1}$ for a $AF$-algebra $\varinjlim \algA_n$, then the multiplicities $\alpha_{ki}$ define the Bratteli diagram of this $AF$-algebra and vice versa. The integers $n_{0,k}$ are defined by complementarity at each step.

When $\alpha_{ki} > 0$, for $1 \leq \alpha \leq \alpha_{ki}$ we define the maps $\phi_{k, \alpha}^{i} : M_{n_i} \to M_{m_k}$ which insert $a_i$ at the $\alpha$-th entry on the diagonal of $\bbbone_{\alpha_{ki}}$ in the previous expression, so that $a_i$ appears only once on the RHS. The maps $\phi_{k}$, $\phi_{k}^{i}$, and $\phi_{k, \alpha}^{i}$ are homomorphisms of algebras and one has
\begin{align}
\phi &= \toplus_{k=1}^{s} \phi_k : \toplus_{i=1}^{r} M_{n_i} \to \toplus_{k=1}^{s} M_{m_k},
\nonumber
\\
\label{eq decompositions phi}
\phi_k &= \tsum_{i=1}^{r} \phi_{k}^{i} \circ \projA_i : \toplus_{i=1}^{r} M_{n_i} \to M_{m_k},
\\
\phi_{k}^{i} &= \tsum_{\alpha=1}^{\alpha_{ki}} \phi_{k, \alpha}^{i} : M_{n_i} \to M_{m_k}.
\nonumber
\end{align}
Notice then that $\phi_k(\bbboneA) = \tsum_{i=1}^{r} \tsum_{\alpha=1}^{\alpha_{ki}} \phi_{k, \alpha}^{i}(\bbbone_{\algA_i})$ fills the diagonal of $M_{m_k}$ with $\tsum_{i=1}^{r} \alpha_{ki} n_i$ copies of $1$ except for the last $n_{0,k}$ entries. When $n_{0,k} = 0$, one gets $\phi_k(\bbboneA) = \bbbone_{\algB_k}$, otherwise, let 
\begin{align}
\label{eq def pn0k}
p_{n_{0,k}} \defeq \bbbone_{m_k} - \phi_k(\bbboneA) \in M_{m_k}
\quad \text{ and } \quad
p_{n_0} \defeq \toplus_{k=1}^{s} p_{n_{0,k}} \in \algB.
\end{align}
The $p_{n_{0,k}}$'s are diagonal matrices with zero entries except for the last $n_{0,k}$ diagonal entries (bottom right) which are equal to $1$.

\medskip
We will use the results in Sect.~\ref{sec normal form finite real spectral triples}, in particular the diagrammatic descriptions of (odd/even) real spectral triples. Let $(\algA, \hsA, \DA, \JA)$ and $(\algB, \hsB, \DB, \JB)$ be two real spectral triples on the algebras $\algA = \toplus_{i=1}^{r} M_{n_i}$ and $\algB = \toplus_{k=1}^{s} M_{m_k}$ with $\hsA = \toplus_{v \in \GammaA^{(0)}} \hsiA[,v]$ and $\hsB = \toplus_{w \in \GammaB^{(0)}} \hsiB[,w]$. As we defined the maps $i,j$ on $\GammaA^{(0)}$, let us define the similar maps $k, \ell$ on $\GammaB^{(0)}$: for any $w \in \GammaB^{(0)}$ with $\pi_{\lambda\rho}(w) = (\bem_k, \bem_\ell)$, $k(w) \defeq k$ and $\ell(w) \defeq \ell$.

Let $\phiH : \hsA \to \hsB$ be a $\phi$-compatible linear map of bimodules ($\phi^{e}$-compatible as left modules). This map decomposes along the maps $\phiH[,w]^v : \hsiA[,v] \to \hsiB[,w]$ between irreps on both sides. For any $a = \toplus_{i=1}^{r} a_i \in \algA$, $b = \toplus_{i=1}^{r} b_i \in \algA$, and $\psi = \toplus_{v \in \GammaA^{(0)}} \psi_v$, one has $\phiH(a b ^\circ \psi) = \phi(a) \phi(b)^\circ \phiH(\psi)$ with (using \eqref{eq pi(a) decomposition along Hv}) $\phiH(a b ^\circ \psi) = \sum_{v \in \GammaA^{(0)}} \toplus_{w \in \GammaB^{(0)}} \phiH[,w]^{v} ( a_{i(v)} b_{j(v)}^\circ \psi_v)$ and $\phi(a) \phi(b)^\circ \phiH(\psi) = \sum_{v \in \GammaA^{(0)}} \toplus_{w \in \GammaB^{(0)}} \phi(a)_{k(w)} \phi(b)_{\ell(w)}^\circ \phiH[,w]^{v} (\psi_v)$. We can select a fixed $v$ and choose $\psi$ with only one non-zero component $\psi_v$. Then one gets, for any $v \in  \GammaA^{(0)}$ and any $w \in \GammaB^{(0)}$, $\phiH[,v]^{w} ( a_{i(v)} b_{j(v)}^\circ \psi_v) = \phi(a)_{k(w)} \phi(a)_{\ell(w)}^\circ \phiH[,v]^{w} (\psi_v)$. Let us now consider a fixed index $i$ and $a$ with only non-zero component at $i(v) = i$, and the same for a fixed $j$ and $b$: with $(k, \ell) = (k(w), \ell(w))$, one has
\begin{align}
\label{eq phiH w v a b psiv}
\phiH[,w]^{v}(a_i b_j^\circ \psi_v) = \phi_{k}^{i}(a_i) \phi_{\ell}^{j}(b_j)^\circ \phiH[,w]^{v}(\psi_v)
\end{align}
This relation, combined with \eqref{eq phi-k-i(ai)}, suggests to decompose $\bbC^{m_k}$ in $\hsiB[,w] = \bbC^{m_k} \otimes \bbC^{m_\ell \circ}$ as $\bbC^{m_k} =  [ \toplus_{i=1}^{r} \bbC^{n_i} \otimes \bbC^{\alpha_{ki}} ] \oplus \bbC^{n_{0,k}}$ and similarly for $\bbC^{m_\ell \circ}$ with a last term $\bbC^{n_{0,\ell}}$, so that one has the orthogonal decomposition
\begin{align}
\label{eq decomp hBw}
\hsB[,w] = \bbC^{m_k} \otimes \bbC^{m_\ell \circ} 
&= \begin{multlined}[t]
[ \toplus_{i,j = 1}^{r} \bbC^{n_i} \otimes \bbC^{\alpha_{ki}} \otimes \bbC^{\alpha_{\ell j}} \otimes \bbC^{n_j \circ} ] 
\\
\oplus [ \toplus_{i=1}^{r}  \bbC^{n_i} \otimes \bbC^{\alpha_{ki}} \otimes \bbC^{n_{0,\ell} \circ} ] 
\\
\oplus [ \toplus_{j=1}^{r} \bbC^{n_{0,k}} \otimes \bbC^{\alpha_{\ell j}} \otimes \bbC^{n_j \circ} ] 
\\
\oplus [ \bbC^{n_{0,k}} \otimes \bbC^{n_{0,\ell} \circ}].
\end{multlined}
\end{align}
For any $i,j = 1, \dots, r$ and $k, \ell = 1, \dots, s$, let us define the inclusion
\begin{align*}
I_{k,\ell}^{i,j} : \bbC^{n_i} \otimes \bbC^{\alpha_{ki}} \otimes \bbC^{\alpha_{\ell j}} \otimes \bbC^{n_j \circ} \hookrightarrow \bbC^{m_k} \otimes \bbC^{m_\ell \circ}.
\end{align*}
Notice that $I_{k,\ell}^{i,j} = I_{k}^{i} \otimes I_{\ell}^{j \circ}$ with the inclusions $I_{k}^{i} : \bbC^{n_i} \otimes \bbC^{\alpha_{ki}} \hookrightarrow \bbC^{m_k}$ and $I_{\ell}^{j \circ} : \bbC^{\alpha_{\ell j}} \otimes \bbC^{n_j \circ} \hookrightarrow \bbC^{m_\ell \circ}$.

Let $F^{(i,k), (j, \ell)}_\algA : \bbC^{n_i} \otimes \bbC^{\alpha_{ki}} \otimes \bbC^{\alpha_{\ell j}} \otimes \bbC^{n_j \circ} \to \bbC^{n_j}  \otimes \bbC^{\alpha_{\ell j}} \otimes \bbC^{\alpha_{ki}} \otimes \bbC^{n_i \circ}$ be the involution $\xi_i \otimes \vm^k_i \otimes \vm^\ell_j \otimes \eta_j^\circ \mapsto \eta_j \otimes \vm^\ell_j \otimes \vm^k_i \otimes \xi_i^\circ$ and $F^{k\ell}_\algB : \bbC^{m_k} \otimes \bbC^{m_\ell \circ} \to \bbC^{m_\ell} \otimes \bbC^{m_k \circ}$ the involution $\varphi_k \otimes \vartheta_\ell^\circ \mapsto \vartheta_\ell \otimes \varphi_k^\circ$. Then, one can check that
\begin{align}
\label{eq flips inclusions}
F^{k\ell}_\algB \circ I_{k,\ell}^{i,j}
&= I_{\ell,k}^{j,i} \circ F^{(i,k), (j, \ell)}_\algA
\end{align}
Notice that $\JA :  \hsA[,v] \to \hsA[, \JimA(v)]$ can be written as $\JA = \epsilonA(v,d_\algA) (J_0 \otimes J_0) \circ F^{i j}_\algA$ since $F^{i j}_\algA = \hJimA[,v]$ with $(i, j) = (i(v), j(v))$.

In the case of a $AF$-algebra, the inclusions $I_{k}^{i}$ (and so $I_{\ell}^{j \circ}$ and $I_{k,\ell}^{i,j}$) are defined directly from the Bratteli diagram of the algebra, that is, they depend only on the one-to-one homomorphism $\phi : \algA \to \algB$. We can now write $\phiH$ in terms of these inclusions.

Combining \eqref{eq phiH w v a b psiv} and \eqref{eq decomp hBw}, the map $\phiH[,w]^{v}$ first reduces to a map $\bbC^{n_i} \otimes \bbC^{n_j \circ} \to \bbC^{n_i} \otimes \bbC^{\alpha_{ki}} \otimes \bbC^{\alpha_{\ell j}} \otimes \bbC^{n_j \circ}$, and using a slight adaptation of Lemma~\ref{lemma technical module endo reduction}, it reduces to a linear map $\bbC \to \bbC^{\alpha_{ki}} \otimes \bbC^{\alpha_{\ell j}}$, that is, to the data of an element $u(v,w) \in \bbC^{\alpha_{ki}} \otimes \bbC^{\alpha_{\ell j}}$, such that $\phiH[,w]^{v}$ is the composition of $\bbC^{n_i} \otimes \bbC^{n_j \circ} \ni \xi_i \otimes \eta_j^\circ \mapsto \xi_i \otimes u(v,w) \otimes \eta_j^\circ \in \bbC^{n_i} \otimes \bbC^{\alpha_{ki}} \otimes \bbC^{\alpha_{\ell j}} \otimes \bbC^{n_j \circ}$ with the inclusion $I_{k,\ell}^{i,j}$. When $\alpha_{ki} = 0$ or $\alpha_{\ell j} = 0$, $\phiH[,w]^{v} = 0$.

Consequently, the $\phi$-compatible map $\phiH$ is completely determined by a family of matrices $u(v,w) \in M_{\alpha_{k i} \times \alpha_{\ell j}} \simeq \bbC^{\alpha_{k i}} \otimes \bbC^{\alpha_{\ell j}}$%
\footnote{\label{fn z=xyT}We use the following convention. Let $x, x' \in \bbC^n$ and $y, y' \in \bbC^m$. Define $\bbC^n \otimes \bbC^m \ni x \otimes y \simeq z = x y^\top \in M_{n \times m}$ and $z' = x' y'^\top \in M_{n \times m}$, so that $z(v) = \langle \bar{y}, v\rangle_{\bbC^m} x$ for any $v \in \bbC^m$. One then gets $\bar{y} \otimes \bar{x} \simeq z^\ast$ and $\tr(z^\ast z') = \langle x, x' \rangle_{\bbC^n} \langle y, y' \rangle_{\bbC^m}$, and by linearity, a similar relation holds for $z = \sum_{i} x_i y_i^\top$ and $z' = \sum_{i} x'_i y_i'^\top$. This relation will be used in the following.}
 by the previous relation, with $(i, j) = (i(v), j(v))$ and $(k, \ell) = (k(w), \ell(w))$. Notice that for $v, v' \in \GammaA^{(0)}$ such that $\pi_{\lambda\rho}(v) = \pi_{\lambda\rho}(v')$, the ranges of $\phiH[,w]^{v}$ and $\phiH[,w]^{v'}$ are at the same place in $\hsiB[,w]$ (the range of $I_{k,\ell}^{i,j}$), and $u(v,w)$ and $u(v',w)$ define a kind of relative positioning and weight between the two ranges. If $\pi_{\lambda\rho}(v) \neq \pi_{\lambda\rho}(v')$, the ranges are orthogonal in $\hsiB[,w]$ since the ranges of $I_{k,\ell}^{i,j}$ and $I_{k,\ell}^{i',j'}$ are distinct in the orthogonal decomposition \eqref{eq decomp hBw} when $(i,j) \neq (i',j')$.

\begin{remark}
For non-real spectral triples, a similar (simpler) result can be obtained: a $\phi$-compatible map $\phiH : \hsA \to \hsB$ is completely determined by the linear maps $\phiH[,w]^{v} : \hsiA[,v] = \bbC^{n_i} \to \hsiB[,w] = \bbC^{m_k}$ for $i=i(v)$ and $k=k(w)$, and so by a family of vectors $u(v,w) \in \bbC^{\alpha_{k i}}$ such that $\phiH[,w]^{v}$ is the composition of $\bbC^{n_i} \ni \xi_i \mapsto \xi_i \otimes u(v,w) \in \bbC^{n_i} \otimes \bbC^{\alpha_{ki}}$ with the inclusion $I_{k}^{i} : \bbC^{n_i} \otimes \bbC^{\alpha_{ki}} \hookrightarrow \bbC^{m_k}$.
\end{remark}

The following result summarizes the construction so far:
\begin{lemma}
\label{lemma matrices u(v,w)}
There is a family of matrices $u(v,w) \in M_{\alpha_{k i} \times \alpha_{\ell j}}$ such that, for any $v \in \GammaA^{(0)}$ and $w \in \GammaB^{(0)}$, with $(i, j) = (i(v), j(v))$ and $(k, \ell) = (k(w), \ell(w))$, one has
\begin{align}
\label{eq phiH Iklij xi eta}
\phiH[,w]^{v}( \xi_i \otimes \eta_j^\circ) 
&= I_{k,\ell}^{i,j}( \xi_i \otimes u(v,w) \otimes \eta_j^\circ)
\quad \text{for any $\xi_i \otimes \eta_j^\circ \in \hsA[, v]$.}
\end{align}
For any $v \in \GammaA^{(0)}$, any $w \in \GammaB^{(0)}$, and any $a = \toplus_{i=1}^{r} a_i \in \algA$, one has
\begin{align*}
\phi_{k}^{i}(a_i) I_{k,\ell}^{i,j} ( \xi_i \otimes u(v,w) \otimes \eta_j^\circ) 
&= I_{k,\ell}^{i,j} ( a_i \xi_i \otimes u(v,w) \otimes \eta_j^\circ)
\end{align*}
with $(i, j) = (i(v), j(v))$ and $(k, \ell) = (k(w), \ell(w))$.
 
 In the even case, if $\gammaB$ is $\phi$-compatible with $\gammaA$, then $\phiH[,w]^{v}$, and so $u(v,w)$, can be non-zero only when $s(v) = s(w)$. 
\end{lemma}

\begin{proof}
The first statement is a summary of the previous decomposition of $\phiH$. The second relation is a direct consequence of this decomposition of $\phiH$ and \eqref{eq phiH w v a b psiv}. Notice that in the LHS, one could replace $\phi_{k}^{i}(a_i)$ by $\phi_{k}(a)$ since only the entries positioned according to $i$ in the matrix $\phi_{k}(a) \in M_{m_k}$, see \eqref{eq phi-k(a)} and \eqref{eq phi-k-i(ai)}, give non-zero contributions once applied on the range of $I_{k,\ell}^{i,j}$. In the even case, the statement is a consequence of Lemma~\ref{lemma gammaB gammaA phi compatibiliy diagonal}, which implies here that $\phiH[,w]^{v} = 0$ when $s(v) \neq s(w)$.
\end{proof}

\begin{proposition}
\label{prop st phi comp and phi comp of operators for AF}
Two operators $A$ on $\hsA$ and $B$ on $\hsB$ are strong $\phi$-compatible if and only if
\begin{align}
\label{eq A B str phi comp AF}
\tsum_{v_2 \in \GammaA^{(0)}}\phiH[,w_2]^{v_2}(A_{v_2}^{v_1} \psi_{v_1})
&=
\tsum_{w_1 \in \GammaB^{(0)}} B_{w_2}^{w_1} \phiH[,w_1]^{v_1}(\psi_{v_1})
\end{align}
for any $v_1 \in \GammaA^{(0)}$, $w_2 \in \GammaB^{(0)}$, and $\psi_{v_1} \in \hsiA[,v_1]$. They are $\phi$-compatible if and only if
\begin{align*}
\tsum_{v_2 \in \GammaA^{(0)}}\phiH[,w_2]^{v_2}(A_{v_2}^{v_1} \psi_{v_1})
&=
\tsum_{w_1 \in \GammaB^{(0)}} B_{\phi, w_2}^{\phi, w_1} \phiH[,w_1]^{v_1}(\psi_{v_1})
\end{align*}
for any $v_1 \in \GammaA^{(0)}$, $w_2 \in \GammaB^{(0)}$, and $\psi_{v_1} \in \hsiA[,v_1]$, where $B_{\phi, w_2}^{\phi, w_1} : \hsiB[,w_1] \cap \phiH(\hsA) \to \hsiB[,w_2] \cap \phiH(\hsA)$ is the decomposition of $B_{\phi}^{\phi}$ along the $\hsiB[,w] \cap \phiH(\hsA)$'s.
\end{proposition}

\begin{proof}
For any $\psi = \toplus_{v_1 \in \GammaA^{(0)}} \psi_{v_1} \in \hsA$, one has $A \psi = \toplus_{v_2 \in \GammaA^{(0)}} \sum_{v_1 \in \GammaA^{(0)}} A_{v_2}^{v_1} \psi_{v_1}$ and $\phiH(\psi) = \toplus_{w_1 \in \GammaB^{(0)}} \sum_{v_1 \in \GammaA^{(0)}} \phiH[,w_1]^{v_1}(\psi_{v_1})$, so that $\phiH(A \psi) = \toplus_{w_2 \in \GammaB^{(0)}} \sum_{v_1, v_2 \in \GammaA^{(0)}} \phiH[,w_2]^{v_2}( A_{v_2}^{v_1}  \psi_{v_1})$. In a similar way, $B \phiH(\psi) = \toplus_{w_2 \in \GammaB^{(0)}} \sum_{w_1 \in \GammaB^{(0)}} \sum_{v_1 \in \GammaA^{(0)}} B_{w_1}^{w_2} \phiH[,w_2]^{v_1}(\psi_{v_1})$. 

The strong $\phi$-compatibility is then equivalent to $\sum_{v_1, v_2 \in \GammaA^{(0)}} \phiH[,w_2]^{v_2}( A_{v_2}^{v_1}  \psi_{v_1}) = \sum_{w_1 \in \GammaB^{(0)}} \sum_{v_1 \in \GammaA^{(0)}} B_{w_1}^{w_2} \phiH[,w_2]^{v_1}(\psi_{v_1})$ for any $w_2 \in \GammaB^{(0)}$, and, by linearity (fixing $\psi$ with one non zero component at $v_1$), $\sum_{v_2 \in \GammaA^{(0)}} \phiH[,w_2]^{v_2}( A_{v_2}^{v_1}  \psi_{v_1}) = \sum_{w_1 \in \GammaB^{(0)}}  B_{w_1}^{w_2} \phiH[,w_2]^{v_1}(\psi_{v_1})$ for any $v_1 \in \GammaA^{(0)}$ and $w_2 \in \GammaB^{(0)}$.

The $\phi$-compatibility relation follows the same computation with $B$ replaced by $B_{\phi}^{\phi}$.
\end{proof}

\begin{lemma}
For any $a = \toplus_{i=1}^{r} a_i \in \algA$ and $b = \toplus_{k=1}^{s} b_i \in \algB$, $\piA(a)$ and $\piB(b)$ are strong $\phi$-compatible if and only if, for any $v \in \GammaA^{(0)}$, any $w \in \GammaB^{(0)}$, and any $\xi_{i(v)} \otimes \eta_{j(v)}^\circ \in \hsA[, v]$, one has
\begin{align*}
b_{k(w)} I_{k,\ell}^{i,j} ( \xi_{i(v)} \otimes u(v,w) \otimes \eta_{j(v)}^\circ) 
&= I_{k,\ell}^{i,j} ( a_{i(v)} \xi_{i(v)} \otimes u(v,w) \otimes \eta_{j(v)}^\circ) 
\end{align*}
\end{lemma}

\begin{proof}
Inserting \eqref{eq pi(a) decomposition along Hv} into \eqref{eq A B str phi comp AF} for $A = \piA(a)$ and $B = \piB(b)$, one gets $\phiH[,w]^{v}( a_{i(v)} \psi_v) = b_{k(w)} \phiH[,w]^{v}(\psi_v)$ for any $v \in \GammaA^{(0)}$ and $w \in \GammaB^{(0)}$. Using \eqref{eq phiH Iklij xi eta} for $\phiH[,w]^{v}$ then gives the relation.
\end{proof}

\begin{proposition}
\label{prop uB from uA case AF}
Let $\uA \in \algA$ be a unitary element and define $\uB \defeq \phi(\uA) + p_{n_{0}} \in \algB$ (see \eqref{eq def pn0k}). Then $\uB$ is a unitary element such that $\piB(\uB)$ is diagonal (in the orthogonal decomposition defined by $\phiH$) and is strong $\phi$-compatible with $\piA(\uA)$.
\end{proposition}

\begin{proof}
One already knows that $\piB \circ \phi(\uA)$ is strong $\phi$-compatible with $\piA(\uA)$ (see Remark~\ref{rmk piA(a) st phi comp piB(phi(a))}). By construction, the range of $\phiH[,w]$ is contained only in the first term in brackets (the double direct sum over $i,j$) in \eqref{eq decomp hBw}, while $\piB(p_{n_0})$ is non-trivial only on the last two terms (the ones with $\bbC^{n_{0,k}}$ as first factor). This implies that $\piB(p_{n_0}) \phiH(\psi) = 0 = \phiH(\uA \psi)$ for any $\psi \in \hsA$. 
So, one has $\piB(\uB) \phiH(\psi) = \phiH(\uA \psi)$ for any $\psi \in \hsA$, and since $\piA(\uA)$ and $\piB(\uB)$ are unitary, by Prop.~\ref{prop strong and not strong phi compatibility}, $\piB(\uB)$ is diagonal.
\end{proof}

\begin{proposition}
\label{prop JB JA strong phi compatibiliy relation on ukappa}
$\JB$ is strong $\phi$-compatible with $\JA$ if and only if
\begin{align}
\label{eq ukappa epsilon ustar}
u(\JimA(v), \JimB(w)) &= \frac{\epsilonA(v,d_\algA)}{\epsilonB(w,d_\algB)} u(v,w)^\ast
\end{align}
for any $v \in \GammaA^{(0)}$ and $w \in \GammaB^{(0)}$ where $d_\algA$ (resp. $d_\algB$) is the $KO$-dimension of $\algA$ (resp. $\algB)$.
\end{proposition}

Prop.~\ref{prop KO dim AF phi compatibility} below gives a criterion on spectral triples on top of $\algA$ and $\algB$ so that $d_\algA = d_\algB$.

\begin{proof}
For any $\psi_v = \xi_i \otimes \eta_j^\circ \in \hsiA[,v]$, one has $\phiH[,\JimB(w)]^{\JimA(v)}(\JA \psi_v) = \epsilonA(v,d_\algA) \phiH[,\JimB(w)]^{\JimA(v)}(\Beta_{j} \otimes \Bxi_{i}^\circ) = \epsilonA(v,d_\algA) I_{\ell,k}^{j,i}\big( \Beta_{j} \otimes u(\JimA(v), \JimB(w)) \otimes \Bxi_{i}^\circ \big)$ and $\JB \phiH[,w]^{v}(\psi_v) = \JB \circ  I_{k,\ell}^{i,j}(\xi_i \otimes u(v,w) \otimes \eta_j^\circ) = \epsilonB(w,d_\algB) I_{\ell,k}^{j,i} \big( \Beta_{j} \otimes u(v,w)^\ast \otimes \Bxi_{i}^\circ \big)$ when one uses \eqref{eq flips inclusions} and the identification of $M_{\alpha_{k i} \times \alpha_{\ell j}}$ with $\bbC^{\alpha_{k i}} \otimes \bbC^{\alpha_{\ell j}}$ (see Footnote~\ref{fn z=xyT}). This implies the required equivalence.
\end{proof}

\begin{corollary}
\label{corr phiH and phiHkappa}
If $\JB$ is strong $\phi$-compatible with $\JA$, then, for any $v \in \GammaA^{(0)}$ and $w \in \GammaB^{(0)}$, $\phiH[,\JimB(w)]^{\JimA(v)} \neq 0$ if and only if $\phiH[,w]^{v} \neq 0$.
\end{corollary}

\begin{proof}
For any $v \in \GammaA^{(0)}$ and $w \in \GammaB^{(0)}$, with $(i, j) = (i(v), j(v))$ and $(k, \ell) = (k(w), \ell(w))$, from \eqref{eq ukappa epsilon ustar}, one has
\begin{align*}
\phiH[,\JimB(w)]^{\JimA(v)}(\xi_j \otimes \eta_i^\circ)
&= I_{\ell, k}^{j,i} (\xi_j \otimes u(\JimA(v), \JimB(w)) \otimes \eta_i^\circ)
\\
&= \frac{\epsilonA(v,d_\algA)}{\epsilonB(w,d_\algB)} I_{\ell, k}^{j,i} (\xi_j \otimes u(v,w)^\ast \otimes \eta_i^\circ).
\end{align*}
We then get the equivalence since $u(v,w)$ defines $\phiH[,w]^{v}$.
\end{proof}

\begin{figure}
{\centering
\subfloat[A Bratteli diagram]{\includegraphics[]{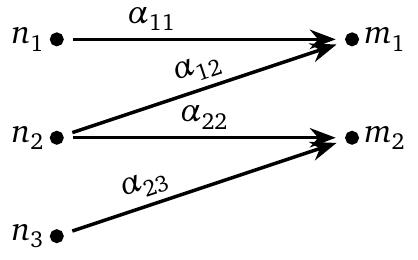}\label{fig bratteliDiagram}}\\
\subfloat[A lifting of a Bratteli diagram between two Krajewski diagrams]{\includegraphics[]{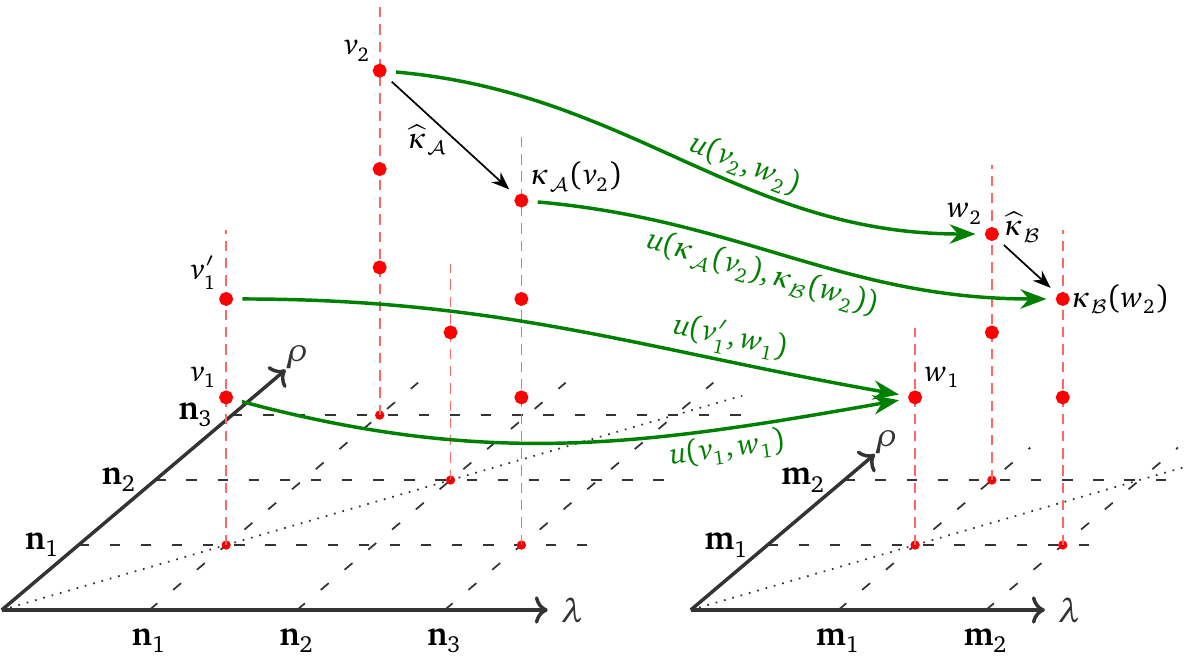}\label{fig liftingBratteliDiagram}}
\par}
\caption{Lifting of a Bratteli diagram between two Krajewski diagrams.\\
\protect\subref{fig bratteliDiagram}: An example of a Bratteli diagram for the inclusion $M_{n_1} \oplus M_{n_2} \oplus M_{n_3} \to M_{m_1} \oplus M_{m_2}$ with multiplicities $\alpha_{ki}$ for the inclusion of $M_{n_i}$ into $M_{m_k}$.\\
\protect\subref{fig liftingBratteliDiagram}: Some liftings of the maps (arrows) given in the Bratteli diagram \protect\subref{fig bratteliDiagram} as maps $\phiH[,w]^{v} : \hsA[,v] \to \hsB[,w]$, here represented as (green) arrows decorated with their defining matrices $u(v,w) \in \bbC^{\alpha_{ki}} \otimes \bbC^{\alpha_{\ell j}}$, see \eqref{eq phiH Iklij xi eta}. The configuration for the arrows $v_2 \to w_2$ and $\JimA(v_2) \to \JimB(w_2)$ is the consequence of Corollary~\ref{corr phiH and phiHkappa}. In the even case, according to Lemma~\ref{lemma matrices u(v,w)}, one should have $s(v_2) = s(w_2)$ for $u(v_2,w_2)$ to be non-zero, and similarly for other arrows. The arrows $M_{n_2} \to M_{m_1}$ and $M_{n_2} \to M_{m_2}$ in \protect\subref{fig bratteliDiagram} are not lifted in order to lighten the drawing.}
\label{fig liftingBratteliDiagramKrajewski}
\end{figure}

Using what we have constructed so far, in Fig.~\ref{fig liftingBratteliDiagramKrajewski} we show an example of the lifting of some arrows in a Bratteli diagram as arrows between two Krajewski diagrams.

For any $i,j = 1, \dots, r$, let $\{ \vm_{i j}^{p} \}_{1 \leq p \leq \mu_{i j}}$ be an orthonormal basis of $\bbC^{\mu_{i j}}$ (for instance as in Prop.~\ref{prop basis odd case} or \ref{prop basis even case}), to which we associate the irreps $\hsA[,v]$ defined as in \eqref{hsv from vmijp basis} for any $v = (i, p, j) \in \GammaA[,\bn_i \bn_j]^{(0)}$. One can then fix an orthonormal basis $\{ e_{ij, \alpha} = \xi_{i, \alpha}^{(1)} \otimes \eta_{j, \alpha}^{\circ (2)} \}_{1 \leq \alpha \leq n_i n_j}$ (sumless Sweedler-like notation) of $\bbC^{n_i} \otimes \bbC^{n_j \circ}$. Let
\begin{align*}
\TGammaA[,\bn_i \bn_j]^{(0)}
&\defeq \GammaA[,\bn_i \bn_j]^{(0)} \times \{ 1, \dots, n_i n_j \}
& \text {and}&&
\TGammaA^{(0)} &\defeq \cup_{i,j= 1}^{r} \TGammaA[,\bn_i \bn_j]^{(0)}
\end{align*}
Then for any $\Tv = (v, \alpha) \in \TGammaA[,\bn_i \bn_j]^{(0)}$, let $e_{\Tv} \defeq \xi_{i, \alpha}^{(1)} \otimes \vm_{i j}^{p} \otimes \eta_{j, \alpha}^{\circ (2)} \in \hsA[,v]$. The family $\{ e_{\Tv} \}_{\Tv \in \TGammaA^{(0)}}$ defines an orthonormal basis of $\hsA$. We define $v : \TGammaA^{(0)} \to \GammaA^{(0)}$ as $v(\Tv) = v$ for $\Tv = (v, \alpha)$. Then, for any $v' = (i, p', j) \in \GammaA[,\bn_i \bn_j]^{(0)}$, define 
\begin{align*}
\iota_{v}^{v'} 
&: \hsiA[,v'] \to \hsiA[,v] 
&\text{as}&&
\iota_{v}^{v'}(\xi_{i, \alpha}^{(1)} \otimes \vm_{i j}^{p'} \otimes \eta_{j, \alpha}^{\circ (2)}) 
&= \xi_{i, \alpha}^{(1)} \otimes \vm_{i j}^{p} \otimes \eta_{j, \alpha}^{\circ (2)}.
\end{align*}

\begin{proposition}
\label{prop scalar products}
Let $v, v' \in \GammaA^{(0)}$, $w \in \GammaB^{(0)}$, and $\psi_{v} \in \hsiA[,v]$ and $\psi'_{v'} \in \hsiA[,v']$. 

When $\pi_{\lambda\rho}(v) \neq \pi_{\lambda\rho}(v')$, one has $\langle \phiH[,w]^{v}(\psi_{v}), \phiH[,w]^{v'}(\psi'_{v'}) \rangle_{\hsiB[,w]} = 0$. 

When $\pi_{\lambda\rho}(v) = \pi_{\lambda\rho}(v')$, one has
\begin{align*}
\langle \phiH[,w]^{v}(\psi_{v}), \phiH[,w]^{v'}(\psi'_{v'}) \rangle_{\hsiB[,w]}
&=
\langle \psi_{v}, \iota_{v}^{v'}(\psi'_{v'}) \rangle_{\hsiA[,v]} \, \tr( u(v,w)^\ast u(v',w) )
\end{align*}
In particular, for any $\psi_{v} \in \hsiA[,v]$ and $\psi'_{v'} \in \hsiA[,v]$, one has
\begin{align}
\norm{\phiH[,w]^{v}(\psi_{v})}_{\hsiB[,w]} 
&= 
\norm{\psi_{v}}_{\hsiA[,v]} \norm{u(v,w)}_{F} \nonumber
\\
\langle \phiH^{v}(\psi_{v}), \phiH^{v'}(\psi'_{v'}) \rangle_{\hsiB} 
&= 
\langle \psi_{v}, \iota_{v}^{v'}(\psi'_{v'}) \rangle_{\hsiA[,v]} \Big( \tsum_{w \in \GammaB^{(0)}} \tr( u(v,w)^\ast u(v',w) )  \Big) \label{eq scalar production phiH u}
\end{align}
where $\norm{-}_{F}$ is the Frobenius norm on matrices, defined as $\norm{A}_{F}^2 \defeq \tr(A^\ast A)$. This implies that $\phiH^{v} : \hsiA[,v] \to \hsiB$ is one-to-one if and only if $\tsum_{w \in \GammaB^{(0)}} \norm{u(v,w)}^2_{F} > 0$.
\end{proposition}

\begin{proof}
From a previous remark, the scalar product is zero when $\pi_{\lambda\rho}(v) \neq \pi_{\lambda\rho}(v')$. So, suppose $\pi_{\lambda\rho}(v) = \pi_{\lambda\rho}(v')$. Let $i = i(v) = i(v')$ and $j = j(v) = j(v')$ and consider $\psi_{v} = \xi_{i} \otimes \eta_j^\circ$ and $\psi'_{v'} = \xi'_{i} \otimes \eta_j'^\circ$, so that $\phiH[,w]^{v}(\psi_{v}) = \xi_i \otimes u(v,w) \otimes \eta_j^\circ$ and $\phiH[,w]^{v'}(\psi'_{v'}) = \xi'_i \otimes u(v',w) \otimes \eta_j'^\circ$ both in $\bbC^{n_i} \otimes M_{\alpha_{k i} \times \alpha_{\ell j}} \otimes \bbC^{n_j} \simeq \bbC^{n_i} \otimes \bbC^{\alpha_{k i}} \otimes \bbC^{\alpha_{\ell j}} \otimes \bbC^{n_j} \subset \bbC^{m_k} \otimes \bbC^{m_\ell \circ}$. Then $\langle \phiH[,w]^{v}(\psi_{v}), \phiH[,w]^{v'}(\psi'_{v'}) \rangle_{\hsiB[,w]} = \langle \xi_i, \xi'_i \rangle_{\bbC^{n_i}} \langle \eta_j, \eta'_j \rangle_{\bbC^{n_j}} \tr(u(v,w)^\ast u(v',w))$ where the trace factor is obtained from the identification of $M_{\alpha_{k i} \times \alpha_{\ell j}}$ with $\bbC^{\alpha_{k i}} \otimes \bbC^{\alpha_{\ell j}}$ and we have used the fact that $\lambda(v) = \lambda(v') = n_i$ and $\rho(v) = \rho(v') = n_j$ to write the scalar products. This implies the formula in terms of the scalar product on $\hsiA[,v]$, from which we deduce the relations on the norm in $\hsiB[,w]$ and on the scalar product in $\hsB$. This last relation implies the norms relation $\norm{\phiH^{v}(\psi_{v})}^2_{\hsiB} = \norm{\psi_{v}}^2_{\hsiA[,v]} \Big( \tsum_{w \in \GammaB^{(0)}} \norm{u(v,w)}^2_{F}  \Big)$. Then, suppose $\tsum_{w \in \GammaB^{(0)}} \norm{u(v,w)}^2_{F} > 0$: if $\psi_{v} \in \hsiA[,v]$ is such that $\phiH^{v}(\psi_{v}) = 0$, then $\norm{\psi_{v}}^2_{\hsiA[,v]} = 0$, so that $\psi_{v} = 0$, that is, $\phiH^{v}$ is one-to-one. Suppose $\tsum_{w \in \GammaB^{(0)}} \norm{u(v,w)}^2_{F} = 0$, then $\norm{\phiH^{v}(\psi_{v})}^2_{\hsiB} = 0$ for any $\psi_{v} \in \hsiA[,v]$, so that $\phiH^{v} = 0$, that is, $\phiH^{v}$ is not one-to-one.
\end{proof}

Notice that the condition $\tsum_{w \in \GammaB^{(0)}} \norm{u(v,w)}^2_{F} > 0$ for any $v \in \GammaA^{(0)}$ does not implies that $\phiH$ is one-to-one: one can consider a situation where, for $v, v' \in \GammaA^{(0)}$ such that $\pi_{\lambda\rho}(v) = \pi_{\lambda\rho}(v')$, $\psi_{v} \in \hsiA[,v]$, and $\psi'_{v'} \in \hsiA[,v']$ , $\phiH[,w]^v(\psi_{v}) + \phiH[,w]^{v'}(\psi'_{v'}) = 0 \in \hsiB[,w]$ for some $w \in \GammaB^{(0)}$.

\medskip
From \eqref{eq scalar production phiH u}, it is natural to define, for any  $v, v' \in \GammaA^{(0)}$,
\begin{align*}
\spm^{v_1, v_2} &\defeq
\begin{cases}
0 & \text{if $\pi_{\lambda\rho}(v_1) \neq \pi_{\lambda\rho}(v_2)$}\\
\tsum_{w \in \GammaB^{(0)}} \tr( u(v_1,w)^\ast u(v_2,w) ) & \text{if $\pi_{\lambda\rho}(v_1) = \pi_{\lambda\rho}(v_2)$}
\end{cases}
\end{align*}
so that \eqref{eq scalar production phiH u} can be written as $\langle \phiH^{v_1}(\psi_{v_1}), \phiH^{v_2}(\psi'_{v_2}) \rangle_{\hsiB} = \langle \psi_{v_1}, \iota_{v_1}^{v_2}(\psi'_{v_2}) \rangle_{\hsiA[,v_1]} \, \spm^{v_1, v_2}$. 

With $\pi_{\lambda\rho}(v_1) = \pi_{\lambda\rho}(v_2) = (\bn_i, \bn_j)$, one has $\spm^{v_2, v_1} = \tsum_{w \in \GammaB^{(0)}} \tr( u(v_2,w)^\ast u(v_1,w) ) = \tsum_{w \in \GammaB^{(0)}} \overline{\tr( u(v_1,w)^\ast u(v_2,w) )} = \overline{\spm^{v_1,v_2}}$, so that $(\spm^{v_1,v_2})_{v_1,v_2}$ is a Hermitian matrix, so that this matrix can be diagonalized. Recall that the labels $v_1, v_2$ depends on the choices of the orthonormal bases $\{ \vm_{i j}^{p} \}_{1 \leq p \leq \mu_{i j}}$ of the spaces $\bbC^{\mu_{i j}}$'s: this diagonalization (see proof of Prop.~\ref{prop diagonalization spm for KO dims}) is related to a change of these bases. This leads us to introduce the following Hypothesis.

\begin{hypothesis}
\label{hyp phiH diagonalization}
We suppose that $\phiH$ is one-to-one and is such that there are orthonormal bases $\{ \vm_{i j}^{p} \}_{1 \leq p \leq \mu_{i j}}$ of the spaces $\bbC^{\mu_{i j}}$ which conform to Prop.~\ref{prop basis odd case} (in the odd case) or Prop.~\ref{prop basis even case} (in the even case), and such that, for the decomposition of $\hsA$ induced by these bases, $\spm^{v_1, v_2} = \spe_{v_1} \delta^{v_1, v_2}$ when $\pi_{\lambda\rho}(v_1) = \pi_{\lambda\rho}(v_2)$, with real numbers $\spe_{v}$ such that $\spe_{\JimA(v)} = \spe_{v}$.
\end{hypothesis}

A direct consequence of this hypothesis is that $\langle \phiH^{v_1}(\psi_{v_1}), \phiH^{v_2}(\psi'_{v_2}) \rangle_{\hsiB} = 0$ for any $v_1 \neq v_2$ and  $\langle \phiH^{v}(\psi_{v}), \phiH^{v}(\psi'_{v}) \rangle_{\hsiB} = \spe_{v} \, \langle \psi_{v}, \psi'_{v} \rangle_{\hsiA[,v]}$ for any $v$. The one-to-one requirement is natural in the context of $AF$-algebras, since it generalizes the one-to-one requirement on $\phi$. On the other hand, the diagonalization requirement is not mandatory for the formal developments to come, but it will be useful to compare spectral actions for $\phi$-compatible spectral triples on $\algA$ and $\algB$ in Sect.~\ref{sec spectral actions AF AC manifold}. Moreover, this requirement is satisfied for some $KO$-dimensions:

\begin{proposition}
\label{prop diagonalization spm for KO dims}
Suppose that $\JB$ is strong $\phi$-compatible with $\JA$, and, in the even case, that $\gammaB$ is $\phi$-compatible with $\gammaA$. Then, in $KO$-dimensions $0,1,2,6,7$, the diagonalization requirement in Hypothesis~\ref{hyp phiH diagonalization} is satisfied for any $\phiH$. 
\end{proposition}

\begin{proof}
Let $\{ \vm_{i j}^{p} \}_{1 \leq p \leq \mu_{i j}}$ be orthonormal bases of the spaces $\bbC^{\mu_{i j}}$ which satisfy Prop.~\ref{prop basis odd case} (in the odd case) or Prop.~\ref{prop basis even case} (in the even case). Let us first complete the notations introduced before Prop.~\ref{prop scalar products}, where we have introduced the identification $v = (i,p,j)$. With this notation, we define $\JimA(v) = (j, \Bp, i)$ where $\Bp = 1, \dots, \mu_{ji} = \mu_{ij}$ and $\bar{\Bp} = p$ (obviously, this bar is not to be confused with a complex conjugation). 

Let $\{ \wm_{k \ell}^{q} \}_{1 \leq q \leq \nu_{k\ell}}$ be orthonormal bases of the spaces $\bbC^{\nu_{k\ell}}$ where $\nu_{k\ell}$ are the multiplicity of the irreps $\hsB[, \bem_k \bem_\ell]$ in $\hsB$. These bases define the irreps $\hsB[,w]$ for $w = (k,q,\ell) \in \GammaB[,\bem_k \bem_\ell]^{(0)}$ as in \eqref{hsv from vmijp basis}. We have written the map $\phiH[,w]^{v} : \hsA[,v] \to \hsB[,w]$ in terms of a matrix $u(v,w)$. It is convenient to write $\phiH[,w]^{v}$ explicitly in terms of the bases $\{ \vm_{i j}^{p} \}$ and $\{ \wm_{k \ell}^{q} \}$. In order to avoid cumbersome notations, we use the identification $\bbC^{n_i} \otimes \bbC^{\mu_{i j}} \otimes \bbC^{n_j \circ} \simeq \bbC^{n_i} \otimes \bbC^{n_j \circ} \otimes \bbC^{\mu_{i j}}$ (resp. $\bbC^{m_k} \otimes \bbC^{\nu_{k\ell}} \otimes \bbC^{m_\ell \circ} \simeq \bbC^{m_k} \otimes \bbC^{m_\ell \circ} \otimes \bbC^{\nu_{k\ell}}$) so that $\vm_{i j}^{p}$ (resp. $\wm_{k \ell}^{q}$) will appear on the right in the tensor products. Then we can replace the notation $u(v,w)$ by the notation $\Mu_{k\ell,q}^{ij,p} \in M_{\alpha_{k i} \times \alpha_{\ell j}}$ which refers to the bases $\{ \vm_{i j}^{p} \}$ and $\{ \wm_{k \ell}^{q} \}$ for which, similarly to \eqref{eq phiH Iklij xi eta}, one has $\phiH[,w]^{v} (\xi_i \otimes \eta_j^\circ \otimes \vm_{i j}^{p}) = I_{k,\ell}^{i,j} \big( \xi_i \otimes \Mu_{k\ell,q}^{ij,p} \otimes \eta_j^\circ \big) \otimes \wm_{k \ell}^{q}$ (no summation). In the $p$ and $q$ variables, $\spm^{v_1, v_2}$ with $\pi_{\lambda\rho}(v_1) = \pi_{\lambda\rho}(v_2) = (\bn_i, \bn_j)$ then becomes $\hspm^{p_1, p_2}_{ij} = \tsum_{k, \ell, q} \tr( (\Mu^{ij,p_2}_{k\ell,q})^\ast \Mu^{ij,p_1}_{k\ell,q} )$ for $p_1, p_2 = 1, \dots, \mu_{ij}$. Notice the switch $1 \leftrightarrow 2$ which will be convenient later. Since we suppose that $\JB$ is strong $\phi$-compatible with $\JA$, by Prop.~\ref{prop JB JA strong phi compatibiliy relation on ukappa} one get \eqref{eq ukappa epsilon ustar} in terms of the new notations: $\Mu_{\ell k,\Bq}^{ji,\Bp} = \tfrac{\epsilonA(i,p,j,d_\algA)}{\epsilonB(k,\ell,q,d_\algB)} (\Mu_{k\ell,q}^{ij,p})^\ast$. This implies that
\begin{align}
\hspm^{\Bp_1, \Bp_2}_{ji} 
&=
\tsum_{\ell, k, \Bq} \tr( (\Mu^{ji,\Bp_2}_{\ell k,\Bq})^\ast \Mu^{ji,\Bp_1}_{\ell k,\Bq} ) 
= \tsum_{k, \ell, q}  \tfrac{\epsilonA(i,p_2,j,d_\algA) \epsilonA(i,p_1,j,d_\algA)}{\epsilonB(k,\ell,q,d_\algB)^2} \tr( \Mu^{ij,p_2}_{k\ell,q} (\Mu^{ij,p_1}_{k\ell,q})^\ast ) \nonumber
\\
\label{eq hspmBp1Bp2ji = hspmp2p2ij}
&= \epsilonA(i,p_1,j,d_\algA) \epsilonA(i,p_2,j,d_\algA) \hspm^{p_2, p_1}_{ij} 
\end{align}

In the following, we fix the couple $(i,j)$. Let us introduce a change of bases $\{ \vm_{i j}^{p} \}$ to $\{ \vm'^{p'}_{i j} \}$ in $\bbC^{\mu_{i j}}$, where $\vm'^{p'}_{i j} = \tsum_{p} u^{p'p} \vm_{i j}^{p}$ for a unitary matrix $U = (u^{p'p})_{p',p}$. Then a straightforward computation shows that the matrices $\Mu'^{ij,p'}_{k\ell,q}$ defined relatively to the bases $\{ \vm'^{p'}_{i j} \}$ and $\{ \wm_{k \ell}^{q} \}$ are $\Mu'^{ij,p'}_{k\ell,q} = \tsum_{p} u^{p'p} \Mu^{ij,p}_{k\ell,q}$, and the associated $\hspm'^{p'_1, p'_2}_{ij} = \tsum_{k, \ell, q} \tr( (\Mu'^{ij,p'_2}_{k\ell,q})^\ast \Mu'^{ij,p'_1}_{k\ell,q} )$ become
\begin{align*}
\hspm'^{p'_1, p'_2}_{ij}
&= \tsum_{k, \ell, q} 
\tsum_{p_1, p_2} \overline{u^{p'_2 p_2}} u^{p'_1 p_1}\, 
\tr( (\Mu^{ij,p_2}_{k\ell,q})^\ast \Mu^{ij,p_1}_{k\ell,q} )
\\
&= \tsum_{p_1, p_2} \overline{u^{p'_2 p_2}} u^{p'_1 p_1}\, \hspm^{p_1, p_2}_{ij},
\end{align*}
so that $\hspm'_{ij} = U \hspm_{ij} U^\ast$ with $U^\ast = (\overline{u^{p p'}})_{p',p}$ (here we use the switch $1 \leftrightarrow 2$ mentioned before). Since $\hspm_{ij}$ is a Hermitian matrix, there is a unitary matrix $U$ such that $\hspm'_{ij} = U \hspm_{ij} U^\ast$ is diagonal with real eigenvalues $\spe_{ij}^p = \spe_{v}$. So, for the new basis $\{ \vm'^{p'}_{i j} \}$ of $\bbC^{\mu_{i j}}$ defined by $U$, $\hspm'_{ij}$, and so $(\spm^{v_1, v_2})_{v_1,v_2}$, is diagonal.

Let us now look how this diagonalization can be performed according to the constraints in Prop.~\ref{prop basis odd case} (in the odd case) or Prop.~\ref{prop basis even case} (in the even case). The first constraint, common to Prop.~\ref{prop basis odd case} and \ref{prop basis even case}, is $\vm_{j i}^{\Bp} = \epsilonA(i, p, j ,d_\algA) L_{i j}(\Bvm_{i j}^{p})$ for any $p$, where $\epsilonA(i, p, j ,d_\algA) = \epsilonA(v ,d_\algA)$ is defined in \eqref{eq epsilon(v, d)}.

Let us first consider the case $i<j$ (for any $KO$-dimension), for which $\epsilonA(i, p, j ,d_\algA)=1$ and $\Bp = p$, so that, from \eqref{eq hspmBp1Bp2ji = hspmp2p2ij}, one has $\hspm^{p_1, p_2}_{ji} = \hspm^{p_2, p_1}_{ij}$: $\hspm_{ji}$ is the transpose of $\hspm_{ij}$. Since this result is true in any basis of $\bbC^{\mu_{i j}}$, this implies $\hspm'_{ji} = \BU  \hspm_{ji} \BU^\ast$. On the other hand, $\vm'^{p'}_{j i} =  L_{i j}(\Bvm'^{p'}_{i j}) = \tsum_{p} \overline{u^{p'p}} L_{i j}(\Bvm_{i j}^{p}) = \tsum_{p} \overline{u^{p'p}} \vm^{p}_{j i}$, so that the change of bases from $\{ \vm^{p}_{j i} \}$ to $\{ \vm'^{p'}_{j i} \}$ in $\bbC^{\mu_{j i}}$ is performed by the unitary matrix $\BU$. From these two compatible relations, one concludes that the change of basis defined by $U$ in  $\bbC^{\mu_{i j}}$ which diagonalizes $\hspm'_{ij}$ automatically induces a change of bases $\BU$ in $\bbC^{\mu_{j i}}$ which diagonalizes $\hspm'_{ji}$. Notice then that the eigenvalues $\spe_{ji}^p$ in $\hspm'_{ji}$ are the same as the eigenvalues $\spe_{ij}^p$ in $\hspm'_{ij}$, so that $\spe_{\JimA(v)} = \spe_{v}$.

Let us now consider $i=j$ in $KO$-dimensions $0,1,7$. Then, as before, $\epsilonA(i, p, i ,d_\algA)=1$ and $\Bp = p$, so that, from \eqref{eq hspmBp1Bp2ji = hspmp2p2ij}, one has  $\hspm^{p_1, p_2}_{ii} = \hspm^{p_2, p_1}_{ii}$, and we already know that $\hspm^{p_1, p_2}_{ii} = \overline{\hspm^{p_2, p_1}_{ii}}$: the matrix $\hspm_{ii}$ is a real symmetric matrix, and the diagonalizing matrix $U$ can be chosen to be an orthogonal matrix (so a real matrix). This result is compatible with the required condition $\vm_{ii}^{p} = L_{ii}(\Bvm_{ii}^{p})$ on the basis since $\vm'^{p'}_{ii} = L_{ii}(\Bvm'^{p'}_{ii}) = \tsum_{p} \overline{u^{p'p}} L_{ii}(\Bvm_{ii}^{p}) = \tsum_{p} \overline{u^{p'p}} \vm_{ii}^{p} = \tsum_{p} u^{p'p} \vm_{ii}^{p}$. Here, it is trivial that $\spe_{\JimA(v)} = \spe_{v}$ since $\JimA(v) = v$.

Finally, consider $i=j$ in $KO$-dimensions $2,3,4,5,6$. In that situation, if $p=2a$ (resp. $p=2a-1$) then $\Bp = 2a -1$ (resp. $\Bp = 2a$), and  $\epsilonA(i, 2a-1, i ,d_\algA)=1$ and  $\epsilonA(i, 2a, i ,d_\algA)=\epsilonA$. The matrix $\hspm_{ii}$ is a  block matrix $\smallpmatrix{\hspm_{ii}^{e,e} & \hspm_{ii}^{e,o} \\ \hspm_{ii}^{o,e} & \hspm_{ii}^{o,o} }$ where $o$ and $e$ stand for odd and even: for instance $\hspm_{ii}^{e,e} = \left( \hspm_{ii}^{2a_1, 2a_2} \right)$ and  $\hspm_{ii}^{e,o} = \left( \hspm_{ii}^{2a_1, 2a_2-1} \right)$ with $a_1, a_2=1, \dots, \mu_{ii}/2$. Then, from \eqref{eq hspmBp1Bp2ji = hspmp2p2ij}, one has $\hspm_{ii}^{2a_1, 2a_2} = \hspm_{ii}^{2a_2-1, 2a_1-1}$, $\hspm_{ii}^{2a_1, 2a_2-1} = \epsilonA \hspm_{ii}^{2a_2, 2a_1-1}$, and $\hspm_{ii}^{2a_1-1, 2a_2} = \epsilonA \hspm_{ii}^{2a_2-1, 2a_1}$. Considering these block matrices as matrices indexed by $a_1, a_2$, this means that $\hspm_{ii}^{e,e} = {\hspm_{ii}^{o,o}}^\top$, $\hspm_{ii}^{e,o} = \epsilonA {\hspm_{ii}^{e,o}}^\top$, and $\hspm_{ii}^{o,e} = \epsilonA {\hspm_{ii}^{o,e}}^\top$. Since $\hspm_{ii}$ is Hermitian, one also has $\hspm_{ii}^{e,e} = {\hspm_{ii}^{e,e}}^\ast$ and $\hspm_{ii}^{e,o} = {\hspm_{ii}^{o,e}}^\ast$.

In $KO$-dimensions $3,4,5$, one has $\epsilonA = -1$, so that $\hspm_{ii}^{e,o} = - {\hspm_{ii}^{e,o}}^\top = {\hspm_{ii}^{o,e}}^\ast$, which implies that $\hspm_{ii}^{e,o}$ and $\hspm_{ii}^{o,e}$ are antisymmetric matrices. We report the analysis for $KO$-dimensions $2,6$ after the following considerations.

In the even case, since $\gammaB$ is $\phi$-compatible with $\gammaA$, from Lemma~\ref{lemma matrices u(v,w)}, $u(v,w)$ is non-zero only when $s(v) = s(w)$, so that the sum defining $\spm^{v_1,v_2}$ implies $s(w) = s(v_1) = s(v_2)$. The matrix $(\spm^{v_1, v_2})_{v_1,v_2}$ is then block diagonal along the decomposition $s(v) = \pm 1$, and its diagonalization can be done by blocks: in terms of the change of bases in $\bbC^{\mu_{i j}}$, this means that the unitary $U$ introduced above which diagonalizes $\hspm_{ij}$ can be chosen to preserve the eigenspaces defined by the maps $\ell_{ij}$ in Prop~\ref{prop basis even case}. The decomposition along $s(v) = \pm 1$ is preserved by $\JimA$ since $s_{ji}^{\Bp} = \epsilonA'' s_{ij}^{p}$: so all the previous developments are compatible with this choice for $U$.

In the case  $i=j$ and $KO$-dimensions $2,6$, from Prop.~\ref{prop basis even case}, one has $s(i,2a,i) = 1$ and $s(i,2a-1,i) = -1$, so that the block decomposition $\smallpmatrix{\hspm_{ii}^{e,e} & \hspm_{ii}^{e,o} \\ \hspm_{ii}^{o,e} & \hspm_{ii}^{o,o} }$ corresponds to the block decomposition along  $s(v) = \pm 1$, and from the previous considerations, this implies that  $\hspm_{ii}^{e,o} = \hspm_{ii}^{o,e} = 0$. Since $\hspm_{ii}^{e,e}$ is Hermitean, there is a unitary matrix $\TU$ such that $\TU \hspm_{ii}^{e,e} \TU^\ast$ is diagonal, and then by transposition, $\BTU \hspm_{ii}^{o,o} \BTU^\ast$ is also diagonal with the same eigenvalues, that is, $\spe_{\JimA(v)} = \spe_{v}$. The unitary $U = \smallpmatrix{\TU & 0 \\ 0 & \BTU }$ diagonalizes $\hspm_{ii}$ and this diagonalization is compatible with the required conditions $\vm_{ii}^{2a} = L_{ii}(\Bvm_{ii}^{2a-1})$ and $\vm_{ii}^{2a-1} = \epsilonA L_{ii}(\Bvm_{ii}^{2a})$: $\TU = ( \Tu^{a',a})$ (resp. $\BTU = ( \BTu^{a',a})$) induces a change of the sub-basis $\{ \vm_{ii}^{2a} \}$ to $\{ \vm'^{2a}_{ii} \}$ (resp. $\{ \vm_{ii}^{2a-1} \}$ to $\{ \vm'^{2a-1}_{ii} \}$) with $\vm'^{2a'}_{ii} = \tsum_{a} \Tu^{a',a} \vm_{ii}^{2a}$ (resp.  $\vm'^{2a'-1}_{ii} = \tsum_{a} \BTu^{a',a} \vm_{ii}^{2a-1}$). The required condition is satisfied since then $\vm'^{2a'-1}_{ii} = \epsilonA L_{ii}(\Bvm'^{2a'}_{ii}) = \epsilonA  \tsum_{a} \BTu^{a',a} L_{ii}(\Bvm^{2a}_{ii}) = \tsum_{a} \BTu^{a',a} \vm_{ii}^{2a-1}$.
\end{proof}

\begin{remark}
We suspect that the diagonalization property proved in Prop.~\ref{prop diagonalization spm for KO dims} could be true also in $KO$-dimensions $3,4,5$. But we were unable to prove this fact. Nevertheless, the proposition fortunately covers the $KO$-dimension $6$ used in the finite part of the spectral triple for the NC version of the Standard Model of Particles Physics, see \cite{ChamConnMarc07a} and \cite{Suij15a} for instance.
\end{remark}

\begin{proposition}
\label{prop KO dim AF phi compatibility}
If two (odd/even) real spectral triples are $\phi$-compatible and $\phiH$ is such that \eqref{eq ukappa epsilon ustar} holds, then they have the same $KO$-dimension (mod 8).
\end{proposition}

\begin{proof}
Since $\phiH$ satisfies \eqref{eq ukappa epsilon ustar}, by Prop.~\ref{prop JB JA strong phi compatibiliy relation on ukappa}, $\JB$ is strong $\phi$-compatible with $\JA$ and is diagonal. By Lemma~\ref{lemma gammaB gammaA phi compatibiliy diagonal}, $\gammaB$ is strong $\phi$-compatible with $\gammaA$ and is diagonal. The difference with Prop.~\ref{prop KO dim strong phi compatibility}, is that $\DB$ is only $\phi$-compatible with $\DA$. So, we already get $\epsilonB =  \epsilonA$ and $\epsilonB'' = \epsilonA''$: it remains to consider $\epsilonA'$ and $\epsilonB'$. 

Since $\JB$ is diagonal, one has $\JB \DB = \smallpmatrix{ \JB[,\phi]^{\phi} \DB[,\phi]^{\phi} & \JB[,\phi]^{\phi} \DB[,\phi]^{\perp} \\ \JB[,\perp]^{\perp} \DB[,\perp]^{\phi} & \JB[,\perp]^{\perp} \DB[,\perp]^{\perp} }$ and $\DB \JB = \smallpmatrix{ \DB[,\phi]^{\phi} \JB[,\phi]^{\phi} & \DB[,\phi]^{\perp} \JB[,\perp]^{\perp} \\ \DB[,\perp]^{\phi} \JB[,\phi]^{\phi} & \DB[,\perp]^{\perp} \JB[,\perp]^{\perp} }$, so that $\JB[,\phi]^{\phi} \DB[,\phi]^{\phi} = \epsilonB' \DB[,\phi]^{\phi} \JB[,\phi]^{\phi}$. Inserting this relation in the $\phi$-compatibility conditions on $\JB$ and $\DB$ implies $\epsilonB' = \epsilonA'$.
\end{proof}

From Prop.~\ref{prop st phi comp and phi comp of operators for AF}, the strong $\phi$-compatibility condition between $\DB$ and $\DA$ is equivalent to
\begin{align*}
\tsum_{v_2 \in \GammaA^{(0)}}\phiH[,w_2]^{v_2}(\DA[,(v_1, v_2)] \psi_{v_1})
&=
\tsum_{w_1 \in \GammaB^{(0)}} \DB[,(w_1,w_2)] \phiH[,w_1]^{v_1}(\psi_{v_1})
\end{align*}
for any $v_1 \in \GammaA^{(0)}$, $w_2 \in \GammaB^{(0)}$, and $\psi_{v_1} \in \hsiA[,v_1]$, and the $\phi$-compatibility condition is equivalent to
\begin{align*}
\tsum_{v_2 \in \GammaA^{(0)}}\phiH[,w_2]^{v_2}(\DA[,(v_1, v_2)] \psi_{v_1})
&=
\tsum_{w_1 \in \GammaB^{(0)}} \DB[, \phi,(w_1,w_2)]^{\phi} \phiH[,w_1]^{v_1}(\psi_{v_1})
\end{align*}
where $\DB[, \phi,(w_1,w_2)]^{\phi} : \hsiB[,w_1] \cap \phiH(\hsA) \to \hsiB[,w_2] \cap \phiH(\hsA)$. Unfortunately, from this relation, we cannot define the elementary operators $\DB[, \phi,(w_1,w_2)]^{\phi}$ in terms of the elementary operators $\DA[,(v_1, v_2)]$. Only the operators $\tsum_{w_1 \in \GammaB^{(0)}} \DB[, \phi,(w_1,w_2)]^{\phi} : \toplus_{w_1 \in \GammaB^{(0)}} \hsB[, w_1] \cap \phiH(\hsA) \to \hsB[, w_2] \cap \phiH(\hsA)$ can be recovered from the $\DA[,(v_1, v_2)]$'s.

\section{\texorpdfstring{Spectral actions for $AF$-AC Manifolds}{Spectral actions for AF-AC Manifolds}}
\label{sec spectral actions AF AC manifold}

Given a spectral action $(\algA, \hsA, \DA, \JA, \gammaA)$ for a finite dimensional algebra $\algA$ and given a compact Riemannian spin manifold $(M,g)$ equipped with its canonical spectral triple $(C^\infty(M), L^2(S), \DM, \JM, \gammaM)$, we consider the spectral triple $(\halgA \defeq C^\infty(M) \otimes \algA , \hshA \defeq L^2(S) \otimes \hsA, \DhA \defeq \DM \otimes 1 + \JM \otimes \DA, \JhA \defeq \JM \otimes \JA, \gammahA \defeq \gammaM \otimes \gammaA)$\footnote{When the $KO$-dimension for $M$ and $\algA$ produces such a spectral triple, see for instance \cite{DabrDoss11a}} over the Almost Commutative algebra $\halgA$.

Then, given two spectral triples $(\algA, \hsA, \DA, \JA, \gammaA)$ and $(\algB, \hsB, \DB, \JB, \gammaB)$ for two finite dimensional algebras $\algA$ and $\algB$, and a one-to-one homomorphism $\phi : \algA \to \algB$ such that the two spectral triple are $\phi$-compatible, with $\JB$ strong $\phi$-compatible with $\JA$, the aim of this section is to compare the spectral actions on $\halgA$ and $\halgB$ (for the same compact Riemannian spin manifold $(M,g)$). 

In order to have a good physical interpretation of the $\phi$-compatibility, in particular at the level of the fermions, we first need to introduce a “normalized” $\phiH$ map.

\subsection{\texorpdfstring{Normalized $\phiH$ map}{Normalized phiH map}}
\label{subsec normalized phiH map}

Denote by $\phiH^0 : \hsA \to \hsB$ a given one-to-one morphism as in Def.~\ref{def phi phiH}. We suppose that it satisfies Hypothesis~\ref{hyp phiH diagonalization}. Then we can choose the orthonormal bases $\{ \vm_{i j}^{p} \}_{1 \leq p \leq \mu_{i j}}$ of $\bbC^{\mu_{i j}}$ that diagonalize $(\spm^{v_1, v_2})_{v_1,v_2}$ and this implies that $\{ \phiH^0(e_{\Tv}) \}_{\Tv \in \TGammaA^{(0)}}$ is a basis of $\phiH^0(\hsA)$. For any $\Tv = (v,\alpha)$, we  can identify $\phiH^0(e_{\Tv})$ with $\phiH^{0,v}(e_{\Tv})$. Let $\Tv_1 = (v_1, \alpha_1)$ and $\Tv_2 = (v_2, \alpha_2)$. When $v_1 \neq v_2$, one has $\langle \phiH^{0,v}(e_{\Tv_1}), \phiH^{0,v}(e_{\Tv_2}) \rangle_{\hsiB} = 0$, while,  when $v = v_1 = v_2$, $\langle \phiH^{0,v}(e_{\Tv_1}), \phiH^{0,v}(e_{\Tv_2}) \rangle_{\hsiB} = \spe_{v} \, \langle \Tv_1, \Tv_2 \rangle_{\hsiA[,v]} = \spe_{v} \, \delta_{\alpha_1,\alpha_2}$. This implies that $\{ \phiH^0(e_{\Tv}) \}_{\Tv \in \TGammaA^{(0)}}$ is an orthogonal family. Since $\phiH^0$ is one-to-one and $\norm{\phiH^{0,v}(e_{\Tv})}^2 = \spe_{v}$, one has $\spe_{v} > 0$ for any $v \in \GammaA^{(0)}$.

\begin{definition}
The normalized $\phiH$ map associated to the map $\phiH^0 : \hsA \to \hsB$ which satisfies Hypothesis~\ref{hyp phiH diagonalization} is the map $\phiH : \hsA \to \hsB$ defined by
\begin{align*}
\phiH( \toplus_{v \in \GammaA^{(0)}} \psi_v) 
\defeq \tsum_{v \in \GammaA^{(0)}} \spe_{v}^{-1/2} \phiH^{0,v}(\psi_v)
\end{align*}
Using \eqref{eq phiH w v a b psiv} with $b_j = 0$, one can check that $\phiH$ satisfies Def.~\ref{def phi phiH}.
\end{definition}

The normalization has been chosen such that the family $\{ f_{\Tv} \defeq \phiH(e_{\Tv}) \}_{\Tv \in \TGammaA^{(0)}}$ is an orthonormal basis of $\phiH(\hsA)$. This basis of $\phiH(\hsA)$ can be completed with any orthonormal basis $\{ f_{\hw} \}_{\hw \in \hGammaB^{(0)}}$ of $\phiH(\hsA)^\perp$ where $\hGammaB^{(0)}$ is any index set for this basis. So, $\{ f_{\Tv} \}_{\Tv \in \TGammaA^{(0)}} \cup \{ f_{\hw} \}_{\hw \in \hGammaB^{(0)}}$ is an orthonormal basis of $\hsB$ adapted to the decomposition $\phiH(\hsA) \oplus \phiH(\hsA)^\perp$.

\medskip
We now consider the normalized $\phiH$ map in place of $\phiH^0$. The relation between the scalar products in $\hsB$ and $\hsA$ reduces to the simple relation $\langle \phiH(\psi), \phiH(\psi') \rangle_{\hsB} = \langle \psi, \psi' \rangle_{\hsA}$ for any $\psi, \psi' \in \hsA$, so that $\phiH$ is an isometry.

In the following, $\phi$-compatibility of operators will be relative to the normalized $\phiH$ map.

For an operator $B$ on $\hsB$ which is $\phi$-compatible with an operator $A$ on $\hsA$, the components $B_{\phi}^{\perp}$, $B_{\perp}^{\phi}$, and $B_{\perp}^{\perp}$ of $B$ in the $2\times 2$ matrix decomposition induced by $\hsB = \phiH(\hsA) \oplus \phiH(\hsA)^\perp$ will be called non-inherited, while the component $B_{\phi}^{\phi}$ will be called inherited. Let us use the acronym “\TNIC” for “Terms with Non-Inherited Components” in the following technical results, which are the main interest for the use of the normalized $\phiH$ map:
\begin{lemma}
\label{lemma scalar product and trace phi compatible}
For $i=1,\ldots,n$, let $B_i$ be an operator on $\hsB$ which is $\phi$-compatible with an operator $A_i$ on $\hsA$.
\begin{enumerate}
\item For any $\Tv_1, \Tv_2 \in \TGammaA^{(0)}$, one has\label{item scalar product of product}
\begin{align*}
\langle f_{\Tv_1}, B_1 \cdots B_n f_{\Tv_2} \rangle_{\hsB} 
&= \langle e_{\Tv_1}, A_1 \cdots A_n e_{\Tv_2} \rangle_{\hsA} + \text{\TNIC}
\end{align*}

\item As a consequence, one has\label{item trace of product}
\begin{align*}
\tr(B_1 \cdots B_n) = \tr(A_1 \cdots A_n) + \text{\TNIC}
\end{align*}
\end{enumerate}
\end{lemma}

\begin{proof}
First, let us prove the relation in Point~\ref{item scalar product of product} for $n=1$. We omit the index $i$. Using the matrix decomposition $A e_{\Tv} = \tsum_{\Tv' \in \TGammaA^{(0)}} A_{\Tv}^{\Tv'} e_{\Tv'}$ along the basis $\{ e_{\Tv} \}_{\Tv \in \TGammaA^{(0)}}$, the RHS is $A_{\Tv_2}^{\Tv_1}$. For the LHS, one has $\langle f_{\Tv_1}, B f_{\Tv_2} \rangle_{\hsB} = \langle \phiH^{v_1}(e_{\Tv_1}), B_{\phi}^{\phi} \phiH^{v_2}(e_{\Tv_2}) \rangle_{\hsB} = \langle \phiH^{v_1}(e_{\Tv_1}), \phiH^{v_2}(A e_{\Tv_2}) \rangle_{\hsB} = \tsum_{\Tv \in \TGammaA^{(0)}} A_{\Tv_2}^{\Tv} \langle \phiH^{v_1}(e_{\Tv_1}), \phiH^{v}(e_{\Tv}) \rangle_{\hsB}$. This expression is zero for $v_1 \neq v$, so the summation reduces to the summation over the $\Tv = (v_1, \alpha) \in  \TGammaA^{(0)}$: $\tsum_{\Tv = (v_1, \alpha)} A_{\Tv_2}^{\Tv} \langle \phiH^{v_1}(e_{\Tv_1}), \phiH^{v_1}(e_{\Tv}) \rangle_{\hsB} = \tsum_{\Tv = (v_1, \alpha)} A_{\Tv_2}^{\Tv} \langle e_{\Tv_1}, e_{\Tv} \rangle_{\hsA} = A_{\Tv_2}^{\Tv_1}$.

Let us return to the general situation $n \geq 1$ in Point~\ref{item scalar product of product}. With $B_i = \smallpmatrix{ B_{i, \phi}^{\phi} & B_{i, \phi}^{\perp} \\ B_{i, \perp}^{\phi} & B_{i, \perp}^{\perp} }$, a straightforward computation shows that the only component in $B_1 \cdots B_n$ that contains only inherited components is in the block $(B_1 \cdots B_n)_{\phi}^{\phi}$ and it is $B_{1, \phi}^{\phi} \cdots B_{n, \phi}^{\phi}$, so that $\langle f_{\Tv_1}, B_1 \cdots B_n f_{\Tv_2} \rangle_{\hsB}  = \langle f_{\Tv_1}, B_{1, \phi}^{\phi} \cdots B_{n, \phi}^{\phi} f_{\Tv_2} \rangle_{\hsB} + \text{\TNIC}$. The proof that $\langle f_{\Tv_1}, B_{1, \phi}^{\phi} \cdots B_{n, \phi}^{\phi} f_{\Tv_2} \rangle_{\hsB} = \langle e_{\Tv_1}, A_1 \cdots A_n e_{\Tv_2} \rangle_{\hsA}$ is the same as before, with $B_{\phi}^{\phi} = B_{1, \phi}^{\phi} \cdots B_{n, \phi}^{\phi}$ and $A = A_1 \cdots A_n$ which satisfy $\phiH(A \psi) = B_{\phi}^{\phi} \phiH(\psi)$.

Point~\ref{item trace of product} is a direct consequence of Point~\ref{item scalar product of product}. By the previous argument on the product $B_1 \cdots B_n$, one has 
\begin{align*}
\tr(B_1 \cdots B_n) 
&= \tsum_{\Tv \in \TGammaA^{(0)}} \langle f_{\Tv}, B_1 \cdots B_n f_{\Tv} \rangle_{\hsB} + \tsum_{\hw \in \hGammaB^{(0)}} \langle f_{\hw}, B_1 \cdots B_n f_{\hw} \rangle_{\hsB} 
\\
&= \tsum_{\Tv \in \TGammaA^{(0)}} \langle f_{\Tv}, B_1 \cdots B_n f_{\Tv} \rangle_{\hsB} + \text{\TNIC}
\end{align*}  
and $\tsum_{\Tv \in \TGammaA^{(0)}} \langle f_{\Tv}, B_1 \cdots B_n f_{\Tv} \rangle_{\hsB}  = \tsum_{\Tv \in \TGammaA^{(0)}} \langle e_{\Tv}, A_1 \cdots A_n e_{\Tv} \rangle_{\hsA} +  \text{\TNIC} = \tr(A_1 \cdots A_n) + \text{\TNIC}$ by Point~\ref{item scalar product of product}.
\end{proof}

\begin{remark}
\label{remark non normalized}
Point~\ref{item trace of product} can be proved directly without the assumption that $\phiH$ is normalized.
\end{remark}

The notion of $\phi$-compatibility has been developed for operators on $\hsA$ and $\hsB$. We define $\phi$-compatibility for fermions as follows:
\begin{definition}
A vector $\psi_\algB$ is $\phi$-compatible with a vector $\psi_\algA$ if $\psi_\algB - \phiH(\psi_\algA) \in  \phiH(\hsA)^\perp$.
\end{definition}

Using this definition, we extend the acronym “\TNIC” (“Terms with Non-Inherited Components”) to terms which contains fermions.

\smallskip
From a physical point of view, the consequence of the normalization of $\phiH$ is that, for any $\psi \in \hsA$, one has
\begin{align*}
\norm{\phiH(\psi)}_{\hsB} = \norm{\psi}_{\hsA}
\end{align*} 
since, with $\psi = \toplus_{v \in \GammaA^{(0)}} \psi_v$, one has $\norm{\phiH(\psi)}_{\hsB}^2 = \tsum_{v \in \GammaA^{(0)}} \langle \phiH^{v}(\psi_v), \phiH^{v}(\psi_v) \rangle_{\hsB} = \tsum_{v \in \GammaA^{(0)}} \langle \psi_v, \psi_v \rangle_{\hsA} = \norm{\psi}_{\hsA}^2$. This means the $\phiH$ respects the normalization of the state vector when it is injected from $\hsA$ into $\hsB$. This state $\psi$ can be “diluted” at different “places” in $\hsB$ (different irreps) but its norm is conserved.

\subsection{Comparison of Spectral Actions}
\label{subsec comparison of spectral actions}

The purpose of this section is to compare the spectral actions for the two $\hphi$-compatible spectral triples
\begin{align*}
&(\halgA, \hshA, \DhA \defeq \DM \otimes 1 + \JM \otimes \DA, \JhA \defeq \JM \otimes \JA, \gammahA \defeq \gammaM \otimes \gammaA)
\\
&\text{and}
\\
&(\halgB, \hshB, \DhB \defeq \DM \otimes 1 + \JM \otimes \DB, \JhB \defeq \JM \otimes \JB, \gammahB \defeq \gammaM \otimes \gammaB)
\end{align*}
for the same compact Riemannian spin manifold $(M,g)$. We suppose that $\phiH : \hsA \to \hsB$ is normalized.

We extend in a natural way the map $\phi$ to a homomorphism of algebras $\hphi \defeq \Id \otimes \phi : \halgA \to \halgB$. In the same way, we denote by $\hphiH \defeq \Id \otimes \phiH : \hshA \to \hshB$ the natural extension of $\phiH$. The notion of (strong) $\hphi$-compatibility is then naturally defined from the notion of (strong) $\phi$-compatibility: an operator $\hB = B_M \otimes B_F$ on $\hshB$ is (strong) $\hphi$-compatible with an operator $\hA = A_M \otimes A_F$ on $\hshA$ if $\hphiH(\hA (\chi \otimes \psi)) = \hB_{\hphi}^{\hphi} \hphiH( \chi \otimes \psi)$ (resp. $\hphiH(\hA \chi \otimes \psi) = \hB \hphiH( \chi \otimes \psi)$) for any $\chi \otimes \psi \in \hshA$, that is, $\chi \otimes \phiH( A_F \psi) = \chi \otimes B_{F, \phi}^{\phi} \phiH(\psi)$ (resp. $\chi \otimes \phiH( A_F \psi) = \chi \otimes B_F \phiH(\psi)$).

Lemma~\ref{lemma scalar product and trace phi compatible} then extends naturally to:
\begin{lemma}
\label{lemma scalar product and trace phi compatible AC}
Let $\{ \chi_{c} \}_{c \in C}$ an orthonormal basis of $L^2(S)$. For $i=1,\ldots,n$, let $\hB_i$ be an operator on $\hshB$ which is $\hphi$-compatible with an operator $\hA_i$ on $\hshA$.
\begin{enumerate}
\item For any $\Tv_1, \Tv_2 \in \TGammaA^{(0)}$ and $c_1, c_2 \in C$, one has\label{item scalar product of product AC}
\begin{align*}
\langle \chi_{c_1} \otimes f_{\Tv_1}, \hB_1 \cdots \hB_n \chi_{c_2} \otimes f_{\Tv_2} \rangle_{\hshB} 
= \langle \chi_{c_1} \otimes e_{\Tv_1}, \hA_1 \cdots \hA_n \chi_{c_2} \otimes e_{\Tv_2} \rangle_{\hshA}
 + \text{\TNIC}
\end{align*}

\item As a consequence, one has\label{item trace of product AC}
\begin{align*}
\Tr(\hB_1 \cdots \hB_n) = \Tr(\hA_1 \cdots \hA_n) + \text{\TNIC}
\end{align*}
\end{enumerate}
\end{lemma}

\begin{proof}
The proof is similar to the one of Lemma~\ref{lemma scalar product and trace phi compatible}, noticing that the geometrical part plays no role in the main steps.
\end{proof}

We follow \cite{Suij15a} to define the bosonic and the fermionic spectral actions. For any $\omega \in \Omega^1_U(\halgA)$, let consider the operator $\DhA[,\omega] = \DhA + \omega + \epsilonhA' \JhA \omega \JhA^{-1}$ where $\omega$ is used in place of $\piDhA(\omega)$. Let $f : \bbR \to \bbR$ be a positive even function. Then the bosonic spectral action is defined by
\begin{align*}
S_b[\omega]
&\defeq \Tr f(\DhA[,\omega] / \Lambda)
\end{align*}
for a real cutoff parameter $\Lambda$, and where $\Tr$ is the operator trace. We require that $f$ is such that $f(\DhA[,\omega] / \Lambda)$ is a traceclass operator.

To define the fermionic spectral action, we introduce the vector space of Grassmann vectors $\ThshA$ defined from $\hshA$, and the notation $\Tpsi \in \ThshA$ for any $\psi \in \hshA$. Then, in the even case, for any $\Tpsi \in \ThshA^+$, where $\ThshA^+$ corresponds to Grassmann vectors associated to vectors $\psi \in \hshA^+ = \ker (\gammahA - 1)$ (even elements in $\hshA$), the fermionic spectral action is
\begin{align*}
S_f[\omega, \Tpsi] 
&\defeq \langle \JhA \Tpsi, \DhA[,\omega] \Tpsi \rangle_{\ThshA}
\end{align*}

From now on, we suppose that $\dim M = 4$ and, to simplify the presentation (to focus mainly on the algebraic part of the spectral actions), we suppose that $(M,g)$ is compact and flat, so that all the Riemann tensors will be trivial in the following. 

Let us use the following notations. For any $\omega\in \Omega^1_U(\halgA)$ with $\piDhA(\omega) = \gamma^\mu \otimes A_\mu + \gamma_M \otimes \varphi $, for Hermitian operators $A_\mu$ and $\varphi$ on $C^\infty(M) \otimes \hsA$, define $B_\mu \defeq A_\mu - \JA A_\mu \JA^{-1}$ and $\Phi \defeq \DA + \varphi + \JA \varphi \JA^{-1}$, so that $\DhA[, \omega] = \DM \otimes 1 + \gamma^\mu \otimes B_\mu + \gammaM \otimes \Phi$. Let $\nabla_\mu^S$ be the spin connection on $S$, and consider the vector bundle $E = S \otimes (M \times \hsA)$ such that $L^2(E) = \hshA$, and let $\nabla_\mu^E \defeq \nabla_\mu^S \otimes 1 + 1 \otimes (\partial_\mu + i B_\mu)$ be the natural twisted connection on $E$ defined by the spectral triple, so that $\DhA[, \omega] = -i \gamma^\mu \nabla_\mu^E + \gammaM \otimes \Phi$. Finally, let $D_\mu \defeq \partial_\mu + i \ad(B_\mu)$ and $F_{\mu\nu} \defeq \partial_\mu B_\nu - \partial_\nu B_\mu + i[B_\mu, B_\nu]$. In the same way, we introduce $\omega'$, $A'_\mu$, $\varphi'$, $B'_\mu$, $\Phi'$, $E'$, $D'_\mu$, $F'_{\mu\nu}$ for the algebra $\halgB$.

Let $f_n \defeq \int_0^\infty f(x) x^{n-1} \dd x$ be the moments of $f$ for $n > 0$, then we have the general result \cite[Prop.~8.12]{Suij15a} that we have simplified to take into account the fact that the metric $g$ is Euclidean: \footnote{In particular there is no Einstein-Hilbert Lagrangian since the purely geometric part needs not be compared from $\halgA$ to $\halgB$.}
\begin{proposition}
\label{prop bosonic lagrangian}
Suppose that the $KO$-dimension of $\algA$ is even, then
\begin{align*}
\Tr f(\DhA[,\omega] / \Lambda)
\sim
\int_M \calL(B_\mu, \Phi) \, \dd^4 x + \calO(\Lambda^{-1})
\end{align*}
with
\begin{align*}
\calL(B_\mu, \Phi)
&= \calL_B(B_\mu) + \calL_\varphi(B_\mu, \Phi)
\end{align*}
where $\calL_B(B_\mu) = \frac{f(0)}{24 \pi^2} \tr( F_{\mu\nu} F^{\mu\nu})$, and, up to a boundary term,
\begin{align*}
\calL_\varphi(B_\mu, \Phi) 
&=
- \frac{2 f_2 \Lambda^2}{4 \pi^2} \tr(\Phi^2) 
+ \frac{f(0)}{8 \pi^2} \tr(\Phi^4)
+ \frac{f(0)}{8 \pi^2} \tr\big( (D_\mu \Phi) (D^\mu \Phi) \big)
\end{align*}
\end{proposition}

We use the same function $f$ and the same cut-off $\Lambda$ for the spectral actions on $\halgA$ and $\halgB$.

\medskip
We suppose that $\omega \in \Omega^1_U(\halgA)$ and $\omega' \in \Omega^1_U(\halgB)$ are $\hphi$-compatible in the sense that $\piDhB(\omega')$ and $\piDhA(\omega)$ are $\hphi$-compatible. Since the family of vectors $\{ \gamma^\mu, \gamma_M \}$ is free in the Clifford algebra generated by the $\gamma^\mu$'s, this implies that $A'_\mu$ (resp. $\varphi'$) is $\hphi$-compatible with $A_\mu$ (resp. $\varphi$). The strong $\phi$-compatibility between $\JB$ and $\JA$ then implies that $B'_\mu$ (resp. $\Phi'$) is $\hphi$-compatible with $B_\mu$ (resp. $\Phi$). Notice then that $\partial_\mu \Phi'$ (resp. $\partial_\mu B'_\nu$) is $\hphi$-compatible with $\partial_\mu \Phi$ (resp. (resp. $\partial_\mu B_\nu$). \footnote{Thanks to the fact that $\phiH$ does not depend on the points in $M$.} We then have:

\begin{proposition}
Suppose that $\omega \in \Omega^1_U(\halgA)$ and $\omega' \in \Omega^1_U(\halgB)$ are $\hphi$-compatible in the previous sense. Then
\begin{align*}
\calL_{\halgB, B'}(B'_\mu)
&= \calL_{\halgA, B}(B_\mu) + \text{\TNIC}
\\
\calL_{\halgB, \varphi'}(B'_\mu, \Phi')
&= \calL_{\halgA, \varphi}(B_\mu, \Phi) + \text{\TNIC}
\end{align*}
\end{proposition}

\begin{proof}
From Prop.~\ref{prop bosonic lagrangian}, all the terms in $\calL_{\halgB, B'}(B'_\mu)$ and $\calL_{\halgB, \varphi'}(B'_\mu, \Phi')$ are traces of polynomials of $\hphi$-compatible elements. So, according to Lemma~\ref{lemma scalar product and trace phi compatible AC}, up to terms with non-inherited components, they are equal to the similar expression in terms of traces of polynomials of the corresponding elements on $\halgA$.
\end{proof}

\begin{remark}
This Proposition can be proved without the assumption on the normalization of $\phiH$, see Remark~\ref{remark non normalized}.
\end{remark}

A slight extension of Prop.~\ref{prop unitary equi triple and (st) phi comp} shows that $\omega \in \Omega^1_U(\halgA)$ and $\omega' \defeq \hphi(\omega) \in \Omega^1_U(\halgB)$ are $\hphi$-compatible. But then $\omega'$ contains only inherited degrees of freedom, and so this situation is quite trivial from a physical point of view.

\begin{proposition}
If $\Tpsi'$ is $\hphi$-compatible with $\Tpsi$, then
\begin{align*}
S_{\halgB, f}[\omega', \Tpsi'] 
&= \langle \JhB \Tpsi', \DhB[,\omega'] \Tpsi' \rangle_{\ThshB}
= \langle \JhA \Tpsi, \DhA[,\omega] \Tpsi \rangle_{\ThshA} + \text{\TNIC}
= S_{\halgA, f}[\omega, \Tpsi] + \text{\TNIC}
\end{align*}
\end{proposition}

\begin{proof}
Since $\piDhB(\omega')$ and $\piDhA(\omega)$ are $\hphi$-compatible, $\DhB[,\omega']$ and $\DhA[,\omega]$ are $\hphi$-compatible, and since $\JB$ and $\JA$ are strong $\phi$-compatible, then $\JhB$ and $\JhA$ are strong $\hphi$-compatible.

Using previously defined notations, one can write $\Tpsi' = \tsum_{c, \Tv} \Tpsi'^{c, \Tv} \chi_{c} \otimes f_{\Tv} + \tsum_{c, \hw} \Tpsi'^{c, \Tv} \chi_{c} \otimes f_{\hw}$ and $\Tpsi = \tsum_{c, \Tv} \Tpsi^{c, \Tv} \chi_{c} \otimes e_{\Tv}$. Since $\Tpsi'$ and $\Tpsi$ are $\hphi$-compatible, one has $\Tpsi'^{c, \Tv} =\Tpsi^{c, \Tv}$ for any $c, \Tv$. So, $\langle \JhB \Tpsi', \DhB[,\omega'] \Tpsi' \rangle_{\ThshB} = \tsum_{c_1, c_2, \Tv_1, \Tv_2} \Tpsi'^{c_1, \Tv_1} \Tpsi'^{c_2, \Tv_2} \langle \JhB \chi_{c_1} \otimes f_{\Tv_1} , \DhB[,\omega'] \chi_{c_2} \otimes f_{\Tv_2} \rangle_{\ThshB} = \epsilonhB \tsum_{c_1, c_2, \Tv_1, \Tv_2} \Tpsi'^{c_1, \Tv_1} \Tpsi'^{c_2, \Tv_2} \overline{\langle \chi_{c_1} \otimes f_{\Tv_1} ,  \JhB \DhB[,\omega'] \chi_{c_2} \otimes f_{\Tv_2} \rangle_{\ThshB}}$. From Lemma~\ref{lemma scalar product and trace phi compatible AC}, one has $\langle \chi_{c_1} \otimes f_{\Tv_1} ,  \JhB \DhB[,\omega'] \chi_{c_2} \otimes f_{\Tv_2} \rangle_{\ThshB} = \langle \chi_{c_1} \otimes e_{\Tv_1} ,  \JhA \DhA[,\omega] \chi_{c_2} \otimes e_{\Tv_2} \rangle_{\ThshA} + \text{\TNIC}$ so that, since $\epsilonhB = \epsilonhA$, $\langle \JhB \Tpsi', \DhB[,\omega'] \Tpsi' \rangle_{\ThshB} = \epsilonhA \tsum_{c_1, c_2, \Tv_1, \Tv_2} \Tpsi^{c_1, \Tv_1} \Tpsi^{c_2, \Tv_2} \overline{\langle \chi_{c_1} \otimes f_{\Tv_1} ,  \JhA \DhA[,\omega] \chi_{c_2} \otimes e_{\Tv_2} \rangle_{\ThshA}} + \text{\TNIC} = \langle \JhA \Tpsi ,  \DhA[,\omega] \Tpsi \rangle_{\ThshA} + \text{\TNIC}$.
\end{proof}

\begin{remark}
Notice that the proof of this Proposition requires the assumption on the normalization of $\phiH$.
\end{remark}

\begin{remark}
Notice that the formal proofs presented in the previous Propositions, which compare the spectral action on $\halgB$ to the spectral action on $\halgA$, do not reveal the terms which mix inherited and non-inherited components. A concrete and complete computation is necessary to compare precisely the two Lagrangians. These more concrete computations will be presented elsewhere.
\end{remark}

These results can be collected to construct a sequence $\{ (\halgA_n, \hs_{\halgA_n}, D_{\halgA_n}, J_{\halgA_n}, \gamma_{\halgA_n}) \}_{n \geq 0}$) of even real spectral triples (on AC manifolds) which are $\hphi_{n,n+1}$-compatible and a sequence of their corresponding spectral actions $S_b[\omega_n] + S_f[\omega_n, \Tpsi_n]$ with a control about their inherited and non-inherited terms. Using slight modifications of Prop.~\ref{prop universal forms and Dirac phi compatibility} and \ref{prop uB from uA case AF}, a gauge transformation on $\halgA_n$ can be transported to a gauge transformation on $\halgA_{n+1}$. So, we end up with a sequence of compatible NCGFTs constructed on top of the defining sequence of an $AF$-algebra $\algA = \varinjlim \algA_n$.

\subsection{Some considerations about the limit}
\label{subsec some considerations about the limit}

The question of the “limit” of such a sequence $\{ (\algA_n, \hs_n, D_n) \}_{n \geq 0}$ of $\phi$-compatible (or strong $\phi$-compatible) spectral triples will not be discussed in details in this paper, since, as we will explain, it requires a lot more of analysis concerning in particular the involved operators. Current investigations are in progress concerning these points. In the following, we only outline some results in relation to other works.

By construction, the sequence of algebras $\algA_n$ has a limit $\algA_\infty = \bigcup_{n \geq 0} \algA_n$. When completed, this algebra is the $C^\ast$-algebra $\algA$, and $\algA_\infty$ is a natural sub-algebra of “smooth elements”. Since the morphisms $\phiH[,n,n+1]$ are isometries (thanks to the normalization assumption), the direct limit of the sequence $(\hs_{n}, \phiH[,n,n+1])$ is well-defined. Let $\hs \defeq \varinjlim \hs_{n}$ with $\phiH[,n] : \hs_{n} \to \hs$ the isometries such that $\phiH[,m] \circ \phiH[,n,m] = \phiH[,n]$ for any $n < m$. This direct limit can be constructed explicitly as follows. Let $\hsK_0 \defeq \hs_{0}$ and, for any $n \geq 1$, let $\hsK_n \defeq \hs_{n} \ominus \hs_{n-1}$ where we identify $\hs_{n-1}$ with its range in $\hs_{n}$ via $\phiH[,n-1,n]$, so that $\hs_{n} = \hsK_0 \oplus \cdots \oplus \hsK_n$. Then 
\begin{align*}
\hs 
&= \toplus_{n\geq 0} \hsK_n 
= \big\{ (\xi_n)_{n \geq 0} \mid \xi_n \in \hsK_n \text{ and } \tsum_{n\geq 0} \norm{\xi_n}^2 < \infty \big\}.
\end{align*}
This Hilbert space supports a canonical representation $\pi$ of $\algA$ (see \cite{FlorGhor19p} for instance).

A candidate for a spectral triple as a limit of $\{ (\algA_n, \hs_n, D_n) \}_{n \geq 0}$ has been constructed in \cite{FlorGhor19p} when all the Dirac operators are strongly-$\phi$-compatible. But requiring only $\phi$-compatibility needs to make use of a more subtle way to define the Dirac operator at the limit.

For instance, one could use the approach proposed in \cite{Jana95a} (see references therein) which relies on the following assumption (adapted to our finite dimensional situation): a sequence $\{ L_n \}$ of operators on the Hilbert spaces $\hs_{n}$ converges strongly and uniformly if, for any $\epsilon, \delta > 0$, there is a number $n(\epsilon, \delta)$ such that for any $n(\epsilon, \delta) \leq n < m$ and for any $\psi \in \hs_{n}$ such that $\norm{\psi}_{\hs_{n}} < \delta$, then $\norm{ (L_m \phiH[,n,m] - \phiH[,n,m] L_n) \psi}_{\hs_{m}} < \epsilon$. If the sequence $\{ L_n \}$ satisfies this convergence criteria, then, for any $\psi \in \hs_{n}$, the operator $L$ given by the equality $L \phiH[,n] \psi = \lim_{m \to \infty} \phiH[,m] L_m \phiH[,n,m] \psi$ is well defined.

Notice that the criteria given above is trivially satisfied if the $L_n$'s are strongly-$\phi$-compatible, since then $L_m \phiH[,n,m] \psi =  \phiH[,n,m] L_n \psi$ for any $\psi \in \hs_{n}$, so that, for a sequence of strong $\phi$-compatible operators, the limit always exists. As already mentioned, this is the case for the approach in \cite{FlorGhor19p}. However, the existence of this limit is an extra requirement for a sequence of $\phi$-compatible operators, for instance the sequence of Dirac operators (and the real operators as well as the grading operators). Moreover, showing that it produces a spectral triple is also a question to be investigated. This demands for an analysis that is beyond the scope of this paper. It will be the subject of forthcoming studies.

In \cite{ChriIvan06a}, the authors propose a spectral triple on $AF$-algebras for which the Dirac operator is constructed as follows, using our previous notations: one considers a suitable sequence of positive real numbers $\{ \alpha_i \}_{i \geq 0}$ and $D \defeq \tsum_{i=0}^{\infty} \alpha_i Q_i$ where $Q_i$ is the orthogonal projection on  $\hsK_i \subset \hs$. For any $n \geq 0$, let us define a spectral triple $(\algA_n, \hs_n, D_n)$ using the truncated Dirac operator $D_n \defeq \tsum_{i=0}^{n} \alpha_i Q_i$ on $\hs_{n}$. Then it is easy to verify that the $D_n$'s are strong $\phi$-compatible and the limit of this sequence is $D$.

\smallskip
Parts of this work deal with general structures to construct a sequence of spectral triples on any inductive limit of $C^\ast$-algebras. We expect that this could be relevant for applications in various domains. For instance, in \cite{Lai13a} (see also references therein), taking inspiration from Loop Quantum Gravity (LQG), a sequence of spectral triples on an inductive sequence of  $C^\ast$-algebras generated by loops on an inductive system of finite graphs is considered, using the strong $\phi$-compatibility condition. Considering the limit of this sequence, the authors obtain a candidate for a spectral triple over the space of connections as it is used in LQG. A natural question concerns the relevance of the full potential of our framework (real and/or even spectral triples, spectral actions… ) in this context, in particular, what the $\phi$-compatibility condition could bring to these constructions.

\section{Conclusion}
\label{sec conclusion}

In this paper we have presented a framework to construct sequences of spectral triples on top of inductive sequences defining $AF$-algebras. One main result concerns the structure of the map lifting arrows in Bratteli diagrams to arrows between Krajewski diagrams. We emphasize that the normalization assumption plays a key role. For instance, it can be used to show that the spectral action at each step of the sequence contains the spectral action of the previous step, as well as new terms coupling inherited and new degrees of freedom. Moreover, it permits to get a limit for these sequences, as it appears for the Hilbert spaces. Further investigations are in progress in order to construct and interpret more realistic models in this new framework and to study the possible limits of these sequences.

\section*{Aknowledgements}

We would like to thank the referee for very valuable suggestions to improve this paper.

\bibliography{bibliography}

\end{document}